\documentclass[11pt]{article}
\usepackage{amsthm}
\usepackage{amssymb}
\usepackage{amsmath}
\usepackage{relsize}
\usepackage{bbm}
\usepackage{fullpage}
\usepackage{graphicx}
\usepackage{float}
\usepackage{authblk}
\usepackage{url}
\usepackage[font=footnotesize, labelfont=bf]{caption}

\usepackage{algorithm}
\usepackage{algorithmic}

\newtheorem{thm}{Theorem}
\newtheorem{lem}{Lemma}
\newtheorem{coro}{Corollary}

\newcommand{\argmin}{\operatornamewithlimits{arg\,min}}
\newcommand{\argmax}{\operatornamewithlimits{arg\,max}}
\newcommand{\zap}[1]{}

\newcommand{\Bin}{\mathrm{Bin}}
\newcommand{\PB}{\mathrm{PB}}
\newcommand{\Support}{\mathrm{supp}}
\def\eqd{\,{\buildrel d \over =}\,}

\def\rightarrowd{\,{\buildrel \mathcal L \over \longrightarrow}\,}
\usepackage{color-edits} 
\addauthor{pf}{red}     
\addauthor{bj}{green}   
\addauthor{wh}{blue}    
\addauthor{pr}{cyan}    

\begin{document}
\title{Probabilistic Group Testing under Sum Observations: \\
A Parallelizable 2-Approximation for Entropy Loss}
\author[1]{Weidong Han}
\author[2]{Purnima Rajan}
\author[3]{Peter I. Frazier}
\author[4]{Bruno M. Jedynak}
\affil[1]{Department of Operations Research and Financial Engineering, Princeton University}
\affil[2]{Department of Computer Science, Johns Hopkins University}
\affil[3]{School of Operations Research and Information Engineering, Cornell University}
\affil[4]{Department of Mathematics and Statistics, Portland State University}
\renewcommand\Authands{ and }
\date{\today}
\maketitle


\begin{abstract}
We consider the problem of group testing with sum observations and noiseless answers, in which we aim to locate multiple objects by querying the number of objects in each of a sequence of chosen sets. We study a probabilistic setting with entropy loss, in which we assume a joint Bayesian prior density on the locations of the objects and seek to choose the sets queried to minimize the expected entropy of the Bayesian posterior distribution after a fixed number of questions. 
We present a new non-adaptive policy, called the dyadic policy, show it is optimal among non-adaptive policies, and is within a factor of two of optimal among adaptive policies.  This policy is quick to compute, its nonadaptive nature makes it easy to parallelize, and our bounds show it performs well even when compared with adaptive policies.
We also study an adaptive greedy policy, which maximizes the one-step expected reduction in entropy, and show that it performs at least as well as the dyadic policy, offering greater query efficiency but reduced parallelism.
Numerical experiments demonstrate that both procedures outperform a divide-and-conquer benchmark policy from the literature, called sequential bifurcation, and show how these procedures may be applied in a stylized computer vision problem. 
\end{abstract}

\section{Introduction}

We consider the following set-guessing problem, which is similar to classical group testing \cite{Du2000}, but differs in the form of the observations. Let $\Omega=\mathbb R$ be the real line and $\theta =(\theta_1,\dots,\theta_k)\in \Omega^k$ be a vector containing the unknown locations of $k$ objects, where $k \geq 1$ is known. One can sequentially choose subsets $A_1, A_2, \dots$ of $\Omega$, query the number of objects in each set, and obtain a series of noiseless answers $X_1,X_2,\dots$. Our goal is to devise a method for choosing the questions that allows us to find $\theta$ as accurately as possible, given a finite budget of questions. We work in a Bayesian setting, and use the entropy of the posterior distribution on $\theta$ to measure accuracy.

We consider both adaptive policies, i.e., policies that choose the next question $A_n$ based on previous answers, and non-adaptive policies, i.e., policies that choose all questions in advance.  Adaptive policies promise to better localize the objects within the given query budget, by adapting later questions to provide more useful information, but non-adaptive policies offer easy parallelization because all questions may be asked simultaneously.

In this paper, we present two policies: a new non-adaptive policy, called the dyadic policy, which splits the search domain into successively finer partitions; and an adaptive policy, called the greedy policy, which chooses questions to maximize the one-step expected reduction in entropy. We make the following contributions:

We show that the dyadic policy achieves an information-theoretic lower bound on the expected entropy reduction achievable by a non-adaptive policy, showing it is optimal among non-adaptive policies.  We also show that the dyadic policy's performance is within a factor of two of a lower bound on the entropy reduction under any policy, adaptive or non-adaptive.  Moreover, this non-adaptive policy is easy to compute and provides a posterior distribution that supports fast computation. Specifically, after $N$ questions and answers, the dyadic policy allows for explicitly computing the expected number of objects within each element of a partition of $\Omega^k$ into $2^N$ bins which can be used in a second stage of querying.   
We also further characterize the entropy of the posterior under this policy providing an explicit expression for its expected value and its asymptotic variance, and by showing that it is asymptotically Normally distributed.

We also consider the greedy policy, and show its performance is at least as good as that of the dyadic policy, and in some cases is strictly better.  Thus, this policy offers improved query efficiency, though it does not support parallelization and requires substantially more computation than the dyadic policy, making it the more appropriate choice for applications that do not allow asking questions in parallel, and for which questions are substantially more expensive than computation.

We also compare these policies against benchmarks and show that they offer substantial performance benefits over the previous state-of-the-art.


\subsection{Literature Review}
The previous literature on similar problems can be classified into two groups: those that consider a single object ($k=1$); and those that consider multiple objects ($k\ge 1$).

Among single-object versions of this problem, the earliest
is the R\'enyi-Ulam game \cite{Ulam1976,Renyi,berlekamp64}.  In this game, one person (the responder) thinks of a number between one and one million and another person (the questioner) chooses a sequence of subsets to query in order to find this number.  The responder can answer either YES or NO and is allowed to lie a given number of times.

Variations of the Renyi-Ulam game have been considered in \cite{Marini2005}. Among these variations, the following continuous probabilistic version, first studied in \cite{Pelc1989}, is similar to the problem we consider: The responser thinks of a number $\theta \in[0,1]$ and the questioner aims to find a set $A\subset [0,1]$ with measure less than $\epsilon$ such that $\theta \in A$ with probability at least $q$. In addition, the responser lies with probability no more than $p$. Whether the questioner can win this game based on the error probability $p$ is analyzed and searching algorithms using $O(\log \frac{1}{\epsilon})$ queries are provided.


Among previous work on the single-object problem, perhaps the closest to the current work is \cite{JAP}, which considered a Bayesian setting and used the entropy of the posterior distribution to measure of accuracy, as we do here.  It considered two policies, a greedy policy called probabilistic bisection, which was originally proposed in \cite{Horstein63} and further studied in \cite{castro2009active,WaeberFH13}, and the dyadic policy.  \cite{TsiligkaridisSH13} generalized the probabilistic bisection policy to multiple questioners. Here, we generalize both policies to multiple objects.

Our work contrasts with this previous work on the single-object problem by considering multiple objects.

The previous literature includes work on three multiple-object problems:
the Group Testing problem \cite{Du2000,Stinson2000,Eppstein2007,Harvey2007,Porat2008,li2014group};
the subset-guessing game associated with the Random Chemistry algorithm \cite{Kauffman1996,Buzas2013};
and the Guessing Secret game \cite{Chung2001}. We denote the collection of objects by $S$. 
In the Group Testing problem, questions are of the form: \emph{Is $A\cap S \neq \emptyset$?}
In the subset-guessing game associated with the Random Chemistry algorithm,
questions are of the form \emph{Is $S \subset A$?}.
In the Guessing Secret game, when queried with a set $A$, the responder chooses an element from $S$ according to any rule that he likes, and tells the questioner whether this chosen element is in $A$.  The chosen element itself is not revealed, and may change after each question.  Thus, the answer is $1$ when $S\subset A$, $0$ when $A\cap S=\emptyset$, and can be either $0$ or $1$ otherwise.

Our work contrasts with this previous work by considering a problem where the answer provided by the responser is not binary but instead counts the number of objects in the queried set.

Our use of the (differential) entropy as a measure of quality in localizing objects follows a similar use of entropy in other sequential interrogation problems,
including optimization of continuous functions \cite{villemonteix2009informational},
online learning \cite{russo2014learning},
and adaptive compressed sensing \cite{braun2014info}.
In this literature and here, the differential entropy is of direct interest as a measure of concentration of the posterior probability. Indeed, it is the logarithm of the volume of the smallest set containing ``most of the probability", see \cite{Cover2006} p.246. In our setting, when the prior distribution over each object's location is of Uniform distribution, the posterior distribution is also Uniform and the differential entropy is exactly the logarithm of the volume of the support set of the posterior density.
When the querying process discussed here is followed by a second stage involving a different querying process with different kinds of question and answers (as it is in each of the examples discussed below) the differential entropy may be considered as a surrogate for the time complexity required in this second stage.

\subsection{Applications}


The problem we consider, or slight variants of it, appear in three applications discussed below: 
heavy hitter detection in network traffic stream, 
screening for important input factors in complex simulators,
and fast object localization in Computer Vision.
They also appear in searching for auto-catalytic sets of molecules \cite{Kauffman1996}, and searching for collections of multiple contingencies leading to cascading power failures in models of electrical networks \cite{Eppstein2012}.

In each of the three applications discussed below, objects' locations are discrete rather than continuous.  The policies we present, which result from an analysis considering differential entropy and a continuous prior, may still be used profitably even when objects' locations are known to lie on a finite subset $\Omega'$ of $\Omega$, as long as the granularity of the questions asked does not become finer than $\Omega'$.

In heavy hitter detection \cite{wang2014group}, we operate a router within a computer network, and wish to detect a (presumably small) number of source IP addresses that are generating traffic through our network exceeding a given limit on packet rate.  These source IP addresses are called ``heavy hitters''.
Although we could, in theory, keep an ever-expanding list of all source IP addresses with associated packet counts, this would require a prohibitive amount of memory.
Instead, one can choose a set of IP addresses $A$, and count how many packets fall into that set over a short time period\footnote{Our framework allows general $A$, while in practice, the set $A$ should be of a form that allows easily checking whether a packet resides within it, for example, by having the set $A$ consist of all source IP addresses simultaneously satisfying a collection of conditions on individual bits within the address.  The dyadic policy that we construct has this form when the prior is uniform, and the number of allowed queries is below a threshold.}.  By comparing this number to the limit on packet rate, one can obtain information about the number of heavy hitters (which are our objects $\theta_i$) with source addresses within $A$.
By sequentially, or simultaneously, querying several sets $A$, one can obtain a low-entropy posterior distribution on the locations of all heavy hitters.  One can then follow this first stage of queries on sets $A$ by a second confirmatory stage of queries on individual IP addresses that the first stage revealed were likely to be being heavy hitters.


In screening for important factors in complex simulators \cite{screening_WSC_2012,BettonvilKleijnen1997}, we wish to determine which of a large number of input parameters have a significant effect on the output of a computer model.  A factor model models each input to the simulator as a factor taking one of two values, ``on'' or ``off'', and models the output of the simulator as approximately linear in the factors, with unknown coefficients that multiply each of the ``on'' factors to produce the output.
The ``important'' factors are those with nonzero coefficients, and these are the objects $\theta_i$ we seek to identify.  To identify them, we may choose a set of factors $A$ to turn on, and observe the output of the simulator, which gives us information about the number of important factors in the queried set.  By sequentially, or simultaneously, querying several sets $A$, we may obtain a low-entropy posterior distribution on the identity of the important factors.  We can then individually query those factors believed to be important in a second confirmatory stage.

In computer vision applications, 
we may wish to localize object instances in images and video streams. 
Examples include detecting faces in images~\cite{AliFleHasFua12},
finding quasars in astronomical data~\cite{Mor09}, 
and counting synapses in electron microscopy volumes~\cite{MerRodAloSchDef09}.
To support this, high-performing but computationally expensive classifiers exist that can localize object instances accurately. While one way to localize each instance would be to run such a classifier at each and every pixel to assess whether an instance was centered at that pixel, this would be computationally intractable for large images or video sequences.  Instead, one can divide the image into various sub-regions $A$, and use a computation to count how many instances fall in that region. 
Critically, counting the number of instances in a region is substantially faster than running the classifier at every pixel in that region, see ~\cite{lempitsky2010learning,idrees2013multi,barinova2012detection}.
Using a low-entropy posterior distribution obtained from these queries, one can compute the expected number of objects, among $k$, at each pixel. We can then run our expensive classifier in a second confirmatory stage at those pixels where an object instance has been identified as being likely to reside.
We illustrate our policies on a substantially simplified version of this problem in Section~\ref{sec:dyadic_numerical}.

Now, in Section~\ref{sec:formulation}, we state the problem more formally, and summarize our main results.


\section{Problem Formulation and Summary of Main Results}
\label{sec:formulation}

Let $\theta=(\theta_{1},\dots,\theta_{k})$ be a random vector taking values in $\mathbb R^k$. $\theta_i$ represents the location of the $i$th object of interest, $i=1,\dots,k$.  We assume that $\theta_1,\dots,\theta_k$ are i.i.d. with density $f_0$, and joint density $p_0(\theta)=\prod_{i=1}^k f_0(\theta_i)$.  We assume $f_0$ is absolutely continuous with respect to the Lebesgue measure and has finite differential entropy, which is defined in \eqref{eq:diff_entropy}. We refer to $p_0$ as the Bayesian prior probability distribution on $\theta$.  We will ask a series of $N>0$ questions to locate $\theta_{1},\dots,\theta_{k}$, where each question takes the form of a subset of $\mathbb R$, and the answer to this question is the number of objects in this subset.  More precisely, for each $n\in\{1,2,\ldots,N\}$, the $n^{th}$ question is $A_n\subset\mathbb R$ and its answer is
\begin{equation}
\label{eq:answer}
X_n=\mathbbm 1_{A_n}(\theta_1) +\dots+ \mathbbm 1_{A_n}(\theta_k),
\end{equation}
where $\mathbbm 1_A$ is the indicator function of the set $A$.
Unless otherwise stated, our choice of the set $A_n$ may depend upon the answers to all previous questions, and upon some initial randomization through a uniform random variable $Z$ on $[0,1]$ chosen independently of $\theta$.
Thus, the set $A_n$ is random, through its dependence on $Z$, and the answers to previous questions.

We call a rule for choosing the questions $A_n$ a {\it policy}.  
Formally, we define a policy $\pi$ to be a sequence $\pi=(\pi_1,\ldots,\pi_N)$, where $\pi_n$ is a Borel-measurable subset of $[0,1] \times \{0,1,\ldots,k\}^{n-1} \times \mathbb R$. We denote the collection of all such policies by $\Pi$. 
With a policy $\pi$ specified, the choice of $A_n$ is then $A_n = \left\{ t \in \mathbb R : (Z,X_{1:n-1},t) \in \pi_n\right\}$, so that specifying $\pi_n$ implicitly specifies a rule for choosing $A_n$ based on the random seed $Z$ and the history $X_{1:n-1}$.
Here, we have used the notation $X_{a:b}$ for any natural numbers $a$ and $b$ to indicate the sequence $(X_a,\ldots,X_b)$ if $a\le b$, and the empty sequence if $a>b$.  We define $\theta_{a:b}$ and $A_{a:b}$ similarly.
The distribution of $A_n$ thus implicitly depends on $\pi$.  When we wish to highlight this dependence, we will use the notation $P^\pi$ and $E^\pi$ to indicate probability and expectation respectively.  However, when the policy being studied is clear, we will simply use $P$ and $E$ .

This definition of $\Pi$ allows the choice of question to depend upon previous answers, and when we wish to emphasize this fact we will refer to $\Pi$ as the set of {\it adaptive} policies.
We also define the set of {\it non-adaptive policies} $\Pi_N \subset \Pi$ to be those under which each $A_n$ depends only on the random seed $Z$, i.e., for a fixed $Z=z$, the questions $A_{1:N}$ are deterministic. 
From a formal point of view, note that the set of adaptive policies includes the set of non-adaptive policies as a special case.
Figure~\ref{fig:Bayesian network} illustrates, as a Bayesian network, the dependence structure of the random variables in our problem under an adaptive policy, and under a non-adaptive policy.  

\begin{figure}[H]
\centering
\minipage{0.4\textwidth}
  \includegraphics[width=\linewidth]{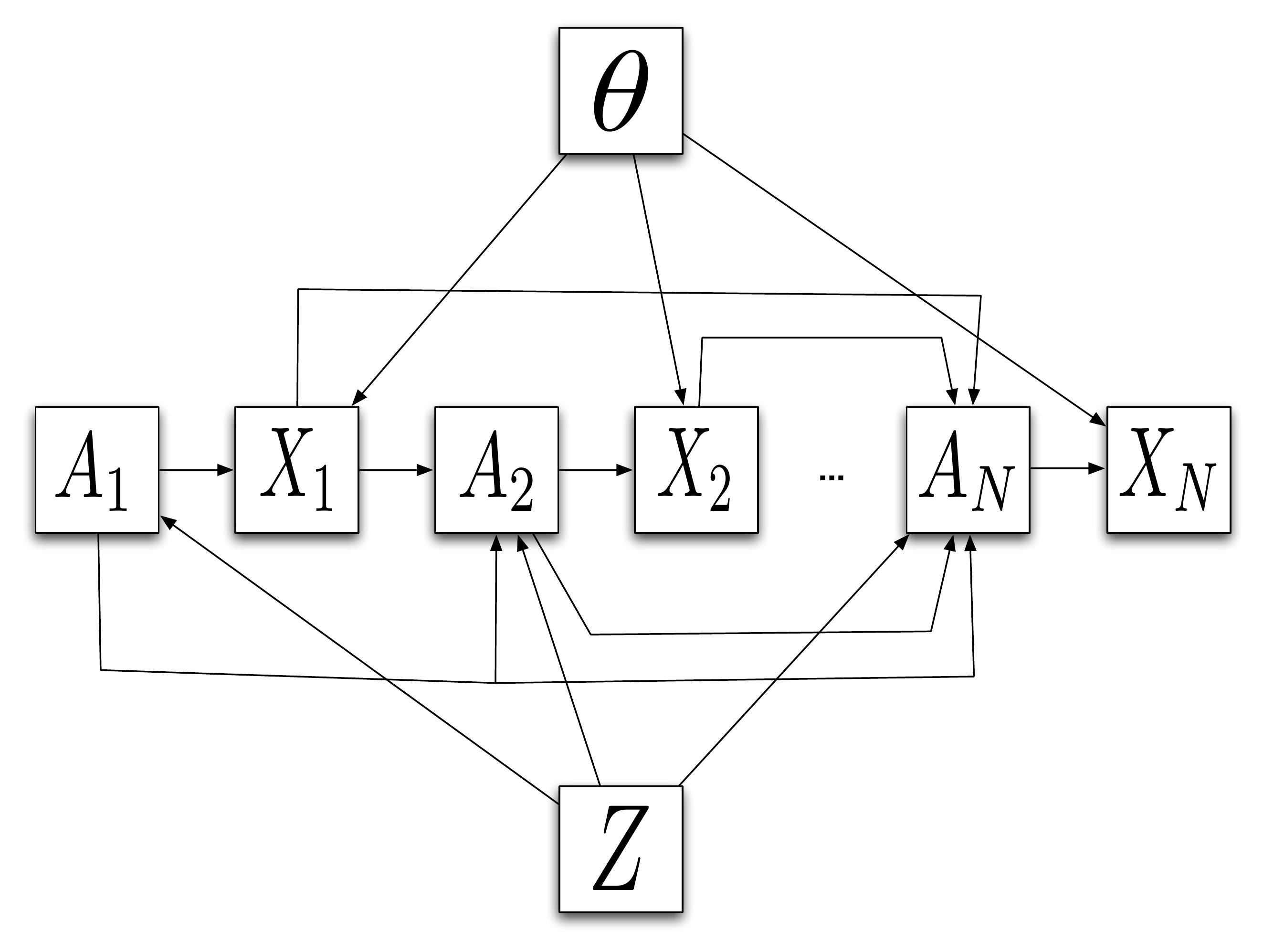}
\endminipage\hfill
\minipage{0.4\textwidth}
  \includegraphics[width=\linewidth]{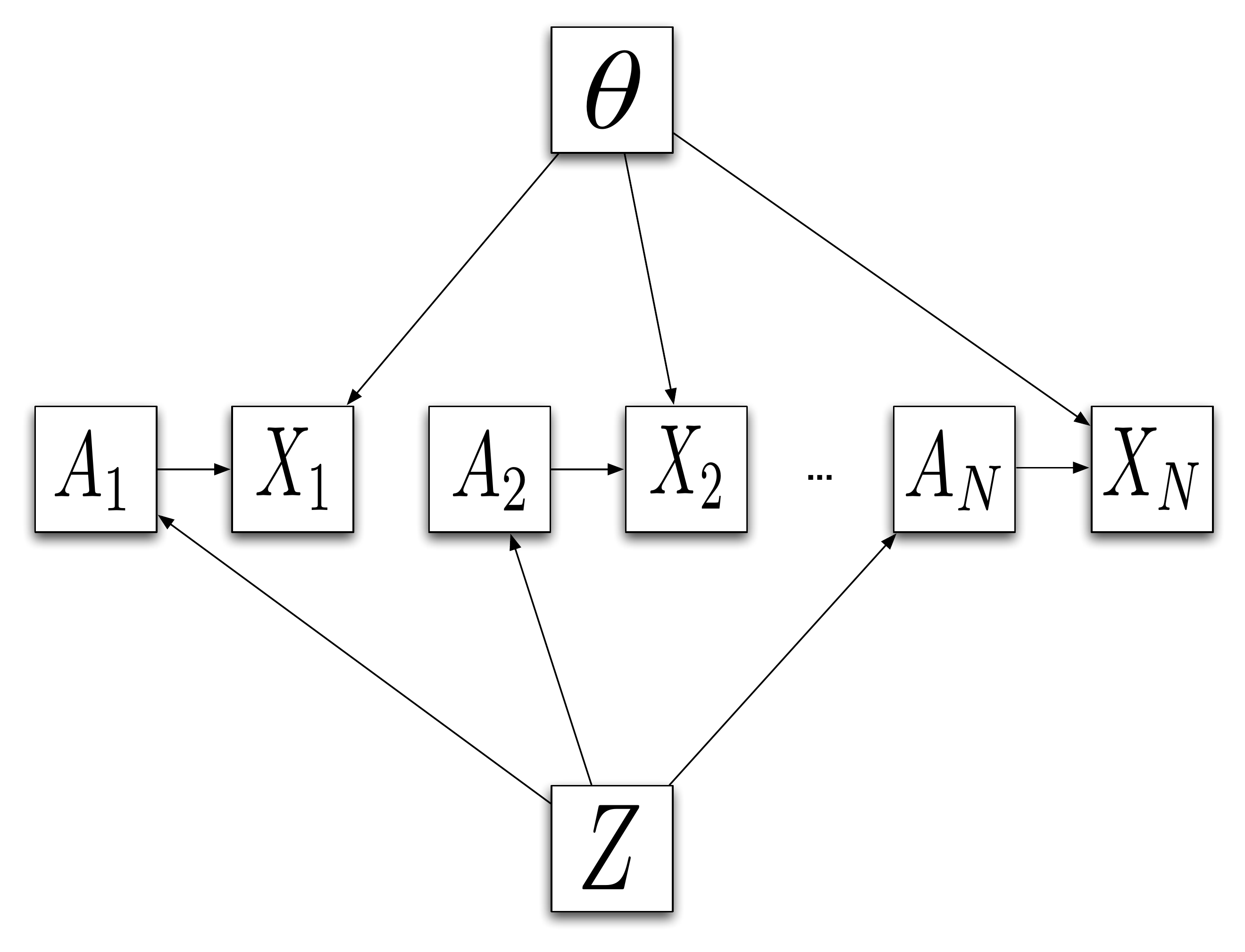}
\endminipage\hfill
\caption{Bayesian network representation of our model under an adaptive policy (left), and a non-adaptive policy (right).}
\label{fig:Bayesian network}
\end{figure}

We refer to the posterior probability distribution on $\theta$ after $n$ questions as $p_n$, so $p_n$ is the conditional distribution of $\theta$ given the past history $B_n=\{Z,A_{1:n},X_{1:n}\}$. The dependence on $Z$ arises because $A_n$ may depend on $Z$, in addition to $X_{1:n-1}$.  
Equivalently, under any fixed policy $\pi$, $p_n$ is the conditional distribution of $\theta$ given $B_n$.
This posterior $p_n$ can be computed using Bayes rule:
$p_n(u)$ is proportional to $p_0(u)$ over the set
$\left\{u \in \mathbb R^k : X_m = \sum_{i=1}^k \mathbbm 1_{A_m}(u_i),\  1\le m\le n\right\}$,
and $0$ outside.


After we exhaust our budget of $N$ questions, we will measure the quality of what we have learned via the differential entropy $H(p_N)$ of the final posterior distribution $p_N$ on $\theta$,
\begin{equation}
\label{eq:diff_entropy}
H(p_N) =-E[\log p_N]=- \int_{\mathbb R^k} p_N(u_{1:k}) \log(p_N(u_{1:k}))\, du_{1:k}.
\end{equation}

Throughout this paper, we use $``\log"$ to denote the logarithm to base 2. We let $H_0 = H(p_0)$, and we assume $-\infty<H(p_0)<+\infty$.
The posterior distribution $p_N$, as well as its entropy $H(p_N)$, are random for $N>0$, as they depend on $X_{1:N}$ and $Z$. Thus, we measure the quality of a policy $\pi\in\Pi$ when given $N$ questions using the \emph{rate of reduction in expected entropy}
\begin{equation}
\label{eq:ROR}
  R(\pi,N) = \frac{H_0-E^\pi[H(p_N)]}{N}.
\end{equation}
This rate is the average number of bits learned per question.

Our goal in this paper is to characterize the solution to the optimization problem
\begin{equation}
  \label{eq:optimal-policy}
\sup_\pi R(\pi,N).
\end{equation}
with $\pi \in \Pi$ or $\pi \in \Pi_N$. Any policy that attains this supremum is called {\it optimal}. According to this definition, an optimal policy may not exist.

While \eqref{eq:optimal-policy} can be formulated as a partially observable Markov decision process \cite{FrWeorDp}, and can be solved, in principle, via dynamic programming, the state space of this dynamic program is the space of posterior distributions over $\theta$, and the extreme size of this space prevents solving this dynamic program through brute-force computation.
Thus, we must characterize optimal policies using other means.

We define two policies, the dyadic policy, which is non-adaptive, and the greedy policy, which is adaptive.  (More precisely, the greedy is a {\it class} of policies, as its definitions allow certain decisions to be made arbitrarily.) We will see below that the dyadic policy attains the supremum in \eqref{eq:optimal-policy} for $\pi\in\Pi_N$, and thus is optimal among non-adaptive policies.  We will also see that its performance comes within a factor of two of the supremum for $\pi\in\Pi$, showing that it is a two-approximation among adaptive policies.  We will also see below that the greedy policy performs at least as well as the dyadic policy, and so is also a two-approximation among adaptive policies.


To define the {\it dyadic policy}, let us recall that the quantile function of $\theta_1$ is
\begin{equation}
\label{eq:quantile}
Q(p) = \inf \left\{u \in \mathbb R: p \leq F_0(u)\right\},
\end{equation}
where $F_0$ is the cumulative distribution function of $\theta_1$, corresponding to its density $f_0$. The dyadic policy consists in choosing at step $n\geq 1$ the set
\begin{equation}
\label{eq:dyadic_brief}
A_n = \left(\bigcup_{j=1}^{2^{n-1}}\left(Q\left(\frac{2j-1}{2^n}\right),Q\left(\frac{2j}{2^n}\right)\right]\right)\bigcap \,\Support(f_0),
\end{equation}
where $\Support(f_0)$ is the support of $f_0$, i.e., the set of values $u \in \mathbb R$ for which $f_0(u)>0$.
For example, when $f_0$ is uniform over $(0,1]$, the dyadic policy is the one in which
  the first question is $A_1=\left(\frac12,1\right]$,
  the second question is $A_2=\left(\frac14,\frac12\right]\cup\left(\frac34,1\right]$, 
and each subsequent question is obtained by subdividing $(0,1]$ into $2^{n}$ equally sized subsets, and including every second subset.  A further illustration of the dyadic question sets $A_n$ is provided in Figure \ref{fig:dyadic} in Section \ref{sec:dyadic}.
This definition of the dyadic policy generalizes a definition provided in \cite{JAP} for single objects.

We define a {\it greedy policy} to be any policy that chooses each of its questions
to minimize the expected entropy of the posterior distribution one step forward in time,
\begin{equation}
  \label{eq:greedy}
A_n \in \argmin_A E[H(p_n)|p_{n-1},A_n=A], \text{for all $n=1,2,\dots,N$,}
\end{equation}
where the argmin is taken over all Borel-measurable subsets of $\mathbb R$.
We show in Section~\ref{sec:greedy} that this argmin exists.

We are now ready to present our main results:
\begin{equation}
\label{eq:main}
\log(k+1) \ge \sup_{\pi\in\Pi} R(\pi,N)
\ge R(\pi_G,N)
\ge R(\pi_D,N)
=H_k=\sup_{\pi\in\Pi_N} R(\pi,N)\geq \frac{1}{2}\log(k+1),
\end{equation}
where $\pi_G$ is a greedy policy, $\pi_D$ is a dyadic policy, and 
\begin{equation}
H_k=H\left(\Bin\left(k,\frac{1}{2}\right)\right),\end{equation}
is the entropy of a Binomial distribution $\Bin(k,\frac{1}{2})$.  

The first inequality in \eqref{eq:main} is an information theoretic inequality (easily proved in Section \ref{sec:bound}). The second  inequality is trivial since an optimal adaptive policy is at least as good as any other policy. The third inequality comes from a detailed computation of the posterior distribution $p_N$ of $\theta$ after observing $N$ answers for any possible sequence of $N$ questions (see Section~\ref{sec:greedy_value}).  Additionally, we show that this inequality cannot be reversed, by presenting a special case in which there is a greedy policy whose performance is strictly better than that of the dyadic policy (see Section \ref{sec:comparison}).  The first equality comes from the characterization of the posterior distribution $p_N$ in the special case of the dyadic policy (see Section~\ref{sec:dyadic_value}). The last equality is an information theoretic inequality which exploits the conditional independence structure of non-adaptive policies. It is proven in Section \ref{sec:bound}. The last inequality, proven in Section \ref{sec:dyadic_value}, shows that the rate of an optimal non adaptive policy is no less than half the rate of an optimal adaptive policy.

The power of these results is illustrated by Figure~\ref{fig:numQs},
which shows, as a function of the number of objects $k$, the number of questions required to reduce the expected entropy of the posterior on their locations by 20 bits per object.
The figure shows the number of questions needed under the dyadic policy (solid line);
under two benchmark policies described below, Benchmark 1 and Benchmark 2 (dotted, and dash-dotted lines); and a lower bound on the number needed under the optimal adaptive policy (dashed line, and left-most expression in \eqref{eq:main}).
By \eqref{eq:main}, we know that the number of extra questions required by using either the dyadic or the greedy, instead of the adaptive optimal policy, is bounded above by the distance between the solid and dashed lines.

\begin{figure}[tb]
\centering
\includegraphics[scale=0.45]{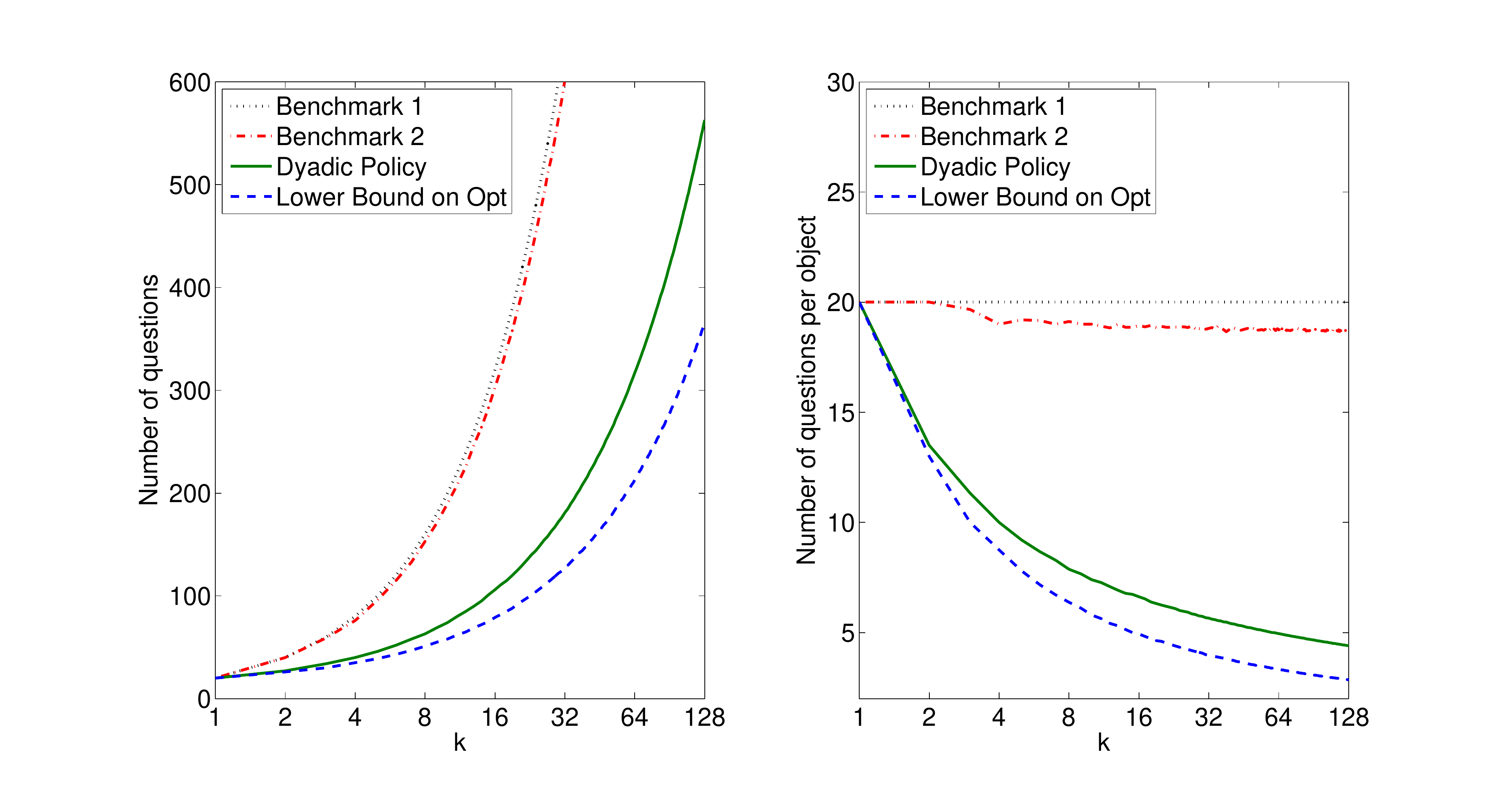}
\caption{Number of questions needed to reduce the entropy by 20 bits per object under two benchmark policies and the dyadic policy, and a lower bound on the number under the optimal adaptive policy. The two graphs show the total number of questions (left) and number of questions per object (right). The dyadic policy significantly outperforms both benchmarks and its performance is relatively close to the lower bound on an optimal adaptive policy's performance.  The performance of the greedy policy is between that of the dyadic and the lower bound.}
\label{fig:numQs}
\end{figure}

Benchmark 1 identifies each object individually, using an optimal single-object strategy.  It first asks questions to localize the first object $\theta_1$, reducing the entropy of our posterior distribution on that object's location by 20 bits. This requires 20 questions, and can be achieved, for example, by a bisection policy, \cite{Horstein63}. It then uses the same strategy to localize each subsequent second object, requiring 20 questions per object. Implementing such a policy would require the ability to ask questions about whether or not a single specified object (e.g., object $\theta_1$) resides in a queried set, rather than the number of objects in that set.  While this ability is not included in our formal model, Benchmark 1 nevertheless provides a useful comparison. The total number of questions required under this policy to achieve 20 bits of entropy reduction per object is $20k$.

Benchmark 2 is adapted from the sequential bifurcation policy of \cite{BettonvilKleijnen1997}.  While \cite{BettonvilKleijnen1997} considered an application setting somewhat different from the problem that we consider here (screening for discrete event simulation), we were able to modify their policy to allow it to be used in our setting.  A detailed description of the modified policy is provided in Appendix~\ref{sec:bifurcation}.  It makes full use of the ability to ask questions about multiple objects simultaneously, and improves slightly over Benchmark 1.  We view this policy as the best previously proposed policy from the literature for solving the problem that we consider.

The figure shows that a substantial saving over both benchmarks is possible through the dyadic or greedy policy.  For example, for $k=2^4=16$ objects, Benchmark 1 and Benchmark 2 require 320 and 304 questions respectively.  In contrast, the dyadic policy requires 106 questions, which is nearly 3 times smaller than required by the benchmarks. Furthermore, \eqref{eq:main} shows that the greedy policy performs at least as well as the dyadic policy.  Thus, localizing objects' locations jointly can be much more efficient than localizing them one-at-a-time, and the dyadic and greedy policies are implementable policies that can achieve much of the potential efficiency gains.  

The figure also shows, again at $k=2^4=16$ objects, that the optimal policy requires at least 80 questions, while the dyadic and greedy require no more than 106 questions, and so are within a factor of 1.325 of optimal.
 This is remarkable, when we
compare how little is lost when going from the hard-to-compute optimal policy to the easily computed dyadic policy, with how much is gained by going to the dyadic from one of the two benchmark policies considered.
Our results also show that this multiplicative factor is never worse than 2.

The dyadic policy can be computed extremely quickly, and can even be
pre-computed, as the questions asked do not depend on the answers to previous
questions.  This makes it convenient in settings
where multiple questions can be asked simultaneously, e.g., in a parallel or
distributed computing environment.
The greedy policy requires more computational effort than the dyadic policy, but is still substantially easier to compute than the optimal policy, and provides performance at least as good as that of the dyadic policy, as shown by \eqref{eq:main}, and sometimes strictly better, as will be shown in Section~\ref{sec:comparison}.


We see in the figure that the dyadic policy's rate and the rate of the optimal policy come together at $k=1$.  This can also be seen directly from our theoretical results.  When $k=1$, the left-hand and right-hand sides of \eqref{eq:main} are equal, since $\Bin\left(k,\frac{1}{2}\right)$ becomes a $\mathrm{Bernoulli}(\frac12)$ random variable, whose entropy is $\log(2)=1$.  This shows, when $k=1$, that the rate of expected entropy reduction under the dyadic is the same as the upper bound on this rate under the optimal policy, which in turn shows that both dyadic and greedy policies are optimal, and the upper bound is tight. When $k=1$, the well-known bisection policy is a greedy policy, and the dyadic is also greedy, i.e., satisfies \eqref{eq:greedy}.

We begin our analysis in Section~\ref{sec:bound}, by justifying the left-most inequality in \eqref{eq:main}.
We then provide an explicit expression for the posterior distribution in Section~\ref{sec:posterior}, which is used in later analysis.  We analyze the dyadic policy in Section~\ref{sec:dyadic}, and the greedy policy in Section~\ref{sec:greedy}.  
We illustrate the use of our policies on a stylized problem inspired by computer vision applications in Section~\ref{sec:dyadic_numerical}. Finally, we offer concluding remarks in Section~\ref{sec:conclusion}.

\section{Upper Bounds on the Rate of Reduction in Expected Entropy}
\label{sec:bound}

In this section, below in Theorem~\ref{t:thm1}, we prove the first inequality in \eqref{eq:main}, which is an easy upper bound on the reduction in expected entropy for a fixed number of questions and answers under an adaptive policy.  This bound is obtained from the fact that the answer to each question is a number in $\{0,1,\dots,k\}$, and so cannot provide more than $\log(k +1)$ bits. We also prove in Theorem~\ref{t:thm1} that the upper bound cannot be achieved for $k>1$.  

Then, we provide a complementary upper bound for non-adaptive policies in Theorem \ref{thm:nonadaptive_opt}, which we later show in Section~\ref{sec:dyadic_value} is matched by the dyadic policy, showing that it is optimal among non-adaptive policies.




\begin{thm}
\label{t:thm1}
\begin{equation}
\label{eq:thm1}
\sup_{\pi\in\Pi} R(\pi,N) \le \log(k+1).
\end{equation}
Moreover, when $k>1$, this inequality is strict.
\end{thm}
\begin{proof}
According to the definition of rate of reduction in expected entropy in \eqref{eq:ROR}, in order to prove \eqref{eq:thm1}, we need to prove that under any valid policy,
\begin{equation}
\label{eq:thm1_0}
E[H(p_N)] \geq H_0- \log(k+1)N.
\end{equation}
Recall that $H(p_N)$ is the entropy of the posterior distribution of $\theta$, which is random through its dependence on the past history $B_N=\{X_{1:N},A_{1:N},Z\}$. Thus, $E[H(p_N)]=H(\theta|B_N)$. Furthermore, using information theoretic arguments, we have
\begin{equation}\label{eq:thm1_1}
H(\theta|B_N)=H(\theta)-I(\theta;B_N)=H_0-(H(B_N)-H(B_N|\theta))
\end{equation}
Moreover,
\begin{equation}\label{eq:thm1_2}\begin{split}
H(B_N)&=H(X_{1:N},A_{1:N},Z)\\
&=H(A_{1:N}|X_{1:N},Z)+H(X_{1:N}|Z)+H(Z)\\
&=H(X_{1:N}|Z)+H(Z)\\
&\le \sum_{n=1}^N H(X_n) + H(Z)\\
&\le \log(k+1)N + H(Z),
\end{split}\end{equation}
where $H(A_{1:N}|X_{1:N},Z)=0$ because the information contained in the random seed $Z$ and the answers $X_{1:N}$ completely determines the questions $A_{1:N}$. Recall that for all $n=1,2,\dots,N$, $X_n$ is a discrete random variable with $k+1$ possible outcomes, namely $0,1,\dots,k$. The maximum possible value for the entropy $H(X_n)$ is $\log(k+1)$, obtained when each outcome of $X_n$ has the same probability $\frac{1}{k+1}$, i.e. $H(X_n)\leq \log(k+1)$.

On the other hand,
\begin{equation}\label{eq:thm1_3}\begin{split}
H(B_N|\theta)&=H(X_{1:N},A_{1:N},Z|\theta)\\
&=H(A_{1:N}|X_{1:N},Z,\theta)+H(X_{1:N}|Z,\theta)+H(Z|\theta)\\
&=H(Z),
\end{split}\end{equation}
where $H(A_{1:N}|X_{1:N},Z,\theta)=0$ for the same reason as above, and $H(X_{1:N}|Z,\theta)=0$ because the information contained in $\theta$ completely determines $X_{1:N}$. Also, $H(Z|\theta)=H(Z)$ because the random seed $Z$ is assumed to be independent of the objects $\theta$.

Plugging \eqref{eq:thm1_2} and \eqref{eq:thm1_3} back into \eqref{eq:thm1_1}, we obtain the desired result \eqref{eq:thm1_0}.

We now prove that the inequality \eqref{eq:thm1} is strict when $k>1$, i.e. when there is more than one object. Consider any fixed $Z=z$, which specifies the questions set $A_1$. Recall from \eqref{eq:answer} that $X_1=\mathbbm 1_{A_1}(\theta_1)+\dots+\mathbbm 1_{A_1}(\theta_k)$ and that $\theta_1,\dots,\theta_k$ are independent. As a consequence, $X_1 \mid Z=z\sim \Bin(k,p)$, where $p=\int_{A_1}f_0(u)\,du$. Therefore, $H(X_1| Z=z)=H\left(\Bin(k,p)\right)<\log(k+1)$ when $k>1$, implying $H(X_1|Z)<\log(k+1)$. Thus, $H(B_N)=H(X_{1:N}|Z)+H(Z)<\log(k+1)N + H(Z)$, so that there is no policy that can achieve the upper bound.
\end{proof}

Now, we provide an upper bound on the rate of reduction in expected entropy for all non-adaptive policies.

\begin{thm}
\label{thm:nonadaptive_opt}
Under any non-adaptive policy $\pi\in\Pi_N$, we have 
\begin{equation}\label{eq:nonadaptive_opt}
R(\pi,N) \le H\left(\Bin\left(k,\frac 1 2\right)\right).
\end{equation}
\end{thm}
\begin{proof}
To prove the claim \eqref{eq:nonadaptive_opt}, it suffices to prove that under any non-adaptive policy,
\begin{equation}
I(\theta; B_N))\le H\left(\Bin\left(k,\frac 1 2\right)\right)N.
\end{equation}
First of all, the relation between mutual information and entropy gives
\begin{equation}
\label{eq:dyadic_opt_0}
I(\theta;B_N)=I(\theta; (A_{1:N},X_{1:N},Z))=H(A_{1:N},X_{1:N},Z)-H(A_{1:N},X_{1:N},Z|\theta).
\end{equation}
For the first term, we have
\begin{equation}
\label{eq:dyadic_opt_1}
H(A_{1:N},X_{1:N},Z)\le \sum_{n=1}^N H(A_n,X_n|Z)+H(Z) =\sum_{n=1}^N(H(X_n|A_n,Z)+H(A_n|Z))+H(Z).
\end{equation}
For the second term, we have
\begin{equation}
H(A_{1:N},X_{1:N},Z|\theta)=H(X_{1:N}|A_{1:N},Z,\theta)+H(A_{1:N}|Z,\theta)+H(Z|\theta).
\end{equation}
Furthermore, $H(X_{1:N}|A_{1:N},Z,\theta)=0$ since the information contained in $\theta$ and $A_{1:N}$ completely determines $X_{1:N}$. Also, $H(A_{1:N}|Z,\theta)=\sum_{n=1}^N H(A_n|Z,\theta)=\sum_{n=1}^N H(A_n|Z)$, since we can see from Figure \ref{fig:Bayesian network} that $A_1,\dots,A_N$ are conditional independent given $Z,\theta$, and each $A_n$ is independent of $\theta$ conditional on $Z$ as $Z$ is the only parent of $A_n$ in the directed acyclic graph. In addition, $H(Z|\theta)=H(Z)$ since the random seed $Z$ is assumed to be independent of the object $\theta$. Hence, we have
\begin{equation}
\label{eq:dyadic_opt_2}
H(A_{1:N},X_{1:N},Z|\theta)=\sum_{n=1}^N H(A_n|Z)+H(Z).
\end{equation}
Combining \eqref{eq:dyadic_opt_0}, \eqref{eq:dyadic_opt_1} and \eqref{eq:dyadic_opt_2} yields
\begin{equation}
\begin{split}
I(\theta; (A_{1:N},X_{1:N},Z))&\le\sum_{j=1}^N(H(X_n|A_n,Z)+H(A_n|Z))+H(Z)-\sum_{n=1}^N H(A_n|Z)-H(Z)\\
&=\sum_{n=1}^N H(X_n|A_n,Z).
\end{split}
\end{equation}
Recall that by definition, $X_n=\sum_{i=1}^k \mathbbm 1_{A_n}(\theta_i)$, which is a sum of i.i.d. Bernoulli random variables. Hence, for each fixed $A_n=a_n$ and $Z=z$, we have $(X_n|A_n=a_n,Z=z)\sim\Bin(k,P(\theta_1\in a_n))$. Therefore,
\begin{equation}
H(X_n|A_n,Z)\le \sup_{a_n,z} H(X_n|A_n=a_n,Z=z)\le \sup_{p\in [0,1]} H(\Bin(k,p))=H\left(\Bin\left(k,\frac 12\right)\right).
\end{equation}
The claim of the theorem follows.
\end{proof}

\section{Explicit Characterization of the Posterior Distribution}
\label{sec:posterior}

In this section, we first introduce in Section \ref{sec:post_object} some notation to characterize the joint location of objects and provide an example to illustrate these notations. We then derive an explicit formula for the posterior distribution on the locations of the objects.
In Section \ref{sec:post_answer}, we compute the conditional distribution of the next answer $X_n$ given previous answers $X_{1:n-1}$, which we will use later to analyze the rate of a policy.

\subsection{The Posterior Distribution of the Objects}
\label{sec:post_object}
Consider a fixed $n$, where $1 \leq n \leq N$.
For each binary sequence of length $n$, $s=\{s_1,\ldots,s_n\}$, let
\begin{equation}
\label{eq:Cs}
C_s = \left(\bigcap_{1\le j\le n;s_j=1}A_j\right) \bigcap \left(\bigcap_{1\le j\le n;s_j=0} A_j^c\right)\bigcap \Support(f_0).
\end{equation}

The collection $\{C_s:C_s\neq\emptyset, s \in \{0,1\}^n\}$ is a partition of the support of $f_0$.
A history of $n$ questions provides information on which sets $C_s$ contain which objects among $\theta_{1:k}$.

We will think of a sequence of binary sequences $s^{(1)},\dots,s^{(k)}$
as a sequence of codewords indicating the sets in which each of the objects $\theta_{1:k}$ reside,
i.e, indicating that $\theta_1$ is in $C_{s^{(1)}}$, $\theta_2$ is in $C_{s^{(2)}}$, etc.
We may consider each binary sequence $s^{(1)},\dots,s^{(k)}$ to be a column vector,
and place them into an $n\times k$ binary matrix, $\mathcal S$.  This binary matrix then codes the location of all $k$ objects, and is a codeword for their joint location.

Moreover, to characterize the location of the random vector $\theta=(\theta_{1:k})$ in terms of its codeword ${\mathcal S}$, define $C_{\mathcal S}\subset \mathbb R^k$ to be the Cartesian product
\begin{equation}
C_{\mathcal S}=C_{s^{(1)}}\times\dots\times C_{s^{(k)}}.
\end{equation}

To be consistent with an answer $X_j$, we must have exactly $X_j$ objects located in the question set $A_j$ for each $1\le j\le n$.  This can be described in terms of a constraint on the matrix $\mathcal S$ as $s_j^{(1)}+\dots+s_j^{(k)}=X_j$, i.e., that the sum of the $j^{th}$ row in the matrix $\mathcal S$ is $X_j$.
Thus, after observing the answers to the questions $X_{1:n}=x_{1:n}$,
the set of all possible joint codewords describing $\theta_{1:k}$ is
\begin{equation}
\label{eq:collection}
E_n=\{\mathcal S|s^{(1)},\dots,s^{(k)}\in \{0,1\}^n, C_{s^{(1)}},\dots,C_{s^{(k)}}\neq\emptyset, s_j^{(1)}+\dots+s_j^{(k)}=x_j, \text{for all $ 1\leq j\leq n$}\}.
\end{equation}

To illustrate the previous construction, and also to provide the foundation for a later analysis in Section~\ref{sec:comparison} showing the greedy policy is strictly better than the dyadic policy in some settings, we provide two examples of the posterior distribution, arising from two different responses to the same sequence of questions.

Suppose $\theta_1, \theta_2$ are two objects located in (0,1] with a uniform prior distribution $f_0$. Let $A_1$ and $A_2$ be the first two questions of the dyadic policy, so $A_1 = \left(\frac{1}{2},1\right]$ and $A_2 = \left(\frac{1}{4},\frac{1}{2}\right] \cup \left(\frac{3}{4},1\right]$.  Then consider two possibilities for the answers to these questions:

\paragraph{Example 1:} Suppose $X_1=0$ and $X_2=2$. According to \eqref{eq:collection}, there is only one matrix $\mathcal S$ in the collection $E_2$, which has $s^{(1)}=s^{(2)}=(0,1)^T$.  Thus $E_2=\{\mathcal S_1\}$ where
\begin{align}
\label{eq:list_matrix1}
\mathcal S_1=
\begin{pmatrix}
0 & 0\\
1 & 1
\end{pmatrix}.
\end{align}
We can observe that
$p_2(u_{1:2})=16$ when $u_{1:2}$ is in $\left(\frac{1}{4},\frac{1}{2}\right]\times\left(\frac{1}{4},\frac{1}{2}\right]$,
and $0$ otherwise.

\paragraph{Example 2:} Suppose $X_1=1$ and $X_2=1$. According to \eqref{eq:collection}, there are four matrices in the collection $E_2=\{\mathcal S_1, \mathcal S_2, \mathcal S_3, \mathcal S_4\}$,
\begin{align}
\label{eq:list_matrix2}
\mathcal S_1=
\begin{pmatrix}
0 & 1\\
0 & 1
\end{pmatrix},
\mathcal S_2=
\begin{pmatrix}
0 & 1\\
1 & 0
\end{pmatrix},
\mathcal S_3=
\begin{pmatrix}
1 & 0\\
0 & 1
\end{pmatrix},
\mathcal S_4=
\begin{pmatrix}
1 & 0\\
1 & 0
\end{pmatrix}.
\end{align}
We can observe that the posterior distribution has density $p_2(u_{1:2})=4$ when
$u_{1:2}$ is in
$\left(0,\frac{1}{4}\right]\times \left(\frac{3}{4},1\right]$ or
$\left(\frac{1}{4},\frac{1}{2}\right]\times \left(\frac{1}{2},\frac{3}{4}\right]$ or
$\left(\frac{1}{2},\frac{3}{4}\right]\times \left(\frac{1}{4},\frac{1}{2}\right]$ or
$\left(\frac{3}{4},1\right]\times \left(0,\frac{1}{4}\right]$,
and is $0$ otherwise.

\medskip
All possible joint locations of $\theta_1, \theta_2$ in the two examples above are shown in Figure \ref{fig:eg_matrix}.

\begin{figure}[H]
\centering
\includegraphics[scale=0.55]{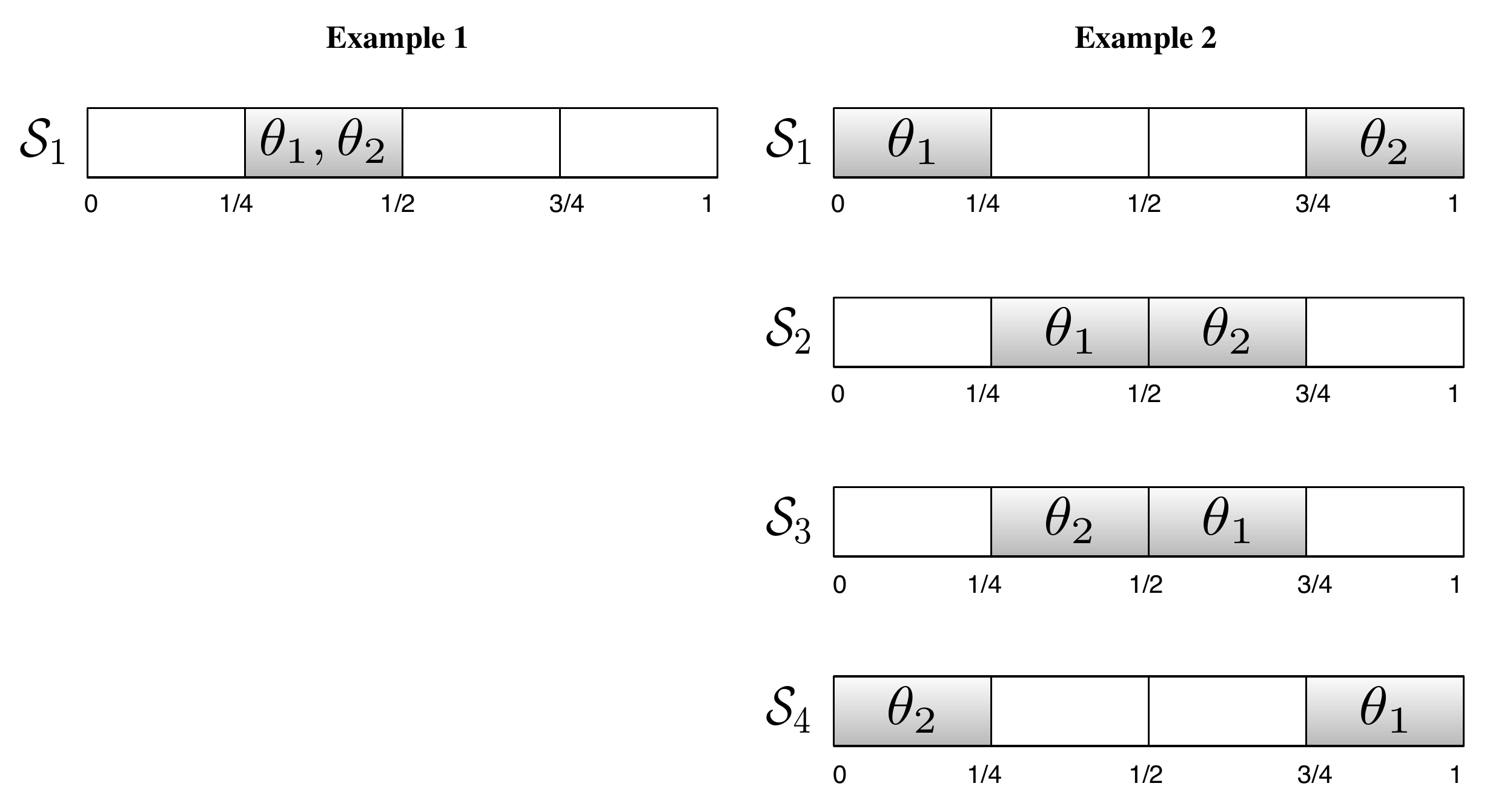}
\caption{Illustration of the locations of the two objects $\theta_1, \theta_2$ specified by each matrix given in \eqref{eq:list_matrix1} and \eqref{eq:list_matrix2}. The dark subsets mark the location of the objects $\theta_1, \theta_2$.}
\label{fig:eg_matrix}
\end{figure}

Given this notation, we observe the following lemma:
\begin {lem}
\label{lem:historyk}
Let a policy $\pi$ and the random seed $Z=z$ be fixed. Then, for each $x_{1:n}$,
the event $\{X_{1:n}=x_{1:n}\}$ can be rewritten
\begin{equation}
\label{eq:historyk}
\{X_{1:n}=x_{1:n}\}
=\left\{\theta \in \bigcup\limits_{\mathcal S\in E_n} C_{\mathcal{S}}\right\},
\end{equation}
where we recall that $E_n$ depends on $x_{1:n}$.
Moreover for any $\mathcal S, \mathcal T\in E_n$ with $\mathcal S\not= \mathcal T$, the two sets $C_{\mathcal S}$ and $C_{\mathcal T}$ are disjoint.
\end{lem}
\begin{proof}
  Clearly, according to the definition of $E_n$ in \eqref{eq:collection}, when $\theta \in \bigcup_{\mathcal S\in E_n}C_{\mathcal{S}}$, the answers that we observe must satisfy $X_{1:n}=x_{1:n}$. On the other hand, suppose $\theta_{1:k}\not\in\bigcup_{\mathcal S\in E_n}C_{\mathcal{S}}$.  Then $\theta_{1:k}$ belongs to some nonempty set $C_\mathcal S$ where $\mathcal S\not\in E_n$. Hence, there exists $j$, $1\le j\le n$, such that $s_j^{(1)}+\dots+s_j^{(k)}\neq x_j$, which implies that the answer to the question $A_j$ is $X_j=s_j^{(1)}+\dots+s_j^{(k)}\neq x_j$. This proves \eqref{eq:historyk}.

  Now, for any $\mathcal S\not= \mathcal T$, there exists $i$ with $1\leq i\leq k$ such that $s^{(i)} \not= t^{(i)}$. This implies that $C_{s^{(i)}}$ and $C_{t^{(i)}}$ are disjoint and the last assertion follows.
\end{proof}

At this point, the explicit characterization of the posterior distribution is immediate and we have the following lemma.
\begin{lem}
\label{lem:productk}
\begin{equation}
\label{eq:productk}
p_n(u_{1:k}) = \frac{p_0(u_{1:k})}{p_0\left(\bigcup\limits_{\mathcal S \in E_n} C_{\mathcal S}\right)} \text{, for $u_{1:k} \in \bigcup\limits_{\mathcal S \in E_n} C_{\mathcal S}$},
\end{equation}
and $p_n(u_{1:k}) = 0$ for $u_{1:k} \notin \bigcup\limits_{\mathcal S \in E_n} C_{\mathcal S}$.
Here, for any measurable set $A$, $p_0(A)$ denotes the integral $\int_A p_0(u_{1:k})\,du_{1:k}$.
Moreover,
\begin{equation}
p_0\left(\bigcup\limits_{\mathcal S \in E_n} C_{\mathcal S}\right)=\sum\limits_{\mathcal S\in E_n} p_0(C_{\mathcal S})= \sum\limits_{\mathcal S\in E_n} f_0(C_{s^{(1)}})\dots f_0(C_{s^{(k)}}),
\end{equation}
where $f_0(C_{s^{(i)}})$ denotes the integral $\int_{C_{s^{(i)}}} f_0(u)\,du$.
\end{lem}

\subsection{The Posterior Predictive Distribution of $X_{n+1}$}
\label{sec:post_answer}
We now provide an explicit form for the posterior predictive distribution of $X_{n+1}$, i.e., its conditional distribution given the history $X_{1:n}$ and the external source of randomness in the policy $Z$.  This is useful because Lemma \ref{lem:Eentropy_onestep} in the appendix shows that the expected entropy $E[H(p_N)]$ can be computed using the conditional entropy of $X_{n+1}$ given $B_n=(Z,A_{1:n},X_{1:n})$. We use this in Sections~\ref{sec:dyadic_value} and~\ref{sec:greedy_value} to compute the expected entropy for the dyadic and greedy policies respectively.

For $n=0$, we have demonstrated in the proof of Theorem~\ref{t:thm1} that $X_1$ follows the binomial distribution $\Bin(k,f_0(A_1))$ given $Z$.

Now, consider any $n\in\{1,2,\dots,N-1\}$, and any fixed history $b_n=(z,a_{1:n},x_{1:n})$.
Using the equality \eqref{eq:historyk} presented in Lemma \ref{lem:historyk} we have,
\begin{equation}
\begin{split}
&P(X_{n+1}=x|B_n=b_n) \\
&= \sum\limits_{\mathcal S \in E_n} P(X_{n+1}=x,\theta \in C_{\mathcal S}|B_n=b_n) \\
&= \sum\limits_{\mathcal S \in E_n} P(X_{n+1}=x|\theta \in C_{\mathcal S},B_n=b_n)P(\theta \in C_{\mathcal S}|B_n=b_n).
\end{split}
\label{eq:predictive1}
\end{equation}
Now, since for any $\mathcal S \in E_n$, $\{\theta \in C_{\mathcal S},Z=z\} \subset \{B_n=b_n\}$ according to Lemma \ref{lem:historyk}, we can simplify:
\begin{equation}
  P(X_{n+1}=x|\theta \in C_{\mathcal S},B_n=b_n) =P(X_{n+1}=x|\theta \in C_{\mathcal S}, Z=z).
\end{equation}

Also, using Lemma \ref{lem:productk}, we obtain
\begin{equation}
  \label{eq:predictive2}
  P(\theta \in C_{\mathcal S}|B_n=b_n) = \frac{f_0(C_{s^{(1)}})\dots f_0(C_{s^{(k)}})}{\sum\limits_{\mathcal S \in E_n}f_0(C_{s^{(1)}})\dots f_0(C_{s^{(k)}})}.
\end{equation}
Finally, according to \eqref{eq:answer}, $X_{n+1}$ is the sum of $k$ Bernoulli random variables $\mathbbm 1_{A_{n+1}}(\theta_1),\dots,\mathbbm1_{A_{n+1}}(\theta_k)$. Given the event $\{\theta \in C_{\mathcal S},Z=z\}$, these $k$ Bernouili r.v's are conditionally independent with respective parameters $q_1=\frac{f_0(A_{n+1} \cap C_{s^{(1)}})}{f_0(C_{s^{(1)}})},\dots, q_k=\frac{f_0(A_{n+1} \cap C_{s^{(k)}})}{f_0(C_{s^{(k)}})}$. This conditional independence can be verified as follows.
Consider any fixed binary vector $w\in\{0,1\}^k$.  For each $i=1,\ldots,k$, let $D_i$ be equal to $A_{n+1}$ if $w_i=1$ and its complement $A_{n+1}^c$ if $w_i=0$.  Then,
\begin{equation}
\begin{split}
  &P(\mathbbm 1_{A_{n+1}}(\theta_i) = w_i,\  i=1,\ldots,k | \theta \in C_{\mathcal S}, Z=z)
  = P(\theta_i \in D_i,\ i=1,\ldots, k |\theta \in C_{\mathcal S}, Z=z) \\
  &=\frac{p_0(D_1\cap C_{s^{(1)}} \times \cdots \times D_k \cap C_{s^{(k)}})}{p_0(C_\mathcal S)}
  =\frac{\prod_{i=1}^k f_0(D_i\cap C_{s^{(i)}})}{\prod_{i=1}^k f_0(C_{s^{(i)}}))}
  =\prod_{i=1}^k \frac{f_0(D_i\cap C_{s^{(i)}})}{f_0(C_{s^{(i)}}))} \\
  &=\prod_{i=1}^k P(\theta_i \in D_i |\theta \in C_{\mathcal S}, Z=z)
  =\prod_{i=1}^k P(\mathbbm 1_{A_{n+1}}(\theta_i) = w_i | \theta \in C_{\mathcal S}, Z=z).
\end{split}
\end{equation}

Using the fact that $X_{n+1}$ is the sum of $k$ conditionally independent Bernoulli random variables given $\theta \in C_\mathcal S$ and $Z=z$, we may provide an explicit probability mass function. When $q_1=\dots =q_k$, $X_{n+1}$ is conditionally $\Bin(k,q_1)$ given $\theta \in C_\mathcal S$ and $Z=z$. In general, let $W_1,\ldots,W_n$ be $n$ independent discrete random variables with $W_i \sim \mathrm{Bernoulli}(q_i)$, where $q_1,\dots,q_n$ are any real numbers in [0,1]. The distribution of $Y=W_1+\ldots+W_n$ is called \emph{Poisson Binomial} distribution, which was first studied by S. D. Poisson in \cite{Poisson1837}. We denote the distribution of $Y$ by $\PB(q_1,\dots,q_n)$ and its probability mass function $P(Y=y) = f_{\PB}(y;q_1,\ldots,q_n)$ is given by
\begin{equation}
  \label{eq:PB_def}
  f_{\PB}(y;q_1,\ldots,q_n)=\sum_{w_{1:n} \in \{\{0,1\}^n|w_1+\dots+w_n=y\}} \prod_{j=1}^n q_j^{w_j}(1-q_j)^{1-w_j},\end{equation}
and has mean and variance given by
\begin{equation}
\begin{split}
E[Y] &=q_1+\ldots+q_n,\\
Var[Y] &=q_1(1-q_1)+\ldots+q_n(1-q_n).
\end{split}
\end{equation}

Using this definition of the Poisson Binomial distribution, the conditional distribution of $X_{n+1}$ given $\theta \in C_\mathcal S$ and $Z=z$ is $\PB(q_1,\dots,q_n)$.

Finally, putting together equations \eqref{eq:predictive1}, \eqref{eq:predictive2},
and the fact that $X_{n+1}$ is conditionally $\PB(q_1,\ldots,q_n)$ given $\theta \in C_\mathcal S$ and $Z=z$
provides the following characterization of the conditional probability mass function of $X_{n+1}$ given $B_n=(Z,A_{1:n},X_{1:n})=b_n$.
\begin{thm}
\label{thm:postY}
For $n=0$, given $\{B_0=b_0\}=\{Z=z\}$, $X_1 \sim \Bin(k,f_0(A_1))$. For $n=1,2,\dots, N-1$, given $B_n=(Z,A_{1:n},X_{1:n})=b_n$, $X_{n+1}$ is a mixture of Poisson Binomial distributions with probability mass function:
\begin{equation}
\label{eq:postY}
\begin{array}{cl}
&P(X_{n+1}=x | B_n=b_n) \\
=& \mathlarger{\sum_{\mathcal S \in E_n}} \frac{f_0(C_{s^{(1)}})\dots f_0(C_{s^{(k)}})}{\sum\limits_{\mathcal T \in E_n}f_0(C_{t^{(1)}})\dots f_0(C_{t^{(k)}})}f_{\PB}\left(x,q_1=\frac{f_0(A_{n+1} \cap C_{s^{(1)}})}{f_0(C_{s^{(1)}})},\dots, q_k=\frac{f_0(A_{n+1} \cap C_{s^{(k)}})}{f_0(C_{s^{(k)}})}\right).
\end{array}
\end{equation}
\end{thm}

\section{The Dyadic Policy for Localizing Multiple Objects}
\label{sec:dyadic}

We now present the first policy of interest: the \emph{dyadic policy}. This policy is easy to implement, and is non-adaptive, allowing its use in parallel computing environments.  The description of the dyadic policy will be given in Section \ref{sec:dyadic_description}. In Section \ref{sec:dyadic_value}, we will derive the rate of this policy and show that it is optimal among all non-adaptive policies, which is the last equality in our main results \eqref{eq:main}. Finally, asymptotic normality of $H(p_N)$ under the dyadic policy will be provided in Section \ref{sec:dyadic_conv}.
\subsection{Description of the dyadic policy}
\label{sec:dyadic_description}
The definition of the dyadic policy is given in \eqref{eq:dyadic_brief}. In this section, we provide an iterative construction of this policy, introducing notation which will be useful later on.

First, we partition the support of $f_0$ into two subsets, $A_{1,0}$ and $A_{1,1}$:
\begin{subequations}
\begin{align}
A_{1,0}&= \left(Q\left(0\right),Q\left(\frac{1}{2}\right)\right] \cap \Support(f_0), \\
A_{1,1}&= \left(Q\left(\frac{1}{2}\right),Q\left(1\right)\right] \cap \Support(f_0),
\end{align}
\end{subequations}
where $Q$, as defined in \eqref{eq:quantile}, denotes the quantile function. With this partition, the question asked at time 1 is
\begin{equation}
A_1=A_{1,1}.
\end{equation}
Then we adopt a similar procedure recursively for each $n=1,\dots, N-1$ to partition $A_{n,j}$ into two subsets, $A_{n+1,2j}$ and $A_{n+1,2j+1}$ and then construct the question from these partitions. For $j=0,\dots,2^n-1$, define
\begin{subequations}
\begin{align}
A_{n+1,2j}&=\left(Q\left(\frac{2j}{2^{n+1}}\right),Q\left(\frac{2j+1}{2^{n+1}}\right)\right] \cap \Support(f_0),\\
A_{n+1,2j+1}&=\left(Q\left(\frac{2j+1}{2^{n+1}}\right),Q\left(\frac{2j+2}{2^{n+1}}\right)\right] \cap \Support(f_0),
\end{align}
\end{subequations}
Then the question asked at time $n+1$ is
\begin{equation}
A_{n+1}=\bigcup_{j=0}^{2^n-1}A_{n+1,2j+1}.
\end{equation}
An illustration of these sets $A_n$ is provided below in Figure \ref{fig:dyadic}.

Note that the dyadic policy is non-adaptive, as only the prior distribution is used to construct the next set and not the answer to previous questions.

\begin{figure}[H]
\centering
\includegraphics[width=0.8\textwidth,height=201pt]{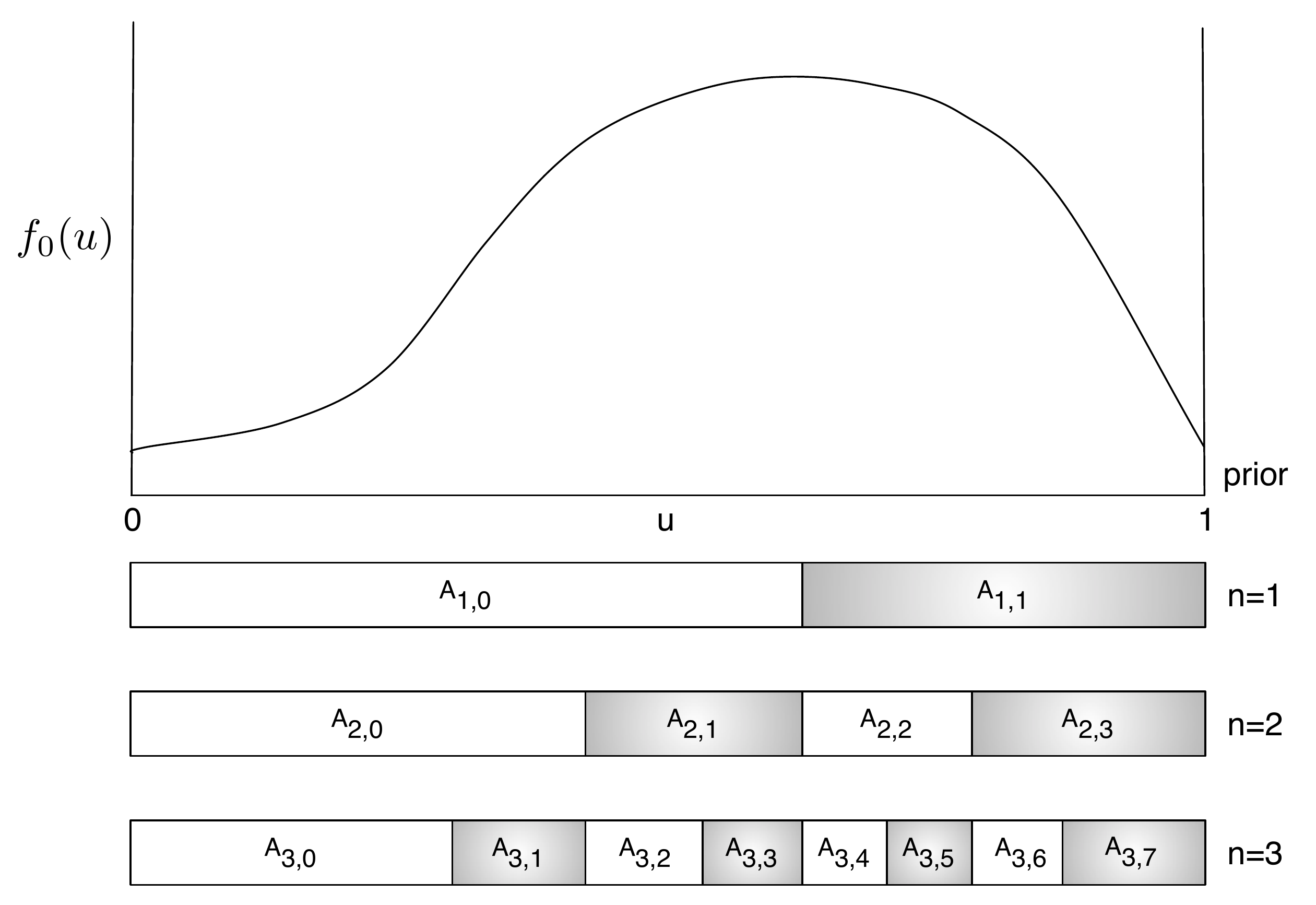}
\caption{Illustration of the dyadic policy. The prior density with support $[0,1]$ is displayed above the illustrations of the sets $A_{n,k}$ for $n=1,2,3$. The question set $A_n$ is the union of the dark subsets $A_{n,k}$ for that value of $n$.}
\label{fig:dyadic}
\end{figure}

\subsection{The rate of the dyadic policy}
\label{sec:dyadic_value}

The rate of the dyadic policy is stated as follows:
\begin{thm}
\label{thm:dyadic}
Under the dyadic policy $\pi_D$,
\begin{equation}
R(\pi_D,N) = H\left(\Bin\left(k,\frac{1}{2}\right)\right).
\end{equation}
Moreover, the dyadic policy is optimal among all non-adaptive policies.
\end{thm}
\begin{proof}
In this proof, we will first simplify the equation \eqref{eq:postY} in Theorem \ref{thm:postY} to obtain the posterior distribution of $X_{n+1}$ under the dyadic policy. Then we will calculate the entropy $H^{\pi_D}(X_{n+1}|B_n)$ and employ Lemma \ref{lem:Eentropy_onestep} in the appendix to compute the rate of the dyadic policy.

At time $n$, where $1 \leq n\leq N$, the support of $f_0$ is partitioned into pairwise disjoint subsets $\{A_{n,0},\dots,A_{n,2^{n}-1}\}$. Recall the definition of $C_s$ in \eqref{eq:Cs}. The sets $C_s$ provide a bijection which maps a binary sequence $s \in \{0,1\}^n$ to a subset $A_{n,j(s)}$ for some $j(s) \in \{0,1,\dots,2^n-1\}$.
Hence, $C_{s^{(i)}}$ in \eqref{eq:postY} can be rewritten as
\begin{equation}
\label{eq:Cs_An}
C_{s^{(i)}}=A_{n,j(s^{(i)})}, \text{ for some index } j(s^{(i)}) \in \{0,1,\dots,2^{n}-1\}.
\end{equation}

According to the construction of dyadic questions in Section \ref{sec:dyadic_description}, $A_{n+1}=\bigcup\limits_{j=0}^{2^n-1}A_{n+1,2j+1}$. Moreover, $A_{n+1,2j(s^{(i)})+1} \subset A_{n,j(s^{(i)})}$ and $A_{n+1,2j+1} \cap A_{n,j(s^{(i)})} = \emptyset$, for all $j \neq j(s^{(i)})$. Thus, by \eqref{eq:Cs_An} we have
\begin{equation}
A_{n+1}\cap C_{s^{(i)}} = A_{n+1,2j(s^{(i)})+1}.
\end{equation}

Combining the above result with the fact that $f_0(A_{n+1,2j(s^{(i)})+1})=\frac{1}{2}f_0(A_{n,j(s^{(i)})})$ yields
\begin{equation}
\frac{f_0(A_{n+1} \cap C_{s^{(i)}})}{f_0(C_{s^{(i)}})}=\frac{1}{2},
\end{equation}
and this is true for all $i=1,2,\dots,k$.

Thus, for $n\ge 1$, we can simplify \eqref{eq:postY} in Theorem \ref{thm:postY} as
\begin{equation}
\begin{array}{cl}
& p_n(X_{n+1}=x|B_n=b_n)\\
=& \mathlarger{\sum_{\mathcal S \in E_n}} \frac{f_0(C_{s^{(1)}})\dots f_0(C_{s^{(k)}})}{\sum\limits_{\mathcal T \in E_n}f_0(C_{t^{(1)}})\dots f_0(C_{t^{(k)}})}f_{\PB}\left(x,q_1=\frac{f_0(A_{n+1} \cap C_{s^{(1)}})}{f_0(C_{s^{(1)}})},\dots, q_k=\frac{f_0(A_{n+1} \cap C_{s^{(k)}})}{f_0(C_{s^{(k)}})}\right)\\
=& \mathlarger{\sum_{\mathcal S \in E_n}} \frac{f_0(C_{s^{(1)}})\dots f_0(C_{s^{(k)}})}{\sum\limits_{\mathcal T \in E_{n+1}} f_0(C_{t^{(1)}})\dots f_0(C_{t^{(k)}})}f_{\PB}\left(x,q_1=\frac{1}{2},\dots,q_k=\frac{1}{2}\right)\\
=& f_{\PB}\left(x,q_1=\frac{1}{2},\dots,q_k=\frac{1}{2}\right)\frac{\sum\limits_{\mathcal S \in E_n} f_0(C_{s^{(1)}})\dots f_0(C_{s^{(k)}})}{\sum\limits_{\mathcal T \in E_n} f_0(C_{t^{(1)}})\dots f_0(C_{t^{(k)}})}\\
=& f_{\PB}\left(x,q_1=\frac{1}{2},\dots,q_k=\frac{1}{2}\right).
\end{array}
\end{equation}

The density above is just the density of the binomial distribution $\Bin\left(k,\frac{1}{2}\right)$. We proved that given $\{B_n=b_n\}$, $X_{n+1}$ is distributed as $\Bin\left(k,\frac{1}{2}\right)$ and $H^{\pi_D}(X_{n+1}|B_n=b_n)=H\left(\Bin\left(k,\frac{1}{2}\right)\right)$ for all $n=1,\dots,N-1$. Thus, taking the expectation over all possible realizations of $B_n$, we obtain
\begin{equation}
H^{\pi_D}(X_{n+1}|B_n) = H\left(\Bin\left(k,\frac{1}{2}\right)\right).
\end{equation}
Since $f_0(A_1)=\frac{1}{2}$ under the dyadic policy, according to Theorem \ref{thm:postY}, $X_1|Z=z$ is distributed as $\Bin \left(k,\frac{1}{2}\right)$ for any fixed $z$ and $H^{\pi_D}(X_1|B_0)=H\left(\Bin\left(k,\frac{1}{2}\right)\right)$ as well.

Therefore, according to \eqref{eq:Eentropy_all} in Lemma \ref{lem:Eentropy_onestep} in the appendix,
\begin{equation}
R(\pi_D,N)
= \frac{H_0-E^{\pi_D}[H(p_N)]}{N}
= \frac{\sum_{n=0}^{N-1}H^{\pi_D}(X_{n+1}|B_n)}{N}
= H\left(\Bin\left(k,\frac{1}{2}\right)\right).
\end{equation}
The last claim follows from Theorem \ref{thm:nonadaptive_opt} immediately.
\end{proof}
Note that this is the last equality in our main result \eqref{eq:main}.

The theorem above implies the following approximation guarantee for the entropy reduction under the dyadic policy, relative to optimal.
\begin{coro}
\label{coro:ratio of rate}
\begin{equation*}
\frac
{R(\pi_D,N)}
{\sup_{\pi\in\Pi} R(\pi,N)}
\ge \frac{1}{2}.
\end{equation*}
\end{coro}
\begin{proof}
According to Theorem \ref{t:thm1} and Theorem \ref{thm:dyadic}, it suffices to show
\begin{equation}
            \frac{H(\Bin(k,\frac12))}{\log(k+1)} \ge \frac12.
\end{equation}

First, note that $H(\Bin(k,\frac12)) = H(\sum_{i=1}^k B_i)$, where $B_i$ are iid $\mathrm{Bernoulli}(\frac12)$.
Using Theorem~1 in \cite{harremoes2003entropy}, (but expressing entropy in base 2 instead of base $e$),
\begin{equation}
    2^{2H(\Bin(k,\frac12))} \ge k2^{2H(B_1)} = 4k.
\end{equation}

This implies that $H(\Bin(k,\frac12)) \ge \frac12\log(4k)$ and 
\begin{equation}
       \frac{H(\Bin(k,\frac12))}{\log(k+1)} \ge \frac12 \frac{\log(4k)}{\log(k+1)} \ge \frac12.
\end{equation}
\end{proof}
This shows that the dyadic policy is a $2$-approximation policy, i.e. the rate of learning under the dyadic policy is at least one-half of the rate under an optimal policy.

\subsection{Convergence in entropy under the dyadic policy}
\label{sec:dyadic_conv}

In real applications, however, we are concerned not only about the expected entropy $E^{\pi_D}[H(p_N)]$ but also about the actual entropy $H(p_N)$ that we obtain in a specific trial. It would be beneficial if the actual entropy did not deviate too much from its expected value.  It turns out to be the case for the dyadic policy under the assumptions that the prior density $f_0$ is bounded from above. Lemma \ref{lem:dyadic_Hp} in the appendix provides a decomposition formula for the actual entropy $H(p_n)$ into a sum of two terms. The first term is a sum of i.i.d. random variables. The second term is a converging martingale as will be shown in Lemma \ref{lem:martingale_conv} in the appendix. Finally, Theorem \ref{thm:dyadic_conv} provides almost sure convergence and asymptotic normality for $H(p_n)$ as a direct consequence of Lemma \ref{lem:dyadic_Hp} and \ref{lem:martingale_conv}.

\begin{thm}
\label{thm:dyadic_conv}
Assume there exists $M>0$ such that $f_0(u)\leq M$ for all $u\in \mathbb R$. Then under the dyadic policy,
\begin{equation}
\label{eq:coro1}
\lim_{N \rightarrow\infty} \frac{H(p_N)}{N} = -H\left(\Bin\left(k,\frac{1}{2}\right)\right) \text{ almost surely},
\end{equation}
and
\begin{equation}
\label{eq:coro2}
\lim_{N \rightarrow\infty} \frac{H(p_N)+NH\left(\Bin\left(k,\frac{1}{2}\right)\right)}{\sqrt{N}} \eqd N(0, \sigma^2),
\end{equation}
where $\sigma^2$ is the variance of the random variable $\log{k \choose X}$ with $X \sim \Bin\left(k,\frac{1}{2}\right)$.
\end{thm}

\begin{proof}
According to Lemma \ref{lem:martingale_conv}, $\lim_{N \rightarrow\infty} \frac{I_2(N)}{N}=\lim_{N \rightarrow\infty} \frac{I_2(\infty)}{N}=0$ almost surely. Hence, by \eqref{eq:Hn} in Lemma \ref{lem:dyadic_Hp},
\begin{equation}
\lim_{N \rightarrow\infty} \frac{H(p_N)}{N}=\lim_{N \rightarrow\infty} \frac{I_2(N)}{N}-\frac{1}{N}\sum_{j=1}^N Z_j=0-E[Z_1]=-H\left(\Bin\left(k,\frac{1}{2}\right)\right)
\end{equation}
almost surely.

To prove \eqref{eq:coro2}, note that
\begin{equation}
\frac{H(p_N)+NH\left(\Bin\left(k,\frac{1}{2}\right)\right)}{\sqrt{N}}=\frac{I_2(N)-\sum_{i=1}^NZ_j+NH\left(\Bin\left(k,\frac{1}{2}\right)\right)}{\sqrt{N}}.\end{equation}

Furthermore, since by Lemma \ref{lem:martingale_conv} $I_2(N)$ converges to $I_2(\infty)$ almost surely and $E[|I_2(\infty)|]<\infty$, $\frac{I_2(N)}{\sqrt{N}}\rightarrow 0$ almost surely, which implies $\frac{I_2(N)}{\sqrt{N}}\rightarrowd 0$. On the other hand, $E(Z_j)=H\left(\Bin\left(k,\frac{1}{2}\right)\right)$ and $Var(Z_j)=Var\left(\log {k \choose X}\right)=\sigma^2$, where $X\sim\Bin\left(k,\frac{1}{2}\right)$. Hence, by the central limit theorem, we have $\frac{-\sum_{i=1}^NZ_j+NH\left(\Bin\left(k,\frac{1}{2}\right)\right)}{\sqrt{N}}\rightarrowd \mathcal N(0,\sigma^2)$. Therefore, by Slutsky's Theorem (Theorem 25.4 in \cite{Billingsley}),
\begin{equation}
\frac{H(p_N)+NH\left(\Bin\left(k,\frac{1}{2}\right)\right)}{\sqrt{N}}=\frac{I_2(N)}{\sqrt{N}}+\frac{-\sum_{i=1}^NZ_j+NH\left(\Bin\left(k,\frac{1}{2}\right)\right)}{\sqrt{N}}\rightarrowd \mathcal N(0,\sigma^2).
\end{equation}

\end{proof}

Figure \ref{fig:dyadic3objects} below shows the simulation results for localizing one object, two objects, and three objects respectively, under the dyadic policy. We assume the prior density $f_0$ is uniform over $(0,1]$ and ask $100$ questions to locate the objects. The top line corresponds to locating a single object. In this case, the dyadic policy is actually optimal and identical to the greedy policy as was proved in \cite{JAP}. Moreover, the entropy process $H(p_n)$ is in this case deterministic.  The middle and bottom lines show the results for respectively $k=2$ and $k=3$ objects. In this case, the entropy process $H(p_n)$ is not deterministic anymore. The entropy reduction per question which is visualized in the second column  is asymptotically equal to $H\left(Bin\left(k,\frac{1}{2}\right)\right)$ according to the law of large numbers. The third column illustrates the asymptotic normality of the entropy process for the dyadic policy.

\begin{figure}[H]
\minipage{0.33\textwidth}
  \includegraphics[width=\linewidth]{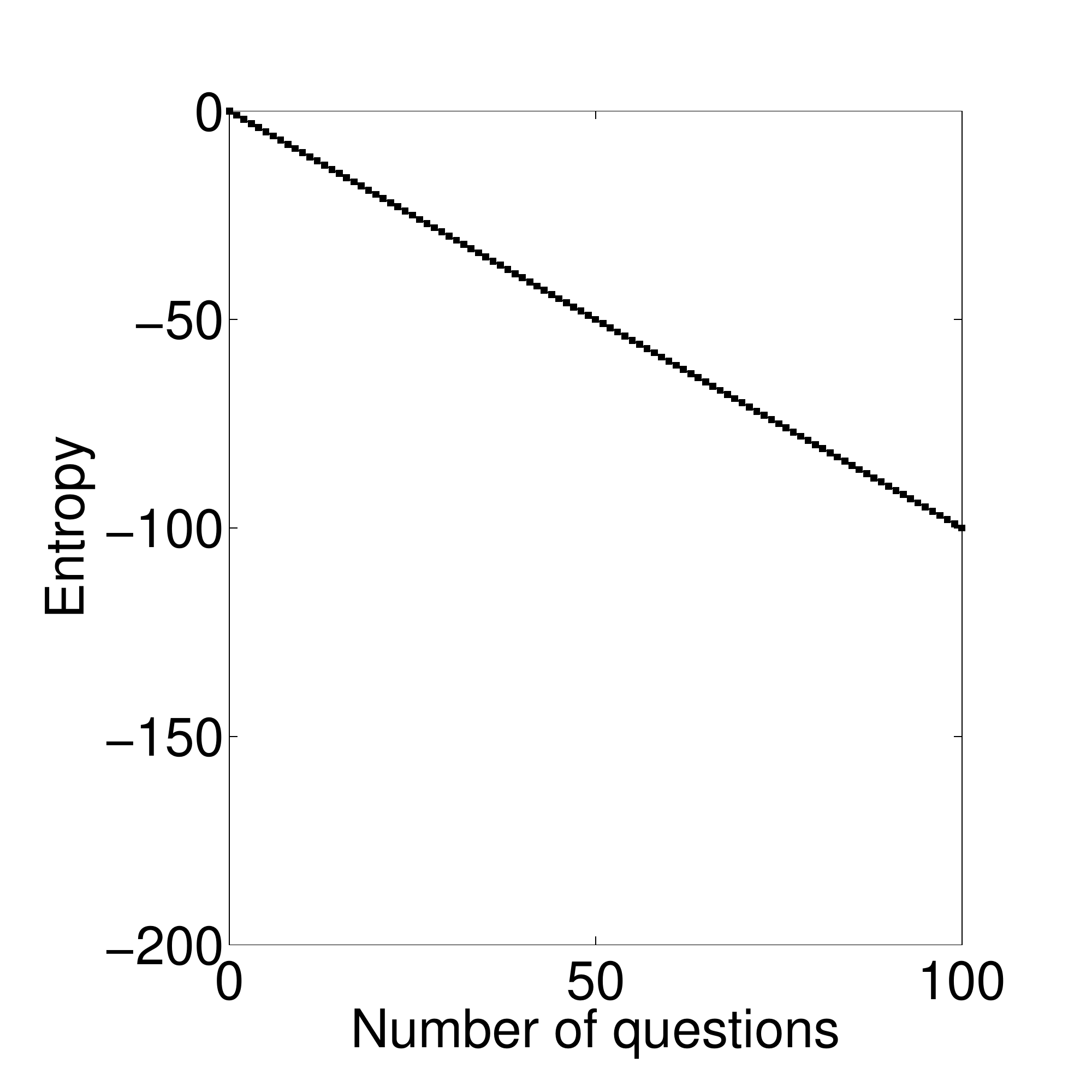}
\endminipage\hfill
\minipage{0.33\textwidth}
  \includegraphics[width=\linewidth]{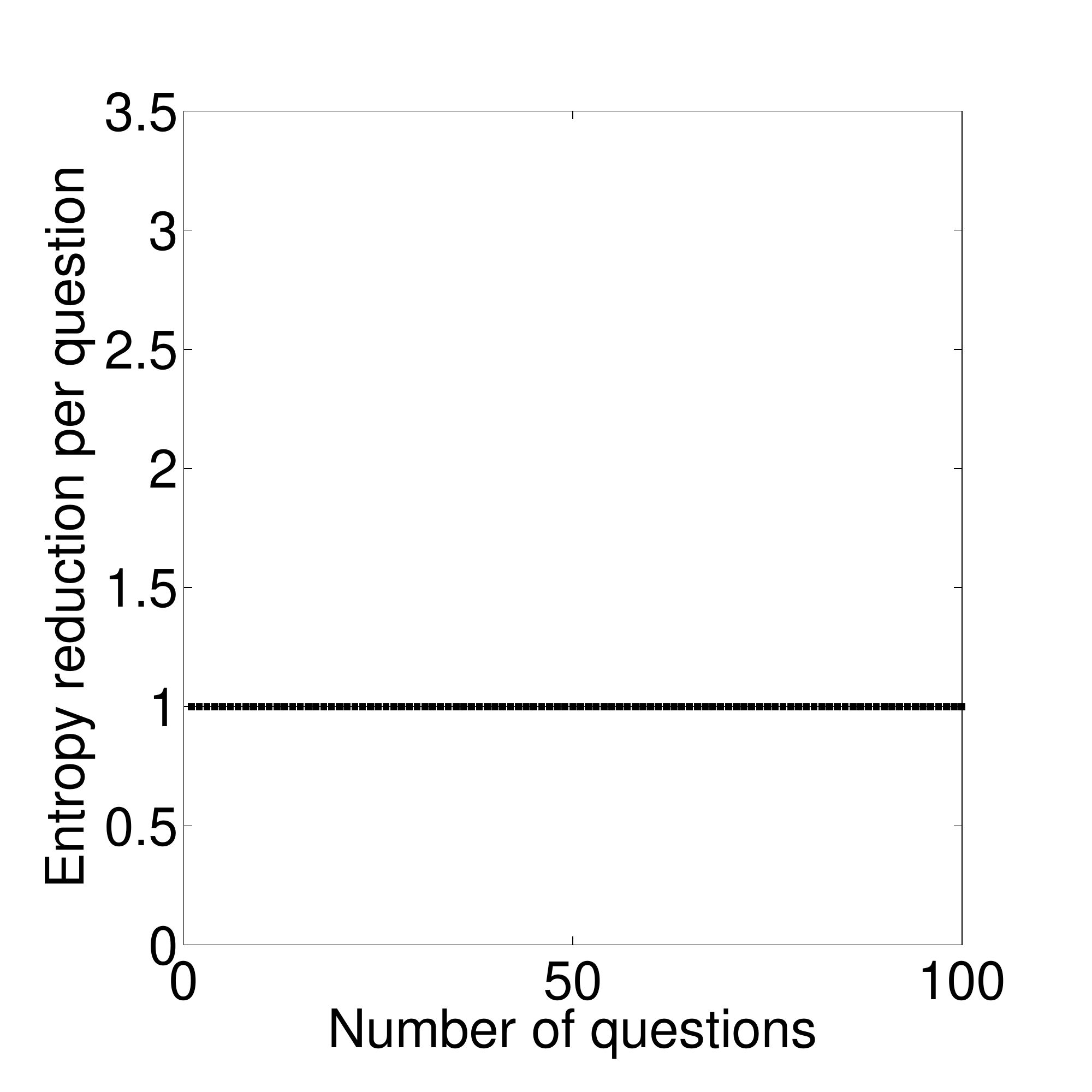}
\endminipage\hfill
\minipage{0.33\textwidth}
  \includegraphics[width=\linewidth]{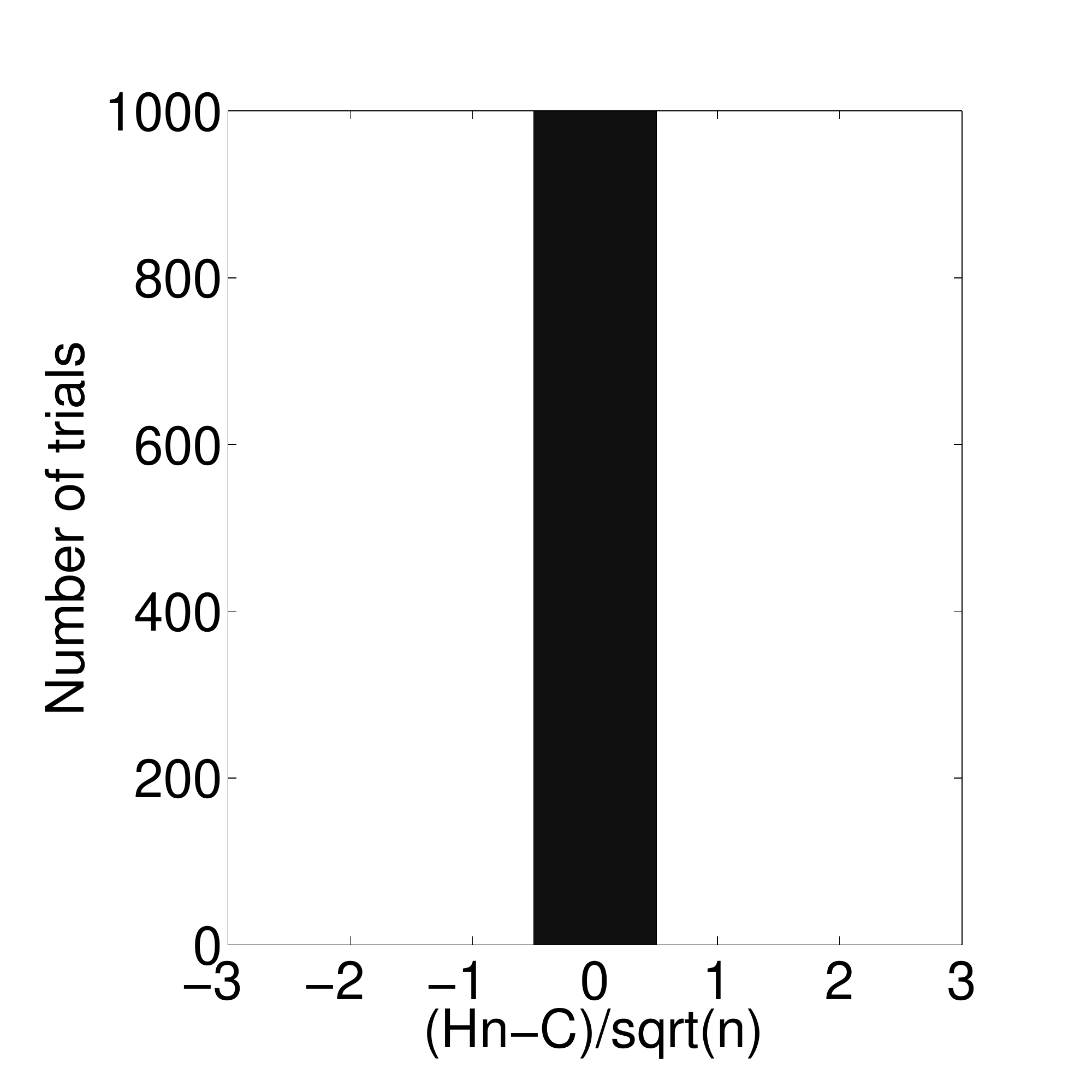}
\endminipage\hfill
\end{figure}

\begin{figure}[H]
\minipage{0.33\textwidth}
  \includegraphics[width=\linewidth]{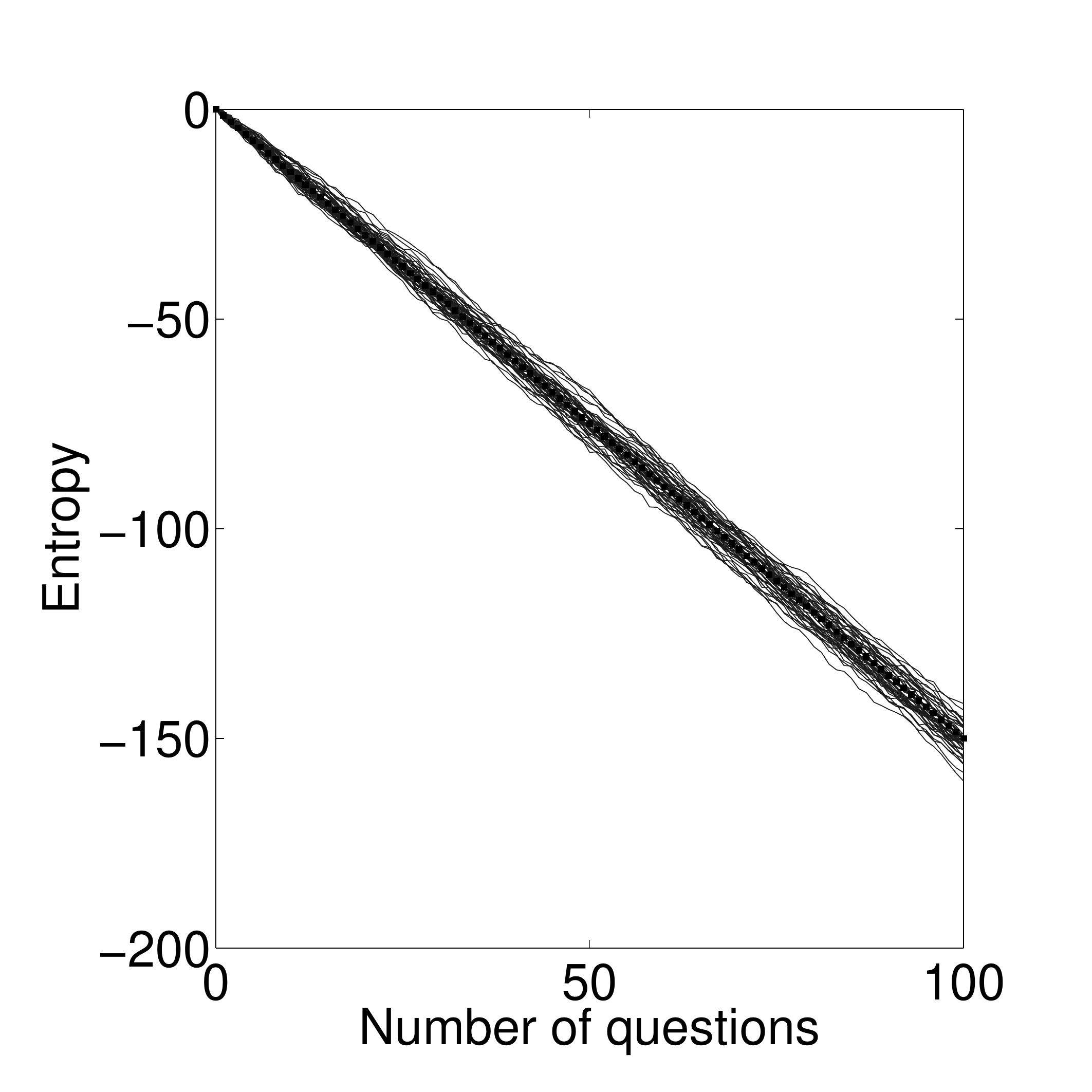}
\endminipage\hfill
\minipage{0.33\textwidth}
  \includegraphics[width=\linewidth]{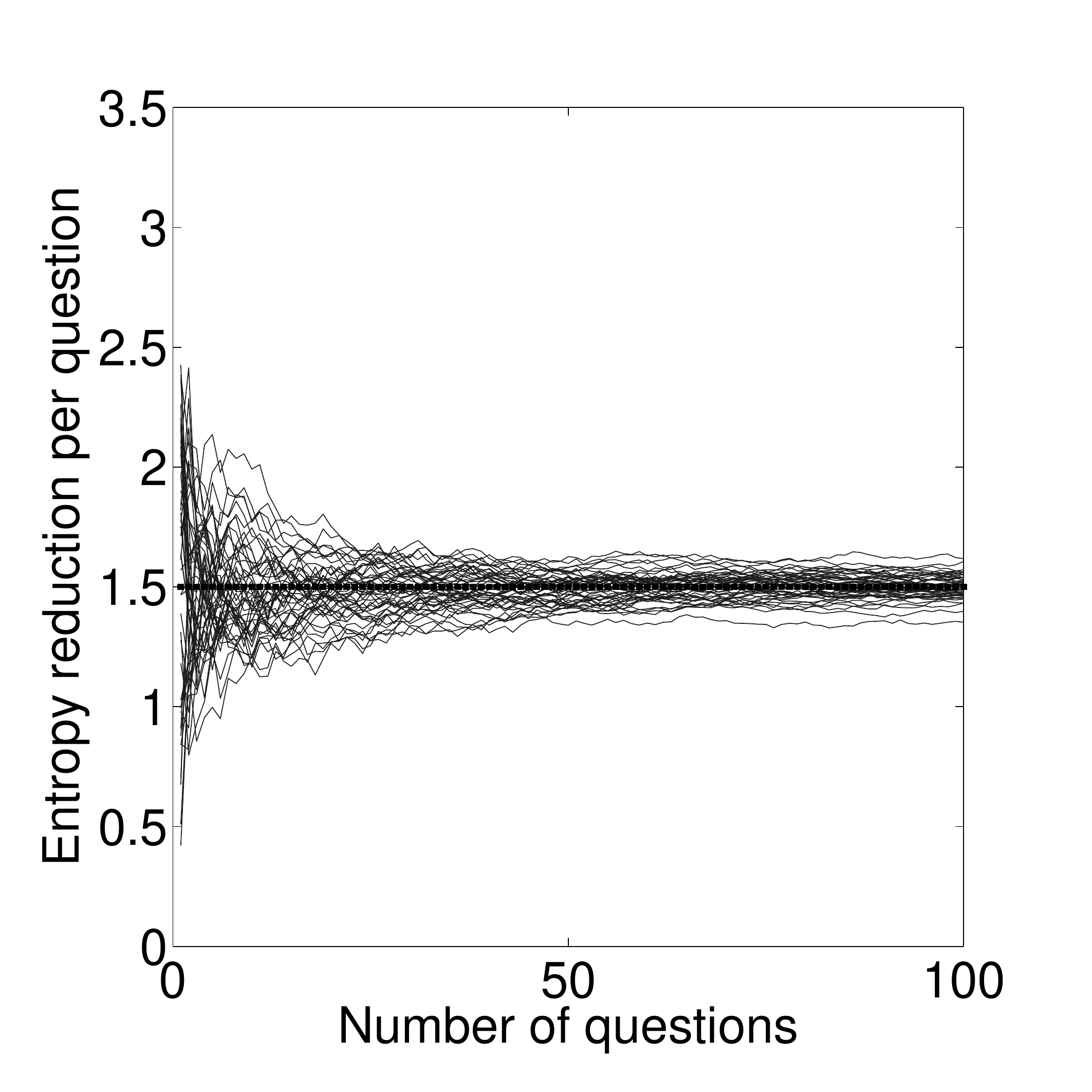}
\endminipage\hfill
\minipage{0.33\textwidth}
  \includegraphics[width=\linewidth]{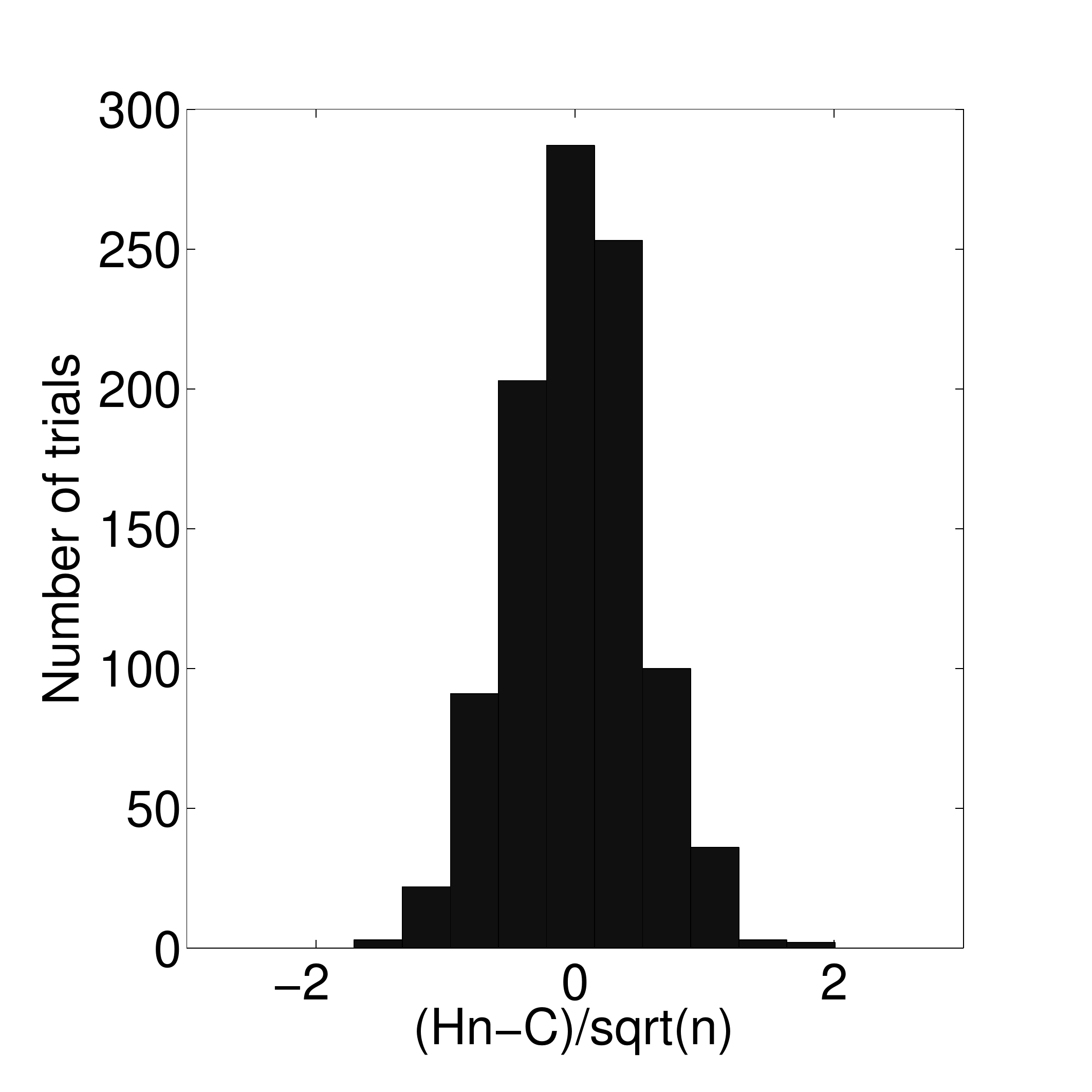}
\endminipage\hfill
\end{figure}

\begin{figure}[H]
\minipage{0.33\textwidth}
  \includegraphics[width=\linewidth]{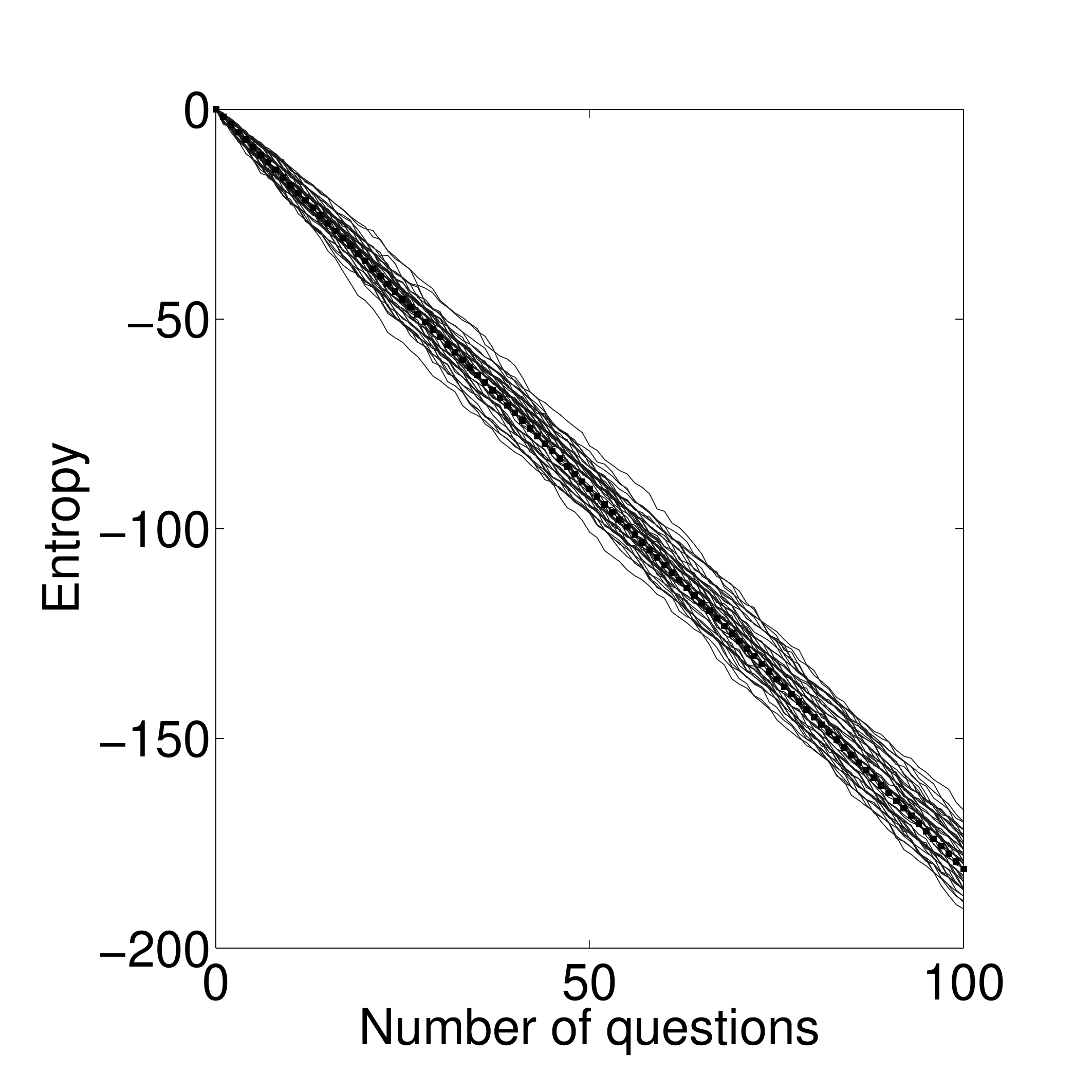}
\endminipage\hfill
\minipage{0.33\textwidth}
  \includegraphics[width=\linewidth]{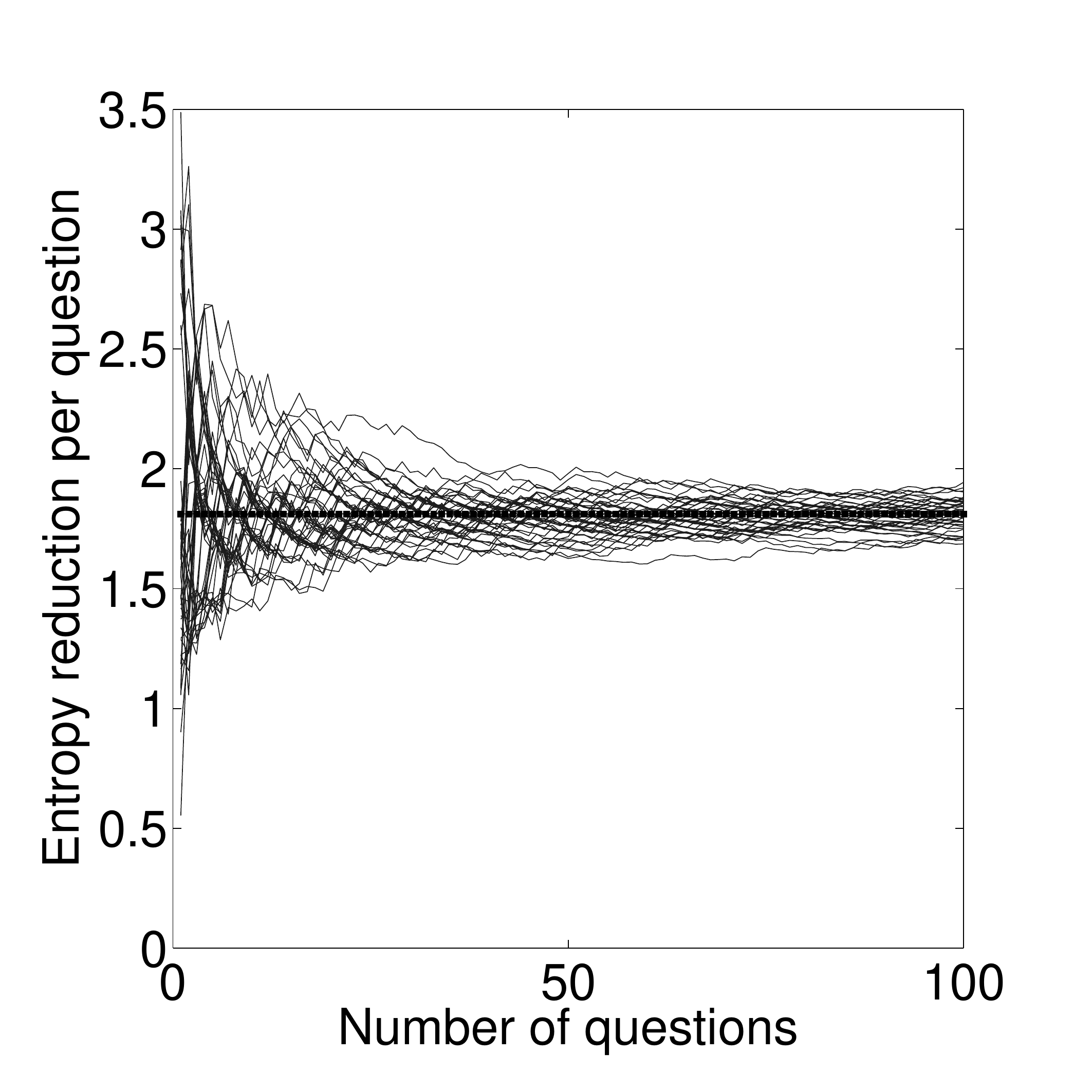}
\endminipage\hfill
\minipage{0.33\textwidth}
  \includegraphics[width=\linewidth]{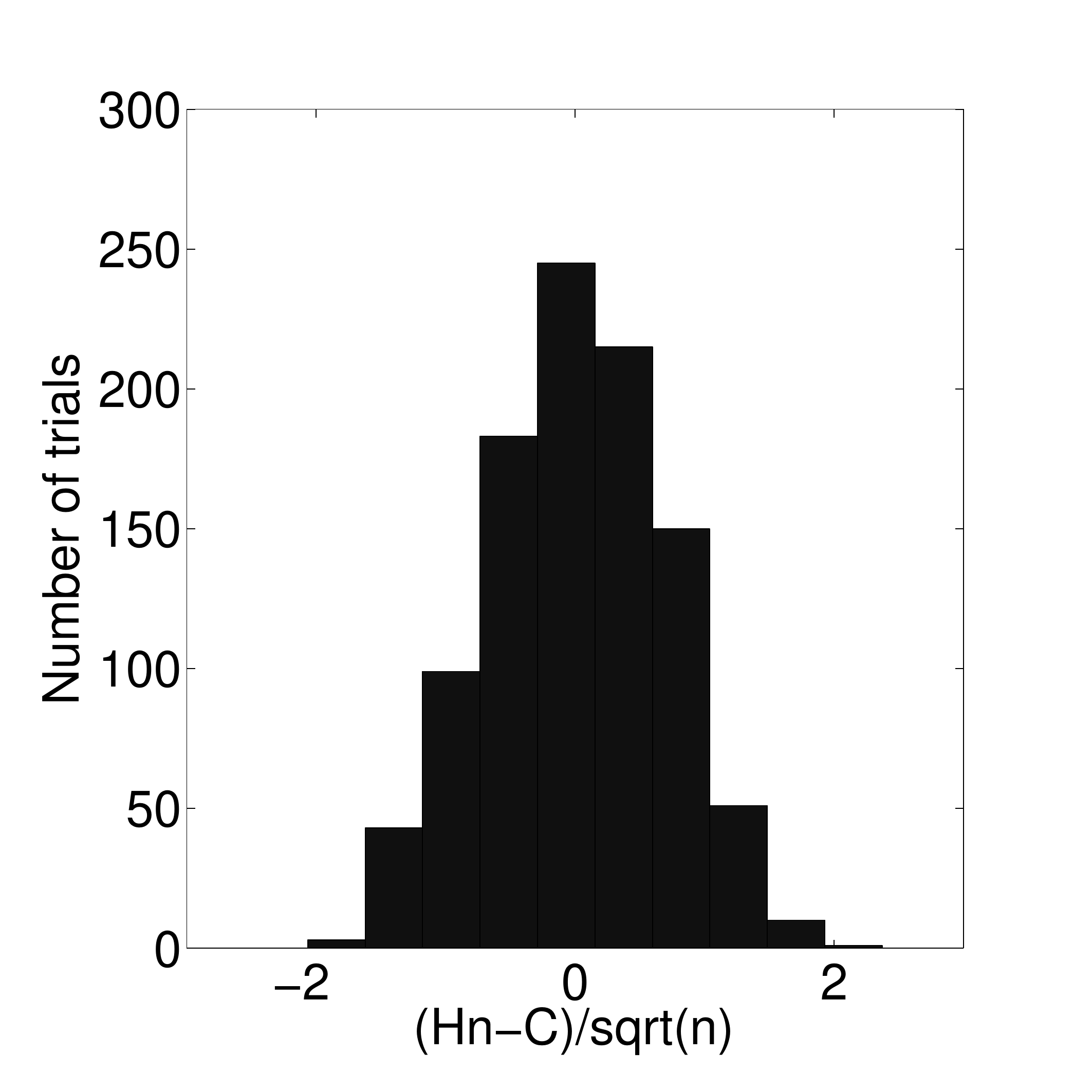}
\endminipage
\caption{Simulation results for localizing one, two, three objects under the dyadic policy. $N=100$ and $f_0$ is uniform over $(0,1]$. The horizontal graphs above show the actual trajectories of entropy $H(p_n)$, average reduction in entropy $-\frac{H(p_n)}{n}$, and normality of  $\frac{H(p_N)+NH\left(\Bin\left(k,\frac{1}{2}\right)\right)}{\sqrt{N}}$, respectively.}
\label{fig:dyadic3objects}
\end{figure}

\section{The Greedy Policy for Localizing Multiple Objects}
\label{sec:greedy}

In this section, we will present the second policy of interest--the \emph{greedy policy}. The greedy policy is a family of policies (not unique) which pursue a maximal one-step expected reduction in entropy. Despite having a better performance than the dyadic policy, the greedy policy is difficult for us to parametrize and implement. A description of the greedy policy will be given in Section \ref{sec:greedy_description} and a lower bound of its rate is shown in Section \ref{sec:greedy_value}, which verifies our claim of the third inequality in the main results \eqref{eq:main}. Furthermore, we will provide an example in which the greedy policy outperforms the dyadic policy in Section \ref{sec:comparison} and thus this inequality cannot be reversed.

\subsection{Description of the greedy policy}
\label{sec:greedy_description}
Unlike the dyadic policy, the greedy policy is adaptive, that is, the actual policy depends on the previous answers that we already observed, and at each step the question set $A_n\subset \mathbb R$ is defined in \eqref{eq:greedy} to maximize the one-step expected reduction in entropy.

We prove that this argmin exists below in Theorem~\ref{thm:greedy}.
The computation of the greedy policy might be complicated in some cases, however, the greedy policy is strictly better than the dyadic policy and we will demonstrate this point in Section \ref{sec:comparison}.

\subsection{The rate of the greedy policy}
\label{sec:greedy_value}
Although deriving the rate of the greedy policy seems impossible, we are able to employ Lemma \ref{lem:Eentropy_onestep} in the appendix to derive a lower bound of it as the following.
\begin{thm}
\label{thm:greedy}
The argmin \eqref{eq:greedy} defining the class of greedy policies exists.
Under any greedy policy $\pi_G$,
\begin{equation}
R(\pi_G,N) \ge H\left(\Bin\left(k,\frac{1}{2}\right)\right).
\end{equation}
\end{thm}

\begin{proof}
Fix some history $B_n=(Z,X_{1:n})=b_n$.
We first show existence of the argmin from \eqref{eq:greedy}, restated here as
\begin{equation}
  \argmin_A E[H(p_{n+1})|p_n,A_{n+1}=A],
  \label{eq:proof-thm-greedy}
\end{equation}
where we recall that the minimum is taken over all Borel-measurable subsets of $\mathbb R$.

Since conditioning on the posterior distribution $p_n$ under any fixed policy is equivalent to conditioning on $\{B_n=(Z, A_{1:n},X_{1:n})=b_n\}$, using \eqref{eq:Eentropy_onestep} in Lemma \ref{lem:Eentropy_onestep} in the appendix, we have
\begin{equation}
\begin{split}
\label{eq:onestep_fixY}
E[H(p_{n+1})| p_n, A_{n+1}=A] &= E[H(p_{n+1})| B_n=b_n, A_{n+1}=A] \\
&=H(p_n|B_n=b_n,A_{n+1}=A)-H(X_{n+1}|B_n=b_n,A_{n+1}=A).
\end{split}
\end{equation}

Since the first term $H(p_n|B_n=b_n,A_{n+1}=A)$ does not depend on $\{A_{n+1}=A\}$, \eqref{eq:proof-thm-greedy} can be rewritten as
\begin{equation}
\label{eq:greedy_maximizer}
\begin{split}
&\argmin\limits_A H(p_n|B_n=b_n,A_{n+1}=A)-H(X_{n+1}|B_n=b_n,A_{n+1}=A)\\
&= \argmax\limits_A H(X_{n+1}|B_n=b_n,A_{n+1}=A).
\end{split}
\end{equation}

When $n=0$, according to Theorem \ref{thm:postY}, we can rewrite the above argmax as
\begin{equation}
\argmax\limits_A H\left(\Bin\left(k,f_0(A)\right)\right).
\end{equation}
The maximum is achieved by any questions set $A$ such that $f_0(A)=\frac{1}{2}$. For example, the first dyadic question $\left(Q\left(\frac{1}{2}\right),Q(1)\right]\cap \Support{f_0}$ is one of such sets. This also proves $H^{\pi_G}(X_1|B_0)=H\left(\Bin\left(k,\frac{1}{2}\right)\right)$.

When $n\ge 1$, using \eqref{eq:postY} in Theorem \ref{thm:postY}, we can rewrite the argmax in \eqref{eq:greedy_maximizer} as
\begin{equation}
\label{eq:greedy_maximizer2}
\argmax\limits_{A} H\left(\sum_{\mathcal S \in E_n}\alpha(\mathcal S)\PB\left(\frac{f_0(A \cap C_{s^{(1)}})}{f_0(C_{s^{(1)}})},\dots, \frac{f_0(A \cap C_{s^{(k)}})}{f_0(C_{s^{(k)}})}\right)\right),
\end{equation}
where $\alpha(\mathcal S)=\frac{f_0(C_{s^{(1)}})\dots f_0(C_{s^{(k)}})}{\sum\limits_{\mathcal T \in E_n}f_0(C_{t^{(1)}})\dots f_0(C_{t^{(k)}})}$ and  $\sum\limits_{\mathcal S \in E_n}\alpha(\mathcal S)=1$.

\newcommand{\sset}{\mathbb{S}}
Let $\sset = \{s\in\{0,1\}^n : C_s \ne \emptyset\}$, and fix some arbitrary order of these elements so that $\sset$ becomes a sequence rather than a set.
For each $s\in\sset$, let $r_s(A) = f_0(A\cap C_s) / f_0(C_s)$ so that \eqref{eq:greedy_maximizer2} can be rewritten as
\begin{equation}
\argmax\limits_{A} H\left(\sum_{\mathcal S \in E_n}\alpha(\mathcal S)\PB\left(r_{s^{(1)}}(A),\ldots,r_{s^{(k)}}(A)\right)\right).
\end{equation}

For each Borel-measurable subset $A$ of $\mathbb R$, $r(A) =(r_s(A) : s\in\sset)$ is an element of $[0,1]^{|\sset|}$.  Moreover, for each $r \in [0,1]^{|\sset|}$, there is a Borel-measurable $A\subset\mathbb R$ such that $r(A)=r$.  This is because the continuity of the prior cumulative density function allows us to construct the desired subset $A$ as a union of sets, one for each element of $\sset$. In this construction, the subset of $A$ corresponding to $s\in\sset$ is a subset of $C_s$ containing a fraction $r_s$ of the prior mass of $C_s$.  This shows that the argmax \eqref{eq:greedy_maximizer2} exists iff the following argmax exists:
\begin{equation}
  \argmax\limits_{r\in[0,1]^{|\sset|}} H\left(\sum_{\mathcal S \in E_n}\alpha(\mathcal S)\PB\left(r_{s^{(1)}},\ldots,r_{s^{(k)}}\right)\right).
\end{equation}

The function $r \mapsto H\left(\sum_{\mathcal S \in E_n}\alpha(\mathcal S)\PB\left(r_{s^{(1)}},\ldots,r_{s^{(k)}}\right)\right)$ is continuous, and the set $[0,1]^{|\sset|}$ is compact, so this argmax is attained.  This shows that the argmax \eqref{eq:greedy} defining the class of greedy policies is well-defined.

We now show a lower bound on the rate of any greedy policy $\pi_G$ by showing a lower bound on
this quantity.  The argument above also shows that under any greedy policy $\pi_G$, for $n\ge 1$,
\begin{subequations}
\begin{align}
H^{\pi_G}(X_{n+1}|B_n=b_n)
&= \max_A H(X_{n+1}|B_n=b_n, A_{n+1}=A) \\
&= \max\limits_{r\in[0,1]^{|\sset|}} H\left(\sum_{\mathcal S \in E_n}\alpha(\mathcal S)\PB\left(r_{s^{(1)}},\ldots,r_{s^{(k)}}\right)\right)\\
&\geq \max\limits_{r\in[0,1]^{|\sset|}} \sum_{\mathcal S \in E_n} \alpha(\mathcal S)\; H\left(\PB\left(r_{s^{(1)}},\ldots,r_{s^{(k)}}\right)\right)\label{eq:Kgreedy1}\\
&\geq \sum\limits_{\mathcal S \in E_n}\alpha(\mathcal S)\; H\left(\PB\left(\frac{1}{2},\dots,\frac{1}{2}\right)\right)\label{eq:Kgreedy2}\\
&= H\left(\Bin\left(k,\frac{1}{2}\right)\right). \label{eq:Kgreedy3}
\end{align}
\end{subequations}
Above, we use the concavity of the entropy function to obtain the inequality \eqref{eq:Kgreedy1}, and that $\PB(\frac{1}{2},\ldots,\frac{1}{2})$ a special case of a Poisson Binomial Distribution to obtain \eqref{eq:Kgreedy2}.  The last line, \eqref{eq:Kgreedy3}, follows from $\sum_{\mathcal S \in E_n} \alpha(\mathcal S) = 1$ and the fact that $\PB(\frac12,\ldots,\frac12)$ is the $\Bin\left(k,\frac12\right)$ distribution.

Furthermore, taking the expectation over all possible realizations of $B_n$, we obtain for $n\ge 1$,
\begin{equation}
H^{\pi_G}(X_{n+1}|B_n) \ge H\left(\Bin\left(k,\frac{1}{2}\right)\right).
\end{equation}
Recall that we already have $H^{\pi_G}(X_1|B_0)=H\left(\Bin\left(k,\frac{1}{2}\right)\right)$ from previous arguments.

Finally, \eqref{eq:Eentropy_all} in Lemma \ref{lem:Eentropy_onestep} in the appendix shows
\begin{equation}
R(\pi_G,N) = \frac{H_0-E^{\pi_G}[H(p_N)]}{N}
= \frac{\sum_{j=0}^{N-1}H^{\pi_G}(X_{n+1}|B_n)}{N}
\ge H\left(\Bin\left(k,\frac{1}{2}\right)\right).
\end{equation}
\end{proof}
Theorem \ref{thm:greedy} above implies that the greedy policy never underperforms the dyadic policy.

\subsection{A setting in which the greedy policy is strictly better than the dyadic policy}
\label{sec:comparison}
From the proof above, we can see that the greedy policy is strictly better than the dyadic policy if there exists some $b_n$ such that the inequality \eqref{eq:Kgreedy1} is strict. In the following, based on the previous examples in Section \ref{sec:post_object}, we develop an example showing that the greedy policy is \emph{strictly} better than the dyadic policy.

\paragraph{Example 3:} Suppose $\theta_1,\theta_2$ are two objects located in (0,1] with the prior $f_0$ being uniform over $(0,1]$, and $A_1$ and $A_2$ the first two questions of the dyadic policy, $A_1 = \left(\frac{1}{2},1\right]$ and $A_2 = \left(\frac{1}{4},\frac{1}{2}\right] \cup \left(\frac{3}{4},1\right]$. Now consider the following family of questions $A_3$ indexed by $0 \leq \alpha,\beta \leq 1$:
\begin{equation}
A_3=\left(\frac{1-\alpha}{4},\frac{1}{4}\right]\cup \left(\frac{2-\beta}{4},\frac{1}{2}\right] \cup \left(\frac{3-\beta}{4},\frac{3}{4}\right]\cup \left(\frac{4-\alpha}{4},1\right].
\end{equation}
Firstly, assume $X_1=0$ and $X_2=2$, which corresponds {\bf Example 1} in Figure \ref{fig:eg_matrix}. According to \eqref{eq:postY}, the probability mass function of $X_3$ is
\begin{equation}
P(X_3=x)=f_{\PB}(x;q_1=\beta,q_2=\beta),
\end{equation}
which is a Binomial distribution with parameter $\beta$. The maximum entropy is then achieved when $\beta=0.5$. Note that the dyadic question, corresponding to $\alpha=\beta=0.5$, verifies this condition and as a consequence is also a valid question for the greedy policy.

Now, more interestingly, assume that $X_1=X_2=1$, which corresponds {\bf Example 2} in Figure \ref{fig:eg_matrix}. According to \eqref{eq:postY}, the probability mass function of $X_3$ is
\begin{eqnarray}
p_2(X_3=x) & = & \frac{1}{4}f_{\PB}(x;q_1=\alpha,q_2=\alpha) +  \frac{1}{4}f_{\PB}(x;q_1=\beta,q_2=\beta) \nonumber\\
& + &  \frac{1}{4}f_{\PB}(x;q_1=\beta,q_2=\beta) +  \frac{1}{4}f_{\PB}(x;q_1=\alpha,q_2=\alpha),
\end{eqnarray}
which simplifies to
\begin{center}
\begin{tabular}{|l|l|}
\hline
$x$ & $p_2(X_3=x)$ \\
\hline
0 & $\frac{1}{2}(1-\alpha)^2 + \frac{1}{2}(1-\beta)^2$\\
\hline
1 & $\alpha(1-\alpha)+\beta(1-\beta)$\\
\hline
2 & $\frac{1}{2}\alpha^2 + \frac{1}{2}\beta^2$\\
\hline
\end{tabular}
\end{center}
Now, one can choose values for $\alpha$ and $\beta$ such that $p_2(X_3=x)=\frac{1}{3}$, $x=0,1,2$. Specifically,
\begin{equation}
\alpha = \frac{1+\frac{\sqrt{3}}{3}}{2} \mbox{ and } \beta=\frac{1-\frac{\sqrt{3}}{3}}{2}.
\end{equation}
In this case $H(p_2(X_3=\cdot))=\log(3)> 1.5$ which shows that the greedy policy is in this case strictly better than the dyadic policy.

\section{Simulation}
\label{sec:dyadic_numerical}

We now present a toy Computer Vision object localization experiment. We show how the dyadic policy, analyzed above in the continuous setting using the entropy, can be used in a simulation setting to 
locate $k$ instances of a given object within a $M \times M$ digital image. 

The dyadic policy is unique in the fact that it has a simple closed form expression for the posterior probability that the object instance is located at a pixel location $C$ (See Lemma \ref{lem:posterior ranking}). We use the term ``screening questions'' to denote the instance count queries on the subset $A$. The ``oracle'' refers to the expensive but highly accurate classifier that will be run on a selected subset of the pixel locations. Note that instead of computing the entropy of the posterior distribution, we use the number of calls to the oracle as the measure of performance.

In this setting we use a $M \times M$ image in which the objects to be located are of size one pixel and have the intensity value $1$ and the rest of the pixels are of intensity $0$, with $k = \{2,3,10\}$ and $M = \{8, \ldots, 1024\}$. This simulation setting is far from a realistic Computer Vision application as the number of instances in each dyadic query set $A$ will not be readily available and we would need to train an appropriate Machine Learning classifier that computes this value, albeit with noise (~\cite{lempitsky2010learning}, ~\cite{idrees2013multi}, ~\cite{barinova2012detection}). Moreover, in realistic applications, the number of objects $k$ will not be known in advance and the answers to the screening questions could be corrupted with noise. Nevertheless, this simulation experiment provides useful analysis of the performance of the dyadic policy and the algorithms that locate object instances using the posterior distribution computed from the answers to the screening questions using Lemma \ref{lem:posterior ranking}.

%

We consider algorithms that proceed in 2 phases, eventually iterated.  The first phase consists in querying the dyadic sets. As opposed to the continuous domain, there is here a limited supply of dyadic sets. Choosing for $M$ a power of 2, there are $\log M$ dyadic horizontal queries and $\log M$ dyadic vertical queries. The two rows of Figure \ref{Dyadic} present the dyadic questions for $M=16$. The second phase consists in ordering the pixels and querying the oracle according to this ordering. We compare three algorithms: Posterior Rank (PR), Iterated Posterior Rank (IPR), and Entropy Pursuit (EP). We will see that all these three algorithms significantly outperform the baseline algorithm--the Index Rank (IR) algorithm--in terms of the expected number of calls to the oracle (see Figure \ref{Results_IR_PR_IPR_EP}).

The algorithms are motivated by the computation of $E[N(C)|B_N]$, the expected number of instances within pixel $C$ given the history of screening questions and answers $B_N$. The following lemma provides an explicit formula for this posterior probability. Using this result, we can order the pixels in decreasing order of $E[N(C)|B_N]$ and run the oracle according to this order until all the instances of the object are found.

\begin{lem}
\label{lem:posterior ranking}
Under the dyadic policy, for each object $\theta_i$ and each binary sequence $s\in\{0,1\}^N$, the posterior likelihood $P\left(\theta_i \in C_s|B_N\right)$ satisfies
\begin{equation}
\label{eq:posterior ranking}
P\left(\theta_i \in C_s|B_N\right) = \prod\limits_{j=1}^N \left(\frac{X_j}{k}\right)^{s_j} \left(1-\frac{X_j}{k}\right)^{1-s_j},
\end{equation}
Moreover, 
\begin{equation}
E[N(C_s)|B_N]=kP\left(\theta_1 \in C_s|B_N\right),
\end{equation}
where $N(C_s)$ denotes the number of objects located in the set $C_s$.
\end{lem}

\begin{proof}
First of all, note that under the dyadic policy, $P\left(\theta_i \in C_s|B_N\right)=P\left(\theta_i \in C_s|X_{1:N}\right)$. This is because the questions $A_{1:N}$ are deterministic by construction, and $Z$ is independent of $\sigma(\theta,X_{1:N})$, so that $Z,A_{1:N}$ can be removed from the condition. Let $x_{1:N}\in\{0,\dots,k\}^N$ be a fixed realization of $X_{1:N}$. Now we consider the event $\{\theta_i \in C_s| X_{1:N}=x_{1:N}\}$, for a fixed binary sequence $s\in\{0,1\}^N$.

Let us denote by $E_N(C_s)$ the collection of matrices $\mathcal S \in E_N$ that are consistent with the event $\{\theta_i \in C_s|X_{1:N}=x_{1:N}\}$, i.e. the $i$-th column of $\mathcal S$ is $s^{(i)}=s$. Note that $p_0(C_{\mathcal S})=2^{-Nk}$ for all $\mathcal S\in E_N$ under the dyadic policy. For simplicity, define $D_{C_s}=\bigcup_{\mathcal S \in E_N(C_s)} C_{\mathcal S}$. Therefore, using Lemma \ref{lem:productk}, we can compute the probability of $P(\theta_1 \in C_s|X_{1:N}=x_{1:N})$ as
\begin{equation}
\label{eq:probX}
\begin{split}
&P(\theta_1 \in C_s|X_{1:N}=x_{1:N})\\
&=\int_{u_{1:k}\in D_{C_s}} p_N(u_{1:k})\,du_{1:k}\\
&=\sum_{\mathcal S\in E_N(C_s)}\int_{u_{1:k}\in C_{\mathcal S}} \frac{p_0(u_{1:k})}{\sum\limits_{\mathcal S\in E_N} p_0(C_{\mathcal S})}\,du_{1:k}\\
&= \sum_{\mathcal S\in E_N(C_s)}\frac{1}{2^{-Nk}|E_N|}\int_{u_{1:k}\in C_{\mathcal S}} p_0(u_{1:k})\,du_{1:k}\\
&=\frac{|E_N(C_s)|}{|E_N|},
\end{split}
\end{equation}
where $|E_N(C_s)|,|E_N|$ denote the cardinalities of $E_N(C_s),E_N$, respectively. $|E_N|$ can be computed using \eqref{eq:En} in the appendix. 
The matrices in $E_N(C_s)$ need to satisfy one more constraint--the $i$-th column is fixed to be $s$. Using similar arguments, we have
\begin{equation}
\label{eq:numC}
|E_N(C_s)|= \prod\limits_{n=1}^N {k-1 \choose x_n-s_n}1_{\{0 \leq x_n-s_n \leq k-1\}}.
\end{equation}

Combining \eqref{eq:probX}, \eqref{eq:En} and \eqref{eq:numC} together yields,
\begin{equation}
\begin{array}{ccl}
P(\theta_i \in C_s|X_{1:N}=x_{1:N}) &=& \frac{\prod\limits_{n=1}^N {k-1 \choose x_n-s_n}1_{\{0 \leq x_n-s_n \leq k-1\}}}{\prod\limits_{n=1}^N {k \choose x_n}}\\
&=&\prod\limits_{n=1}^N\begin{cases}
	\frac{x_n}{k}, &\text{if $s_n=1$}\\
	1-\frac{x_n}{k}, &\text{if $s_n=0$}\end{cases}.
\end{array}
\end{equation}

Equivalently,
\begin{equation}
\label{eq:single}
P\left(\theta_i \in C_s|X_{1:N}=x_{1:N}\right) = \prod\limits_{n=1}^N \left(\frac{x_n}{k}\right)^{s_n} \left(1-\frac{x_n}{k}\right)^{1-s_n}.
\end{equation}
\end{proof}

Based on the lemma above, we obtain the following PR Algorithm. 

\begin{algorithm}[H]
\caption{Posterior Rank (PR) Algorithm}
\label{algo:post_rank}
\begin{algorithmic}[1]
\STATE Compute the answers to the screening questions.
\STATE Compute the posterior rank $r$ according to \eqref{eq:posterior ranking}.
\STATE Run the oracle on the pixels according to $r$ until all the objects are found.
\end{algorithmic}
\end{algorithm}

The IPR algorithm is a variation of the PR Algorithm. As before, the pixels are searched in decreasing order of the expected number of objects. When the oracle locates one(several) object(s) at a pixel, the answers of the screening questions for the remaining objects are recomputed based on the location of the objects already found. This is equivalent to {\em masking} the objects already found and asking the screening questions again. The expected number of objects per pixel is then recomputed and provides an updated ranking for the remainder of the search. The algorithm is provided below.

\vspace{-0.2cm}
\begin{algorithm}[H]
\small
\caption{Iterated Posterior Rank (IPR) Algorithm}
\label{algo:iter_post_rank}
\begin{algorithmic}[1]
\REPEAT
\STATE Compute the answers to the screening questions.
\STATE Compute the posterior rank $r$ according to \eqref{eq:posterior ranking}.
\STATE
Run the oracle on the pixels according to $r$ until one (several) object(s) is (are) found at a pixel.
\STATE Mask this (these) object(s).
\UNTIL {all the objects are found.}
\end{algorithmic}
\end{algorithm}
\vspace{-0.2cm}


Figure \ref{Dyadic} and \ref{Object_Localization} illustrate the procedures in the IPR algorithm for a $16\times16$ image with $k=4$ objects. Figure \ref{Dyadic} illustrates the screening questions under the dyadic policy, with light regions marking the questions sets. The first line of \ref{Object_Localization} shows the true but unknown locations of the objects in each iteration of the IPR algorithm. The second line shows the expected number of objects within each pixel computed after screening questions in each iteration, respectively, with lighter regions having a higher expected number of objects.

\newcommand{\imgWidthTiny}{0.11\textwidth}
\begin{figure}[h]
\vspace{-1.7 mm}
\begin{center}$
\begin{array}{@{\hspace{-0.5em}}c@{\hspace{.5em}}c@{\hspace{.5em}}c@{\hspace{.5em}}c}
\includegraphics[width=\imgWidthTiny]{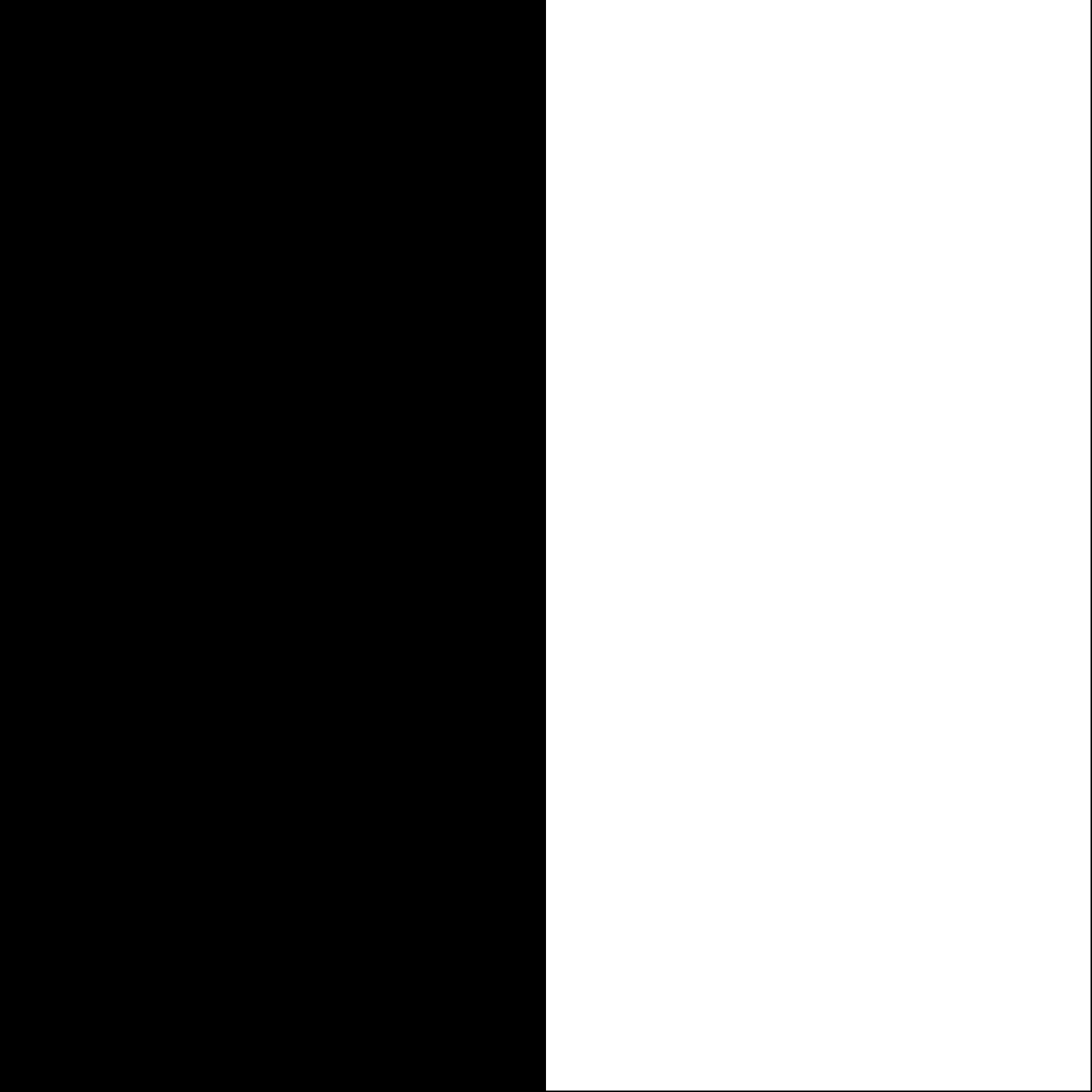} &
\includegraphics[width=\imgWidthTiny]{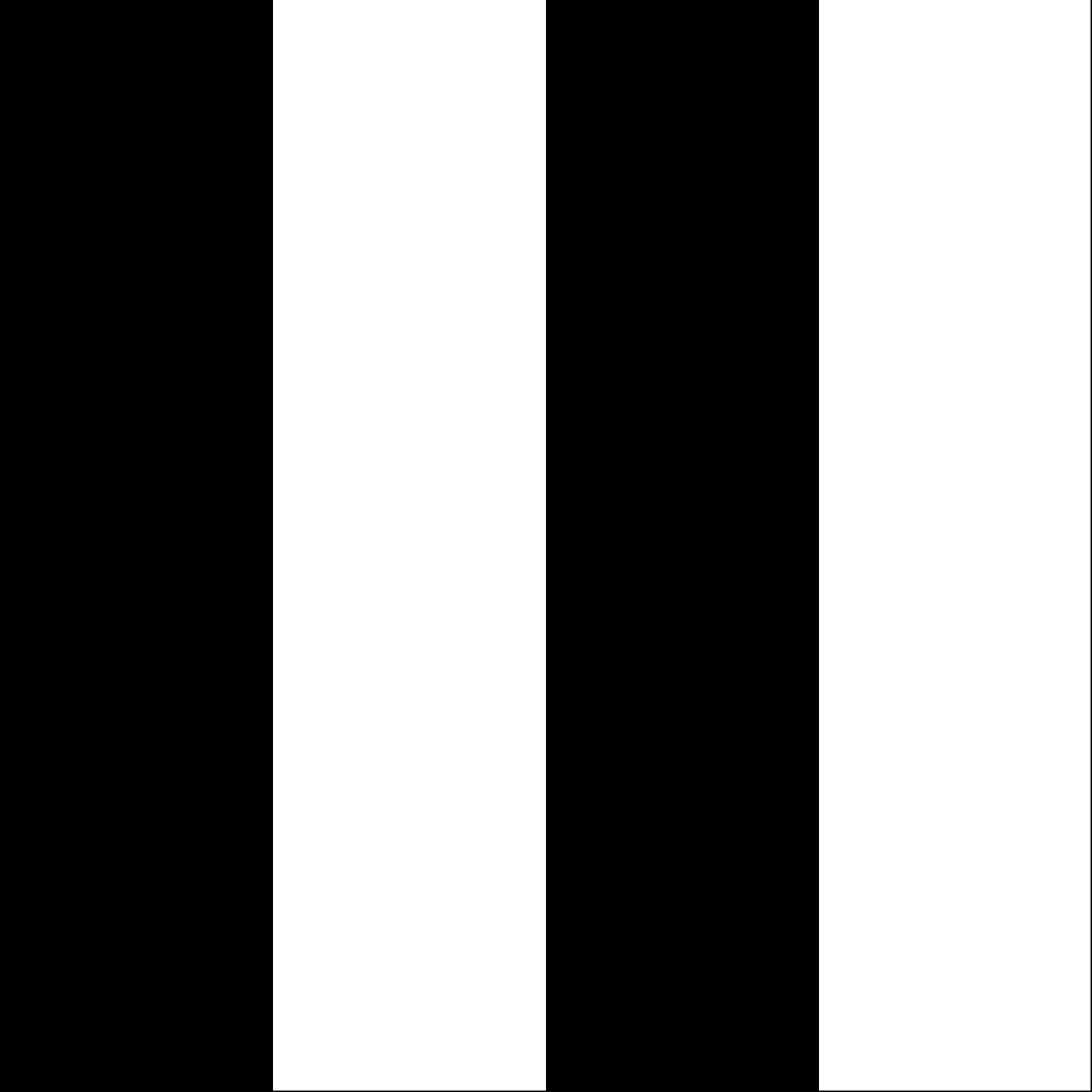} &
\includegraphics[width=\imgWidthTiny]{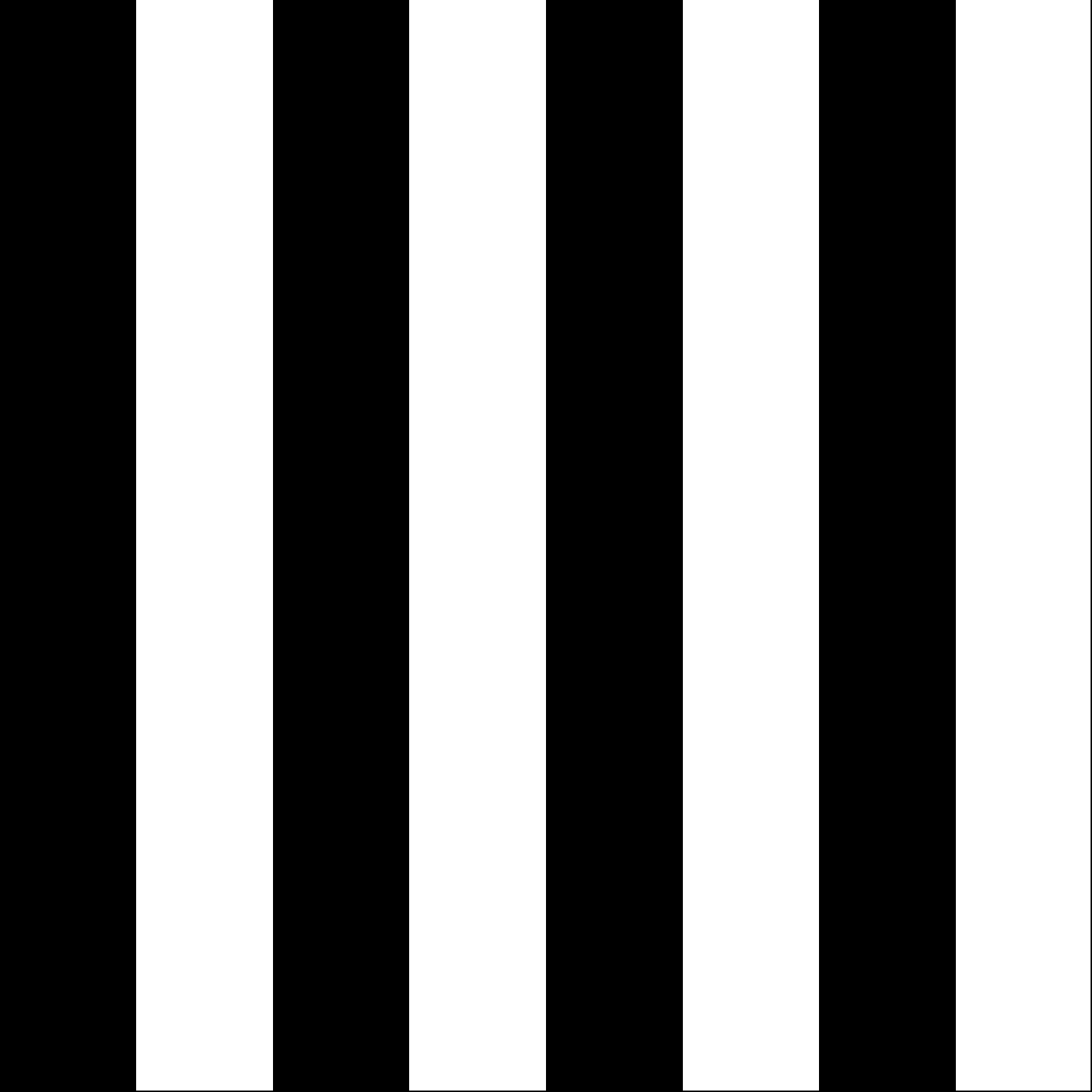} &
\includegraphics[width=\imgWidthTiny]{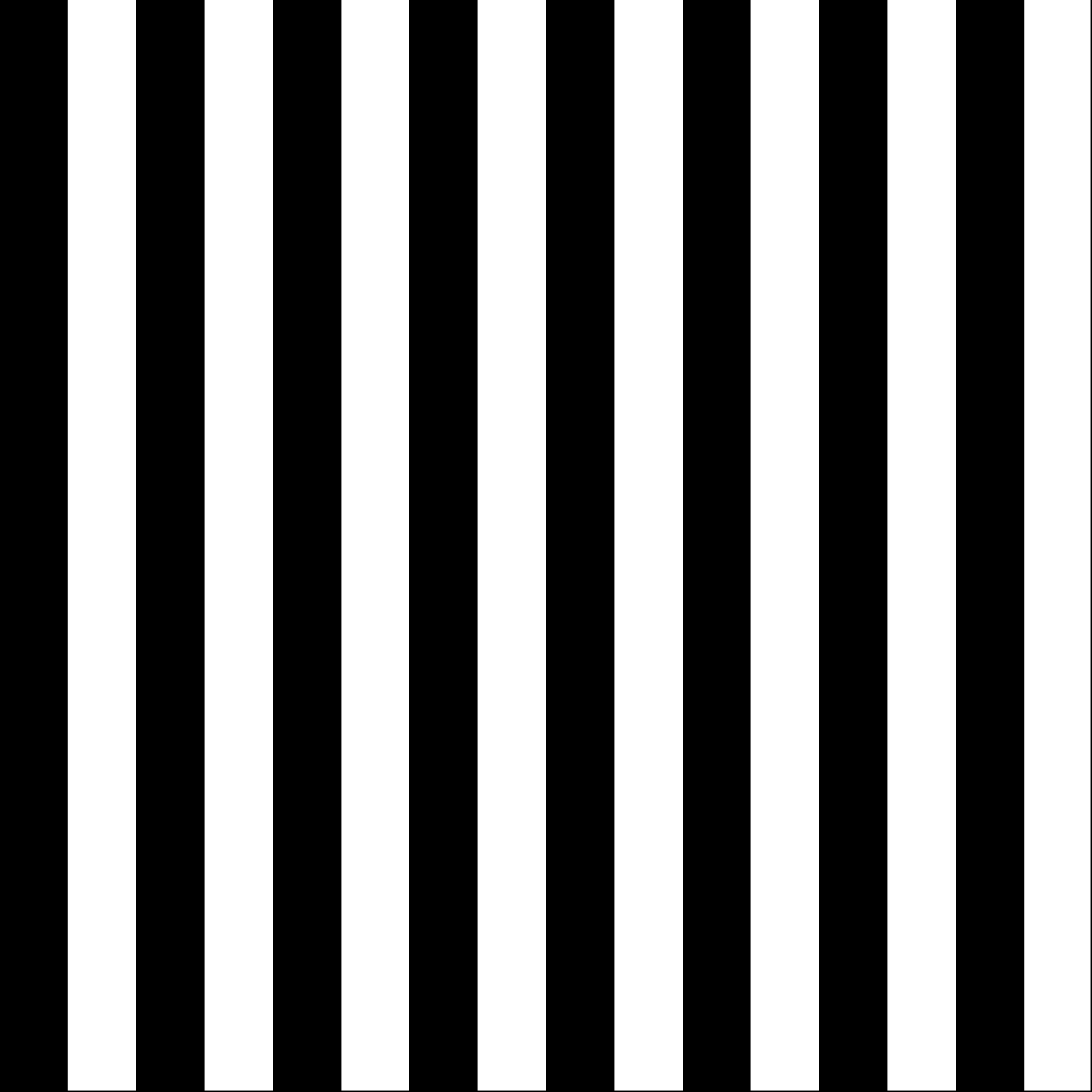} \\
\includegraphics[width=\imgWidthTiny]{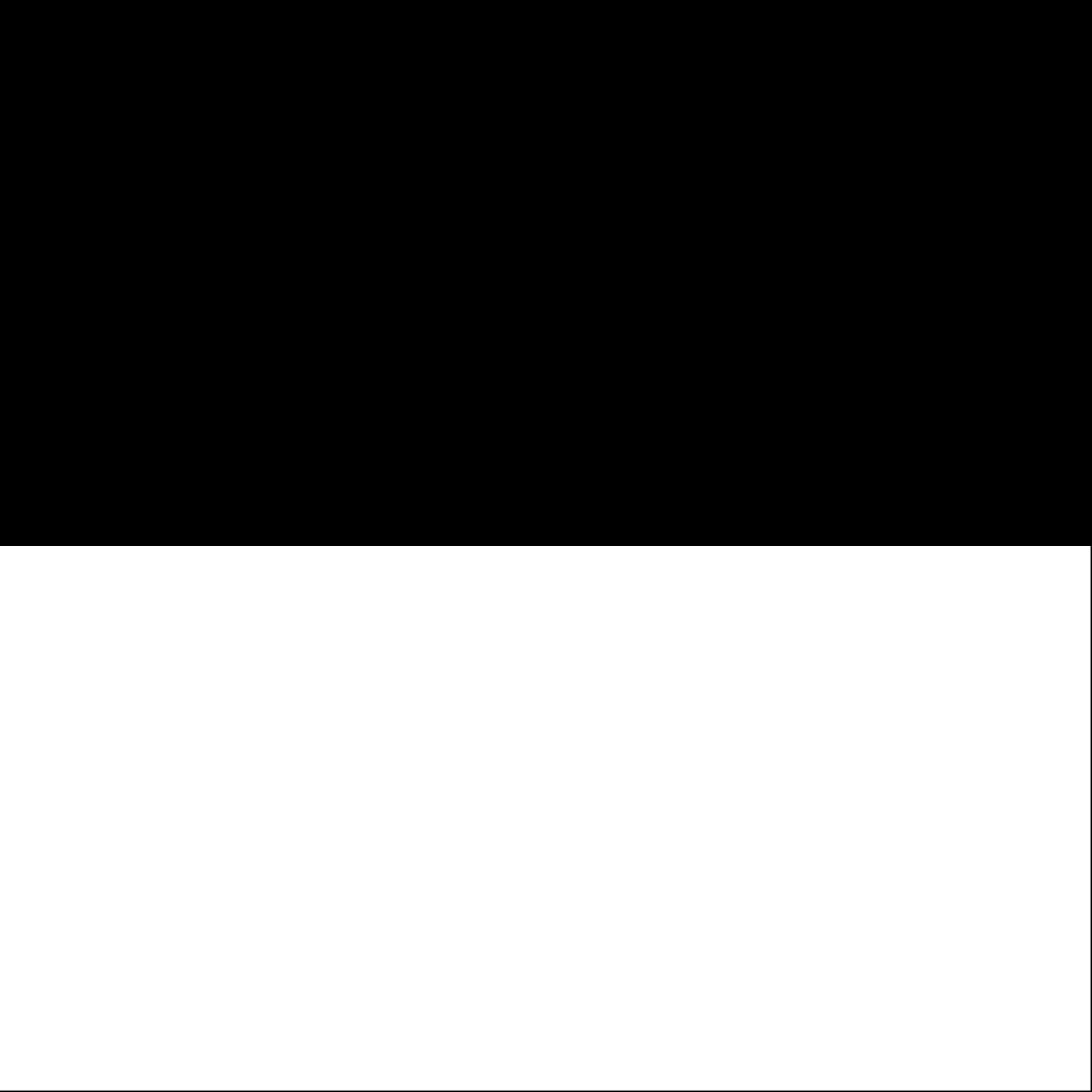} &
\includegraphics[width=\imgWidthTiny]{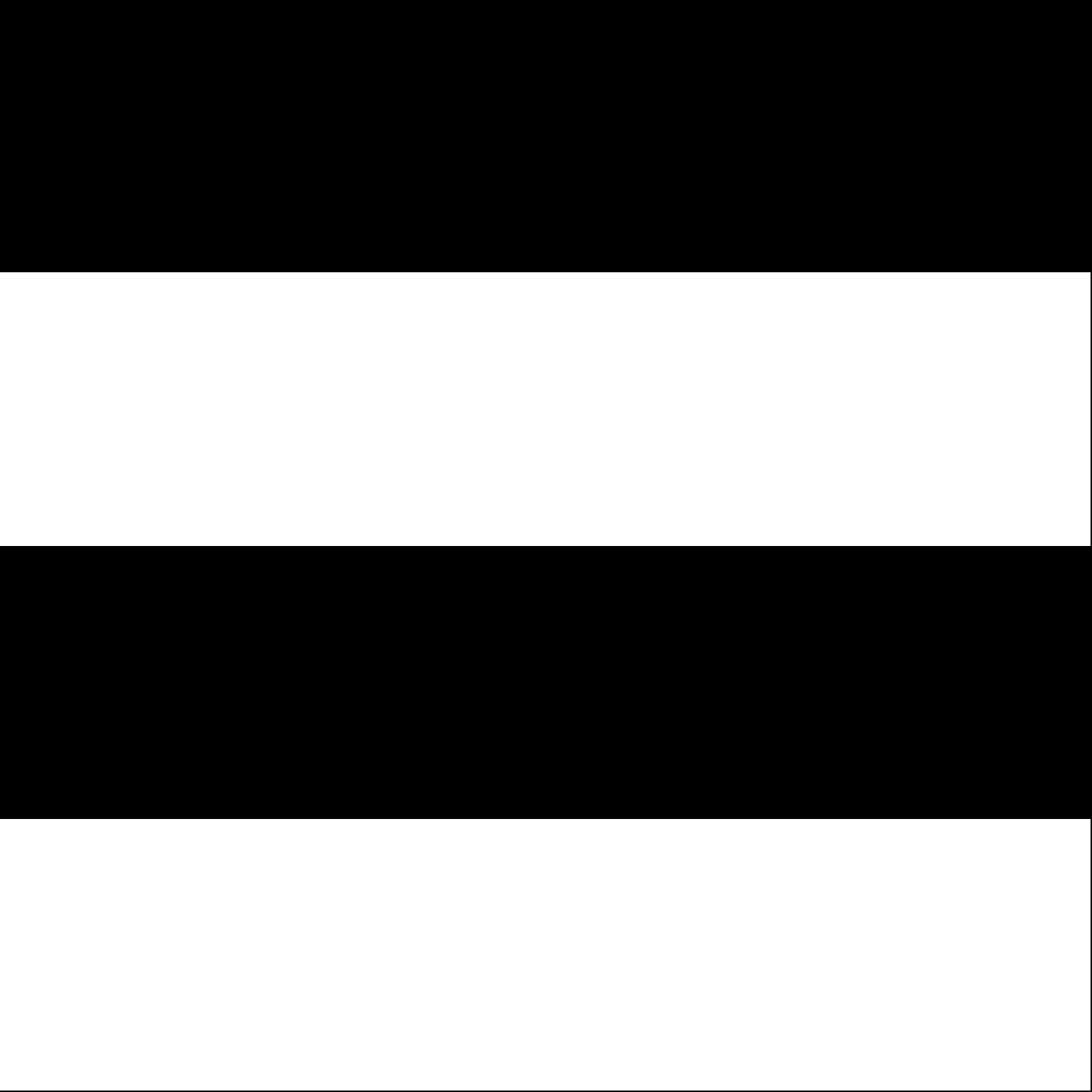} &
\includegraphics[width=\imgWidthTiny]{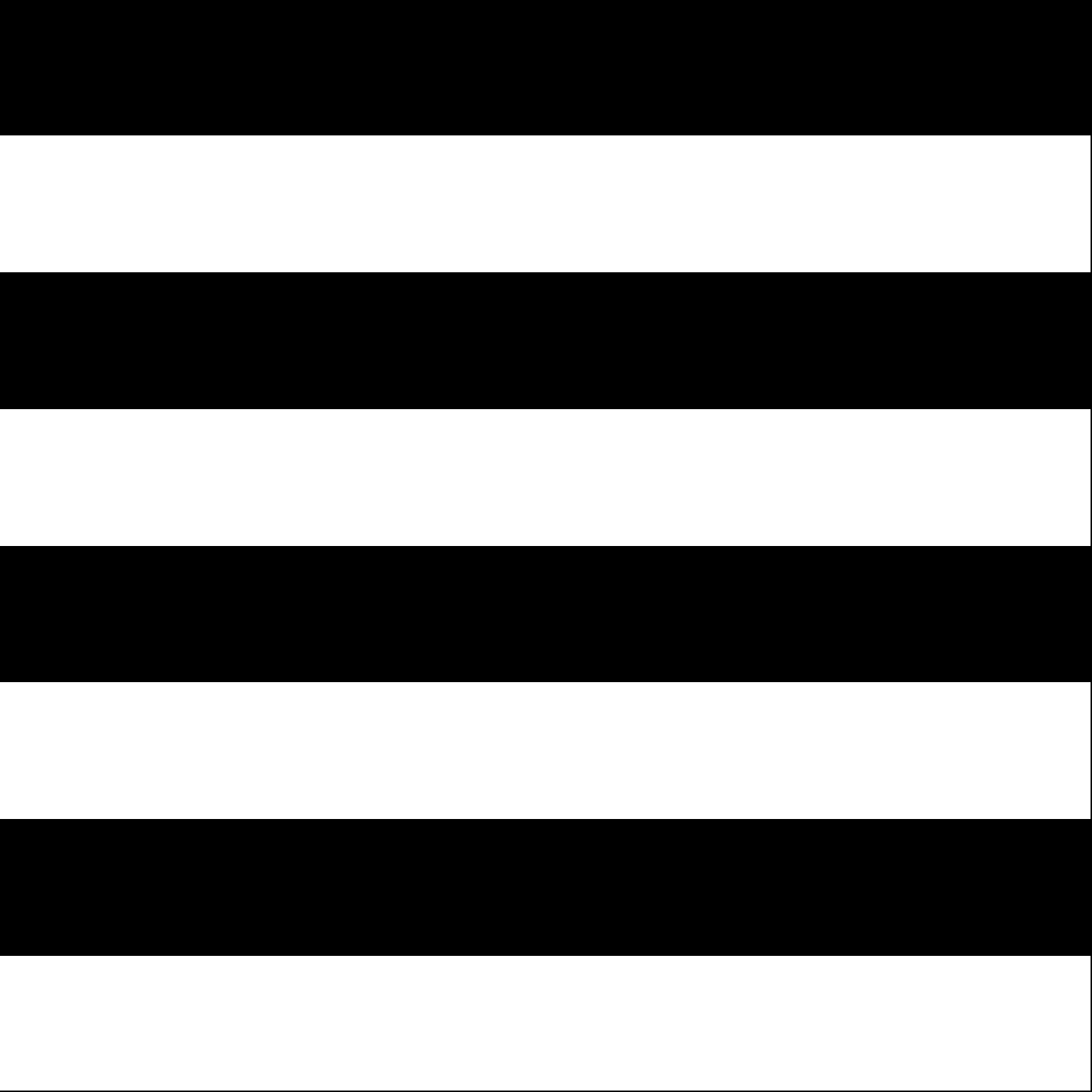} &
\includegraphics[width=\imgWidthTiny]{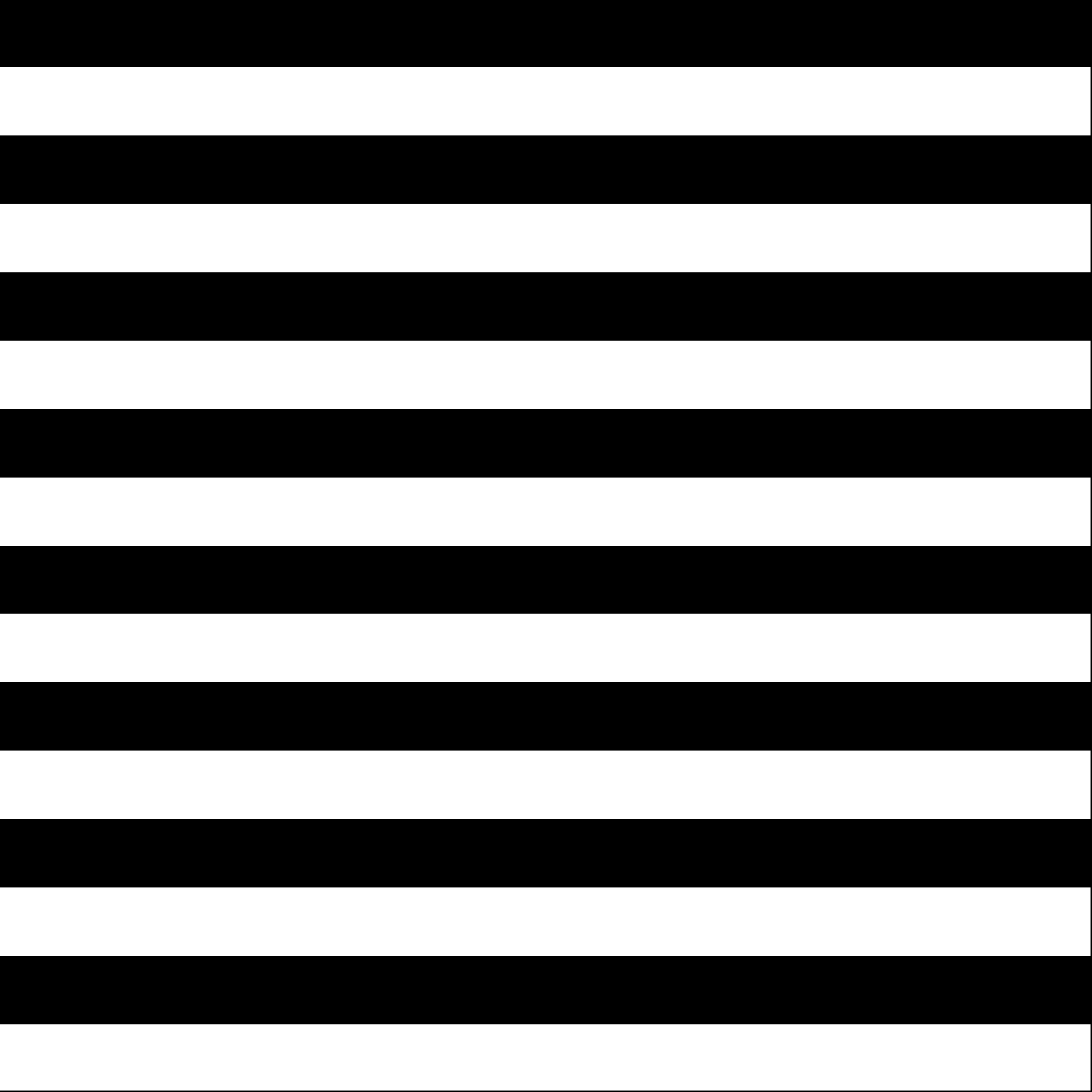}
\end{array}$
\end{center}
\vspace{-0.4cm}
\caption[caption]{The queried regions under the dyadic policy for a  $16\times16$ image shown in white.}
\label{Dyadic}
\end{figure}

\begin{figure}[h]
\vspace{-2.7 mm}
\begin{center}$
\begin{array}{@{\hspace{-0.5em}}c@{\hspace{.5em}}c@{\hspace{.5em}}c@{\hspace{.5em}}c}
\includegraphics[width=\imgWidthTiny]{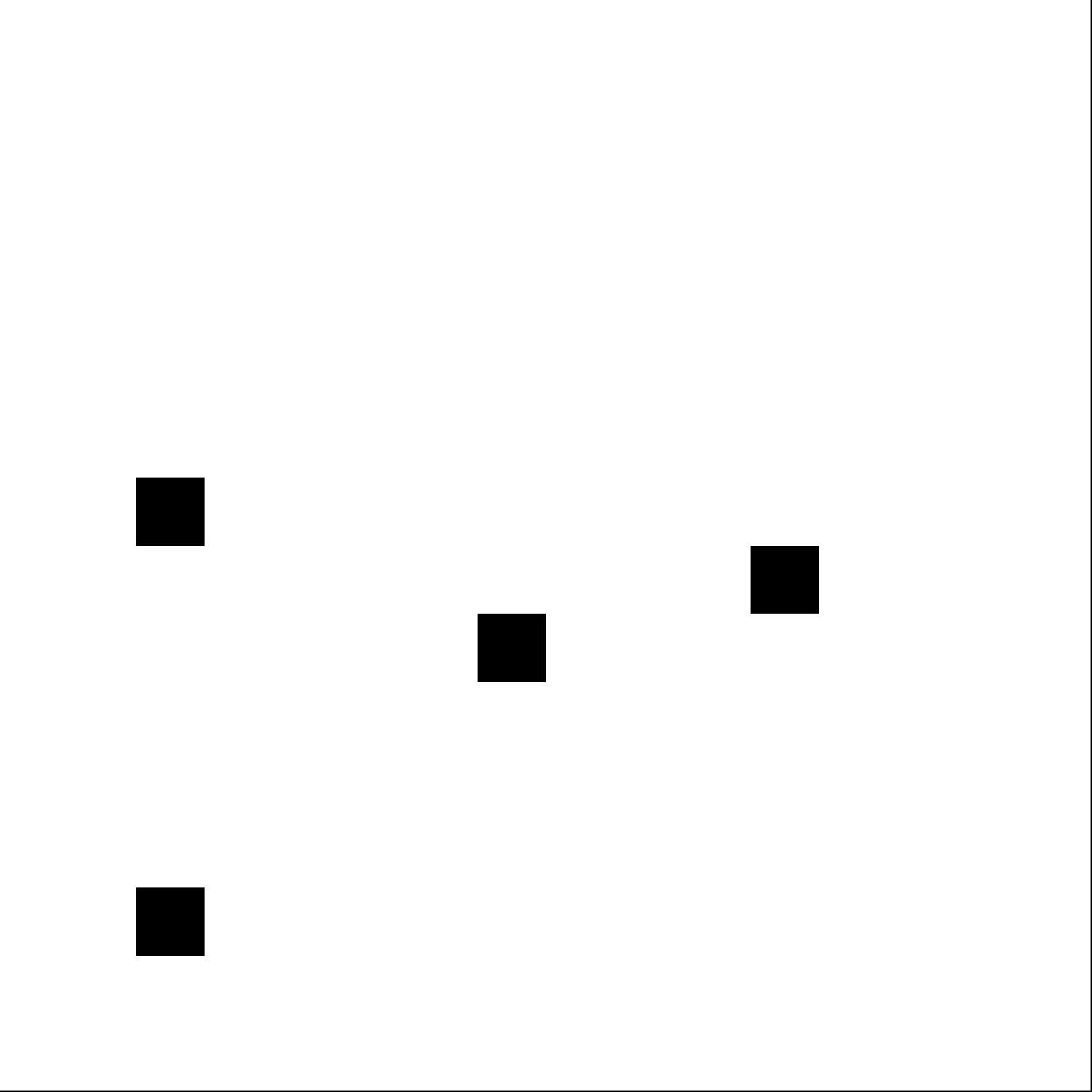} &
\includegraphics[width=\imgWidthTiny]{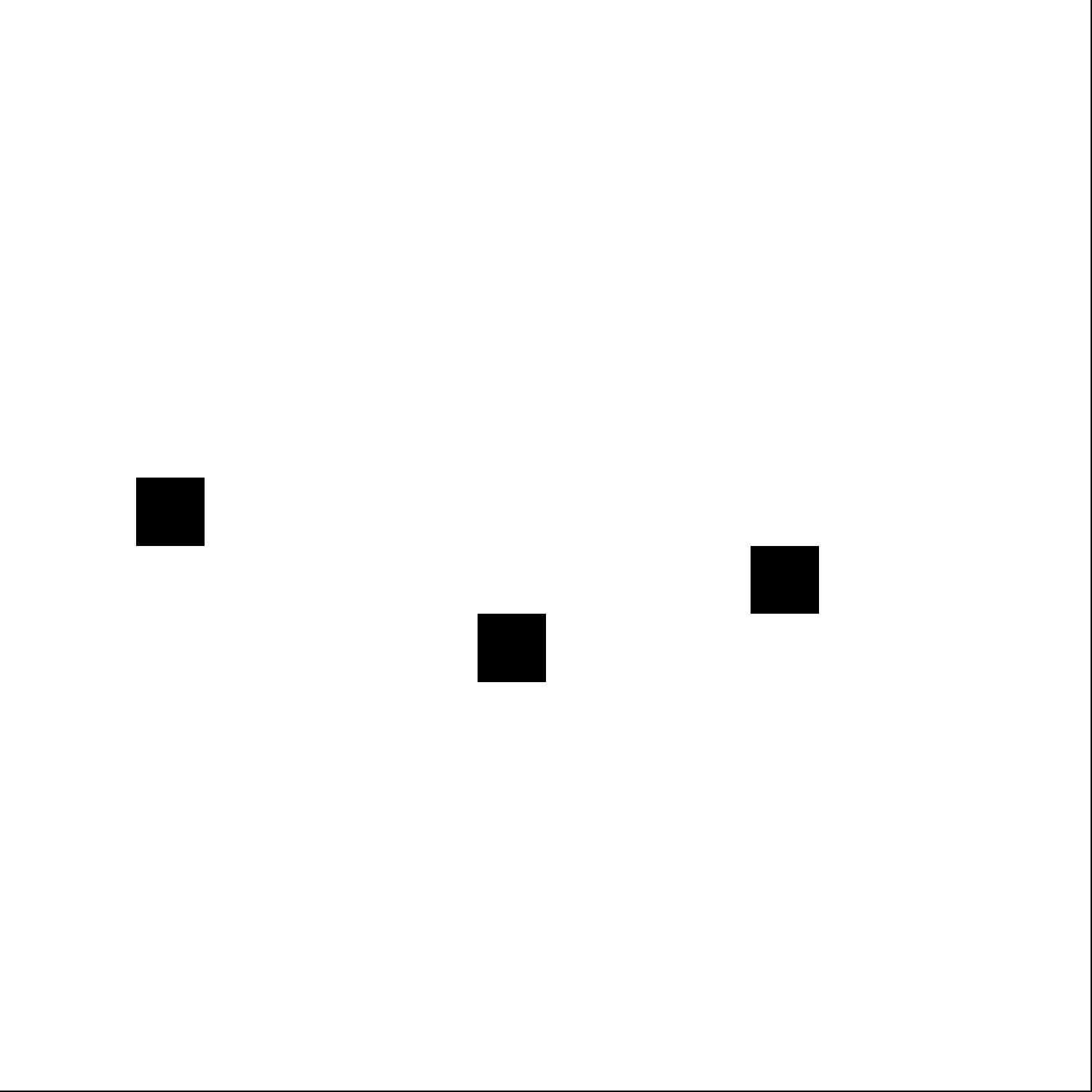} &
\includegraphics[width=\imgWidthTiny]{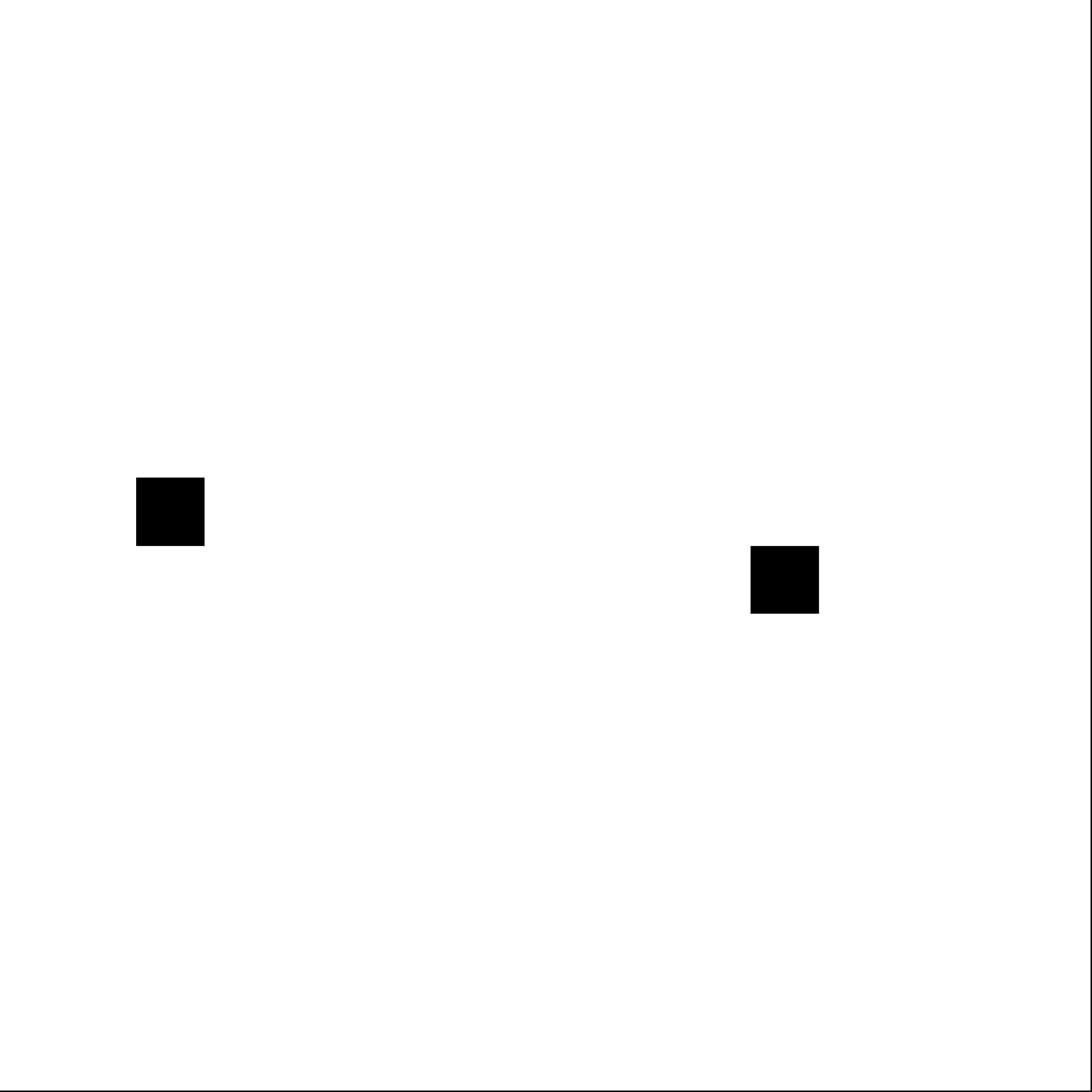} &
\includegraphics[width=\imgWidthTiny]{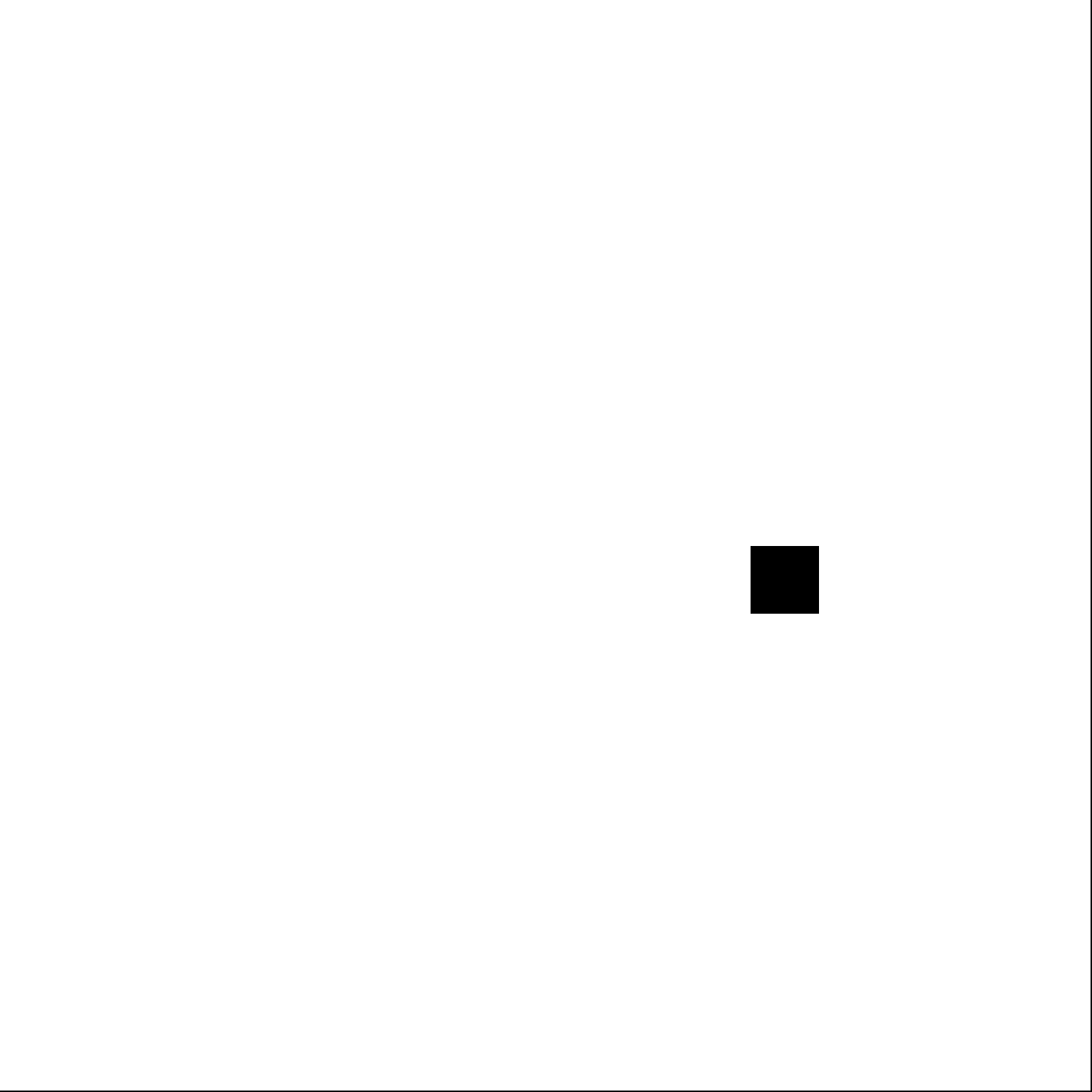} \\
\includegraphics[width=\imgWidthTiny]{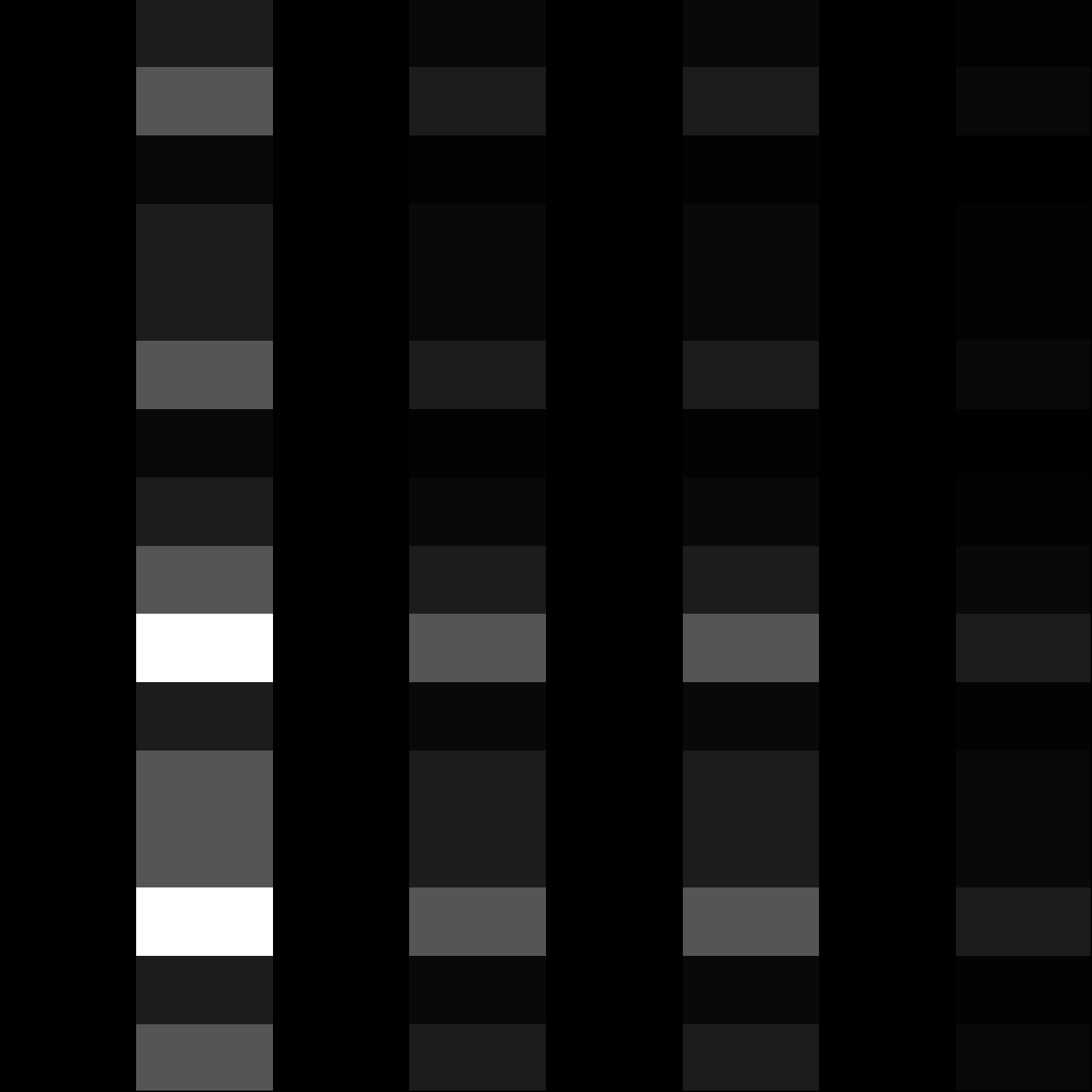} &
\includegraphics[width=\imgWidthTiny]{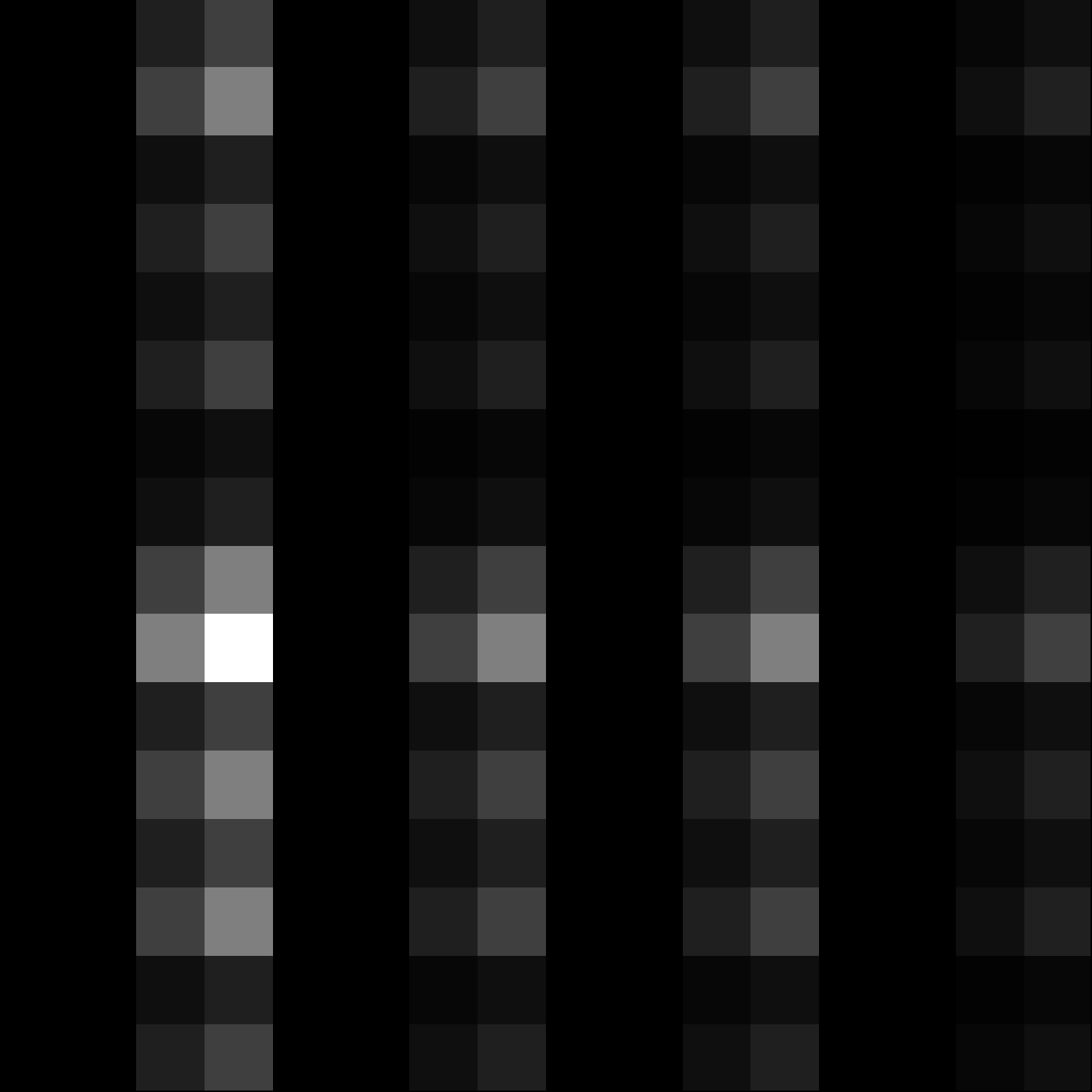} &
\includegraphics[width=\imgWidthTiny]{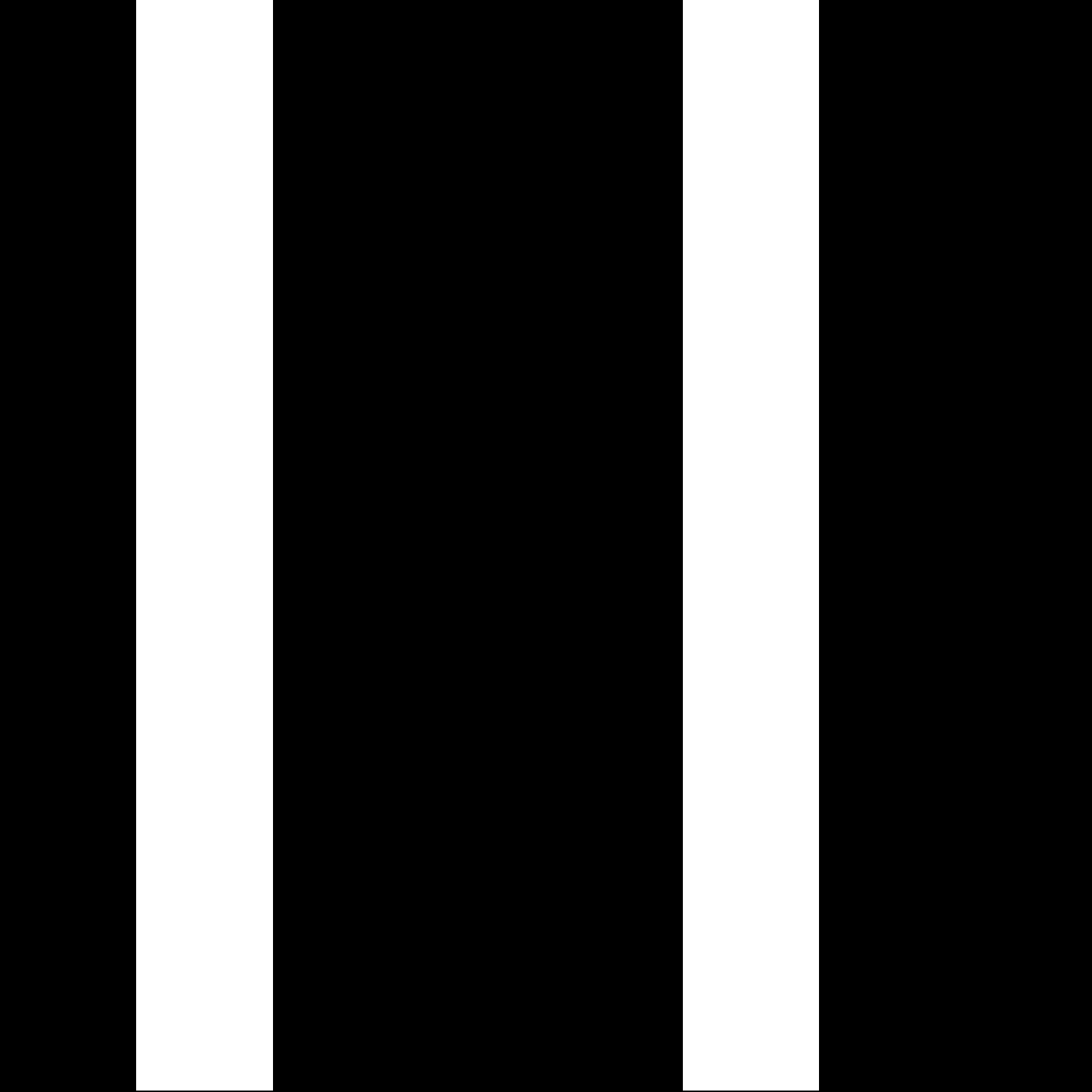} &
\includegraphics[width=\imgWidthTiny]{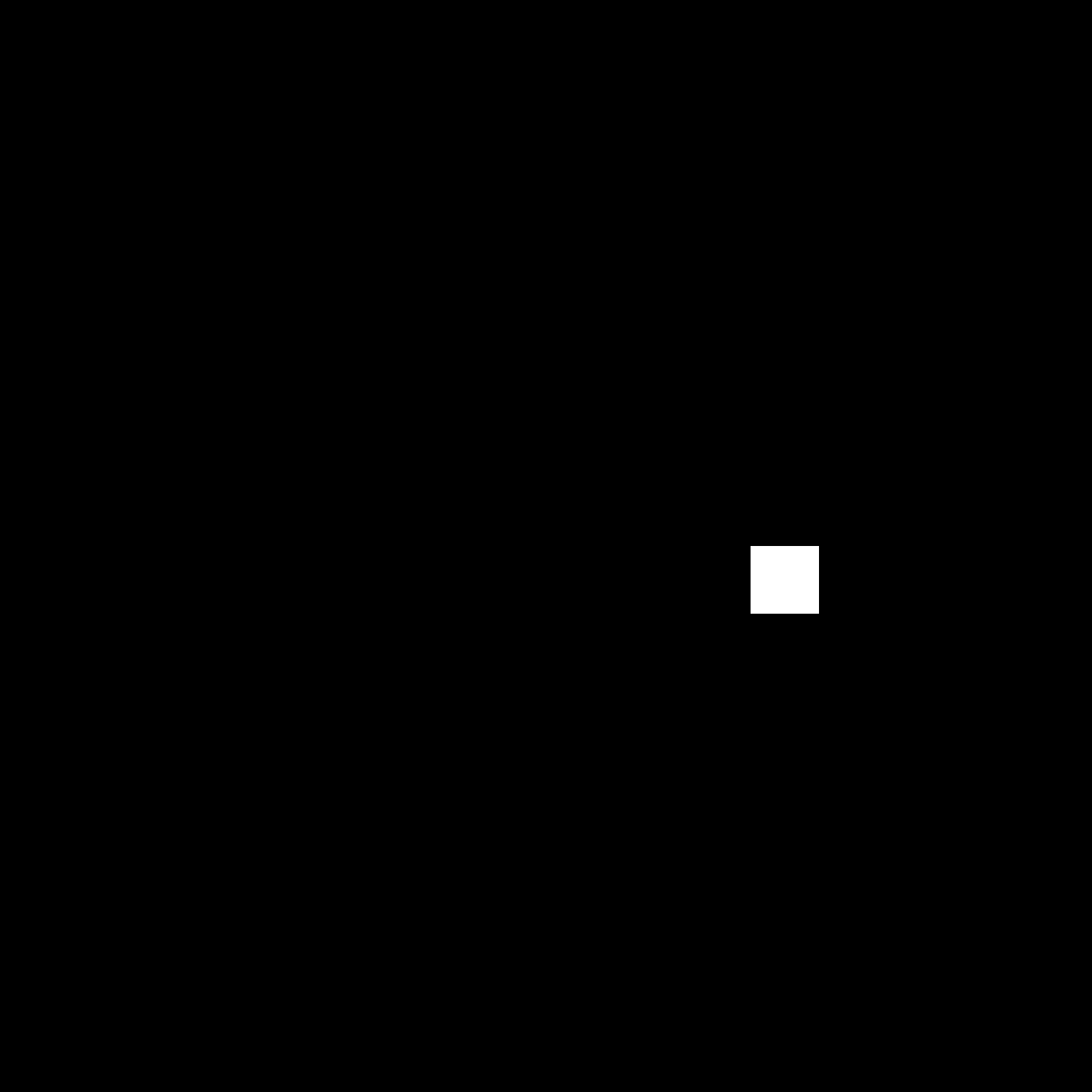}
\end{array}$
\end{center}
\vspace{-0.4cm}
\caption[caption]{(row 1) Example image with $4$ objects initially, one object is found after each iteration of the IPR algorithm. (row 2) The corresponding posterior distribution after each iteration. Light regions indicate pixels more likely to contain the object, while dark regions are less likely.}
\label{Object_Localization}
\end{figure}

The Entropy Pursuit (EP) algorithm is a greedy algorithm aimed at reducing the expected entropy on the joint location of the objects. It has been studied and used for locating and tracking objects in~\cite{SznJed10,JedFraSzn12,GemGem84,SznRicTayJedHag13,SznLucFraJedFua13}.
This algorithm can be related to the IPR algorithm. The differences between EP and IPR are: i) EP uses a different ordering criterion; ii) EP updates the ordering each time after running the oracle at a pixel instead of after an object being found. Specifically, EP computes for each pixel the expected entropy reduction in the distribution of the location of the objects which would be achieved by running the oracle at this pixel. It then selects the pixel for which this quantity is maximal.
The EP algorithm is provided below.
\begin{algorithm}[H]
\small
\caption{Entropy Pursuit (EP) Algorithm}
\label{algo:entropy_pursuit}
\begin{algorithmic}[1]
\STATE Compute the answers to the screening questions.
\STATE Obtain $E_N$ defined in \eqref{eq:collection}, the collection of matrices characterizing possible joint object locations.
\REPEAT
\STATE Select the pixel for which the expected entropy reduction is maximum.
\STATE Run the oracle at this pixel.
\STATE Remove all the inconsistent matrices from the collection $E_N$.
\UNTIL {all the objects are found.}
\end{algorithmic}
\end{algorithm}
\vspace{-0.3cm}

\newcommand{\imgWidthMedium}{0.33\textwidth}
\begin{figure*}[!]
\begin{center}
$\begin{array}{c@{\hspace{.1em}}c@{\hspace{0.1em}}c}
\includegraphics[width=\imgWidthMedium]{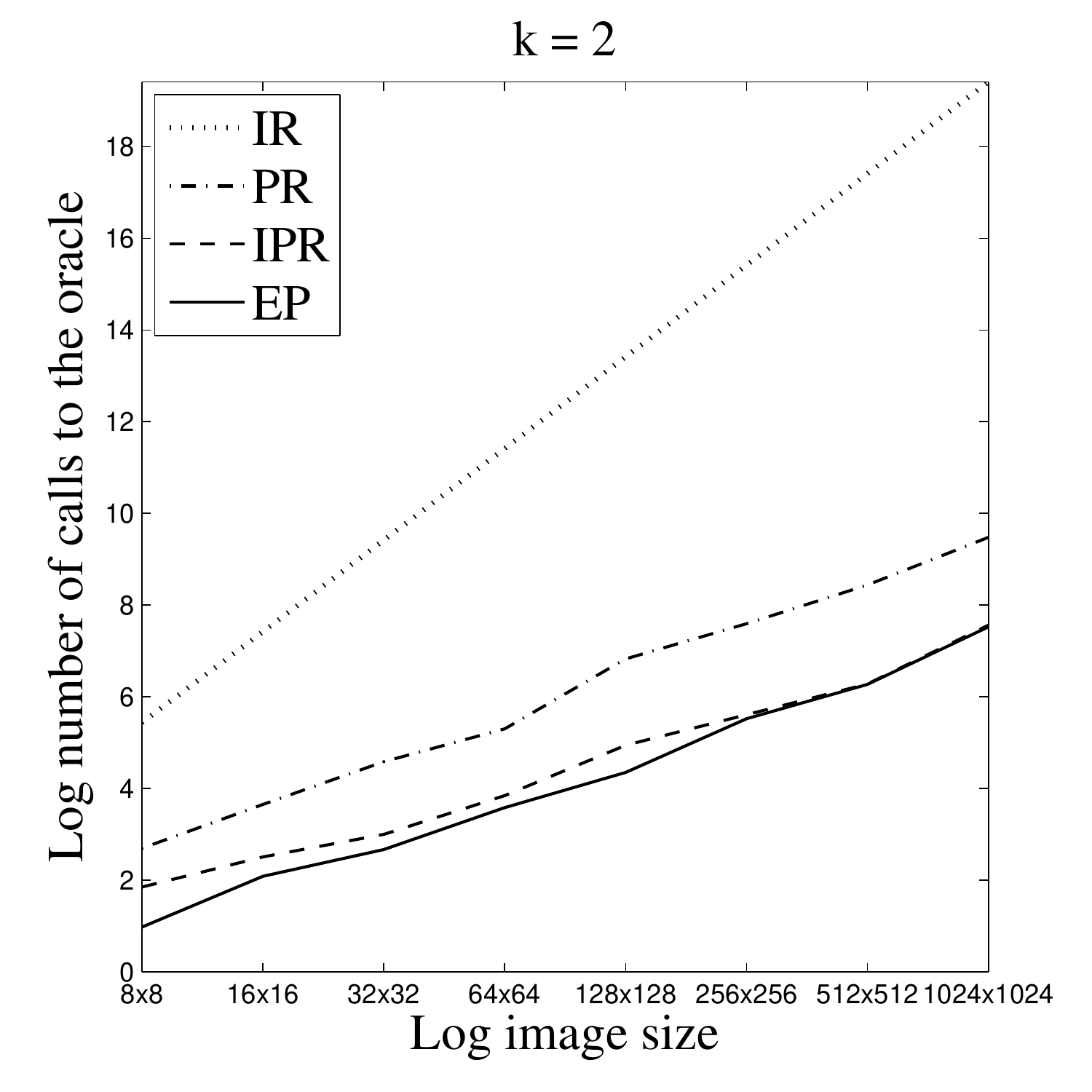} &
\includegraphics[width=\imgWidthMedium]{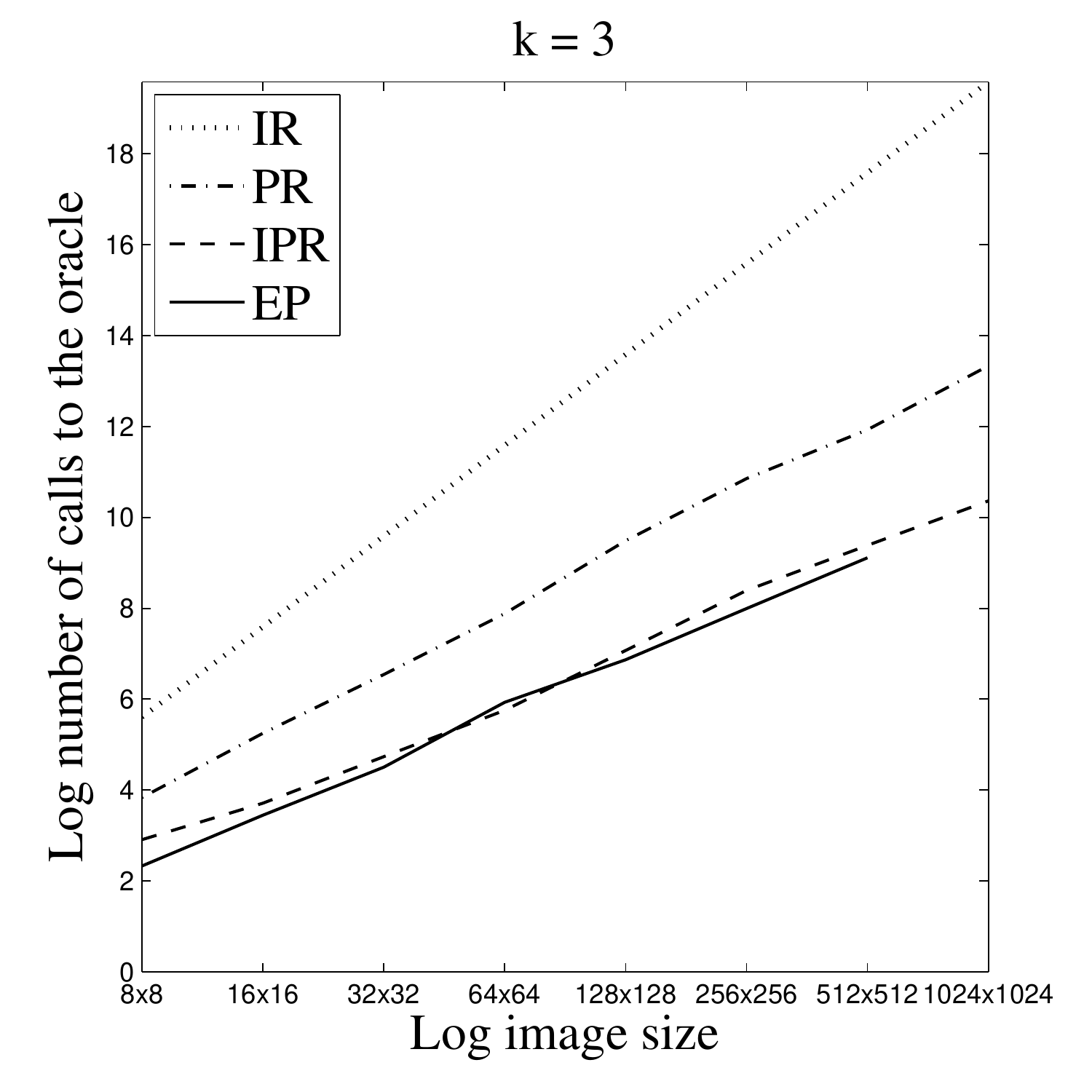} &
\includegraphics[width=\imgWidthMedium]{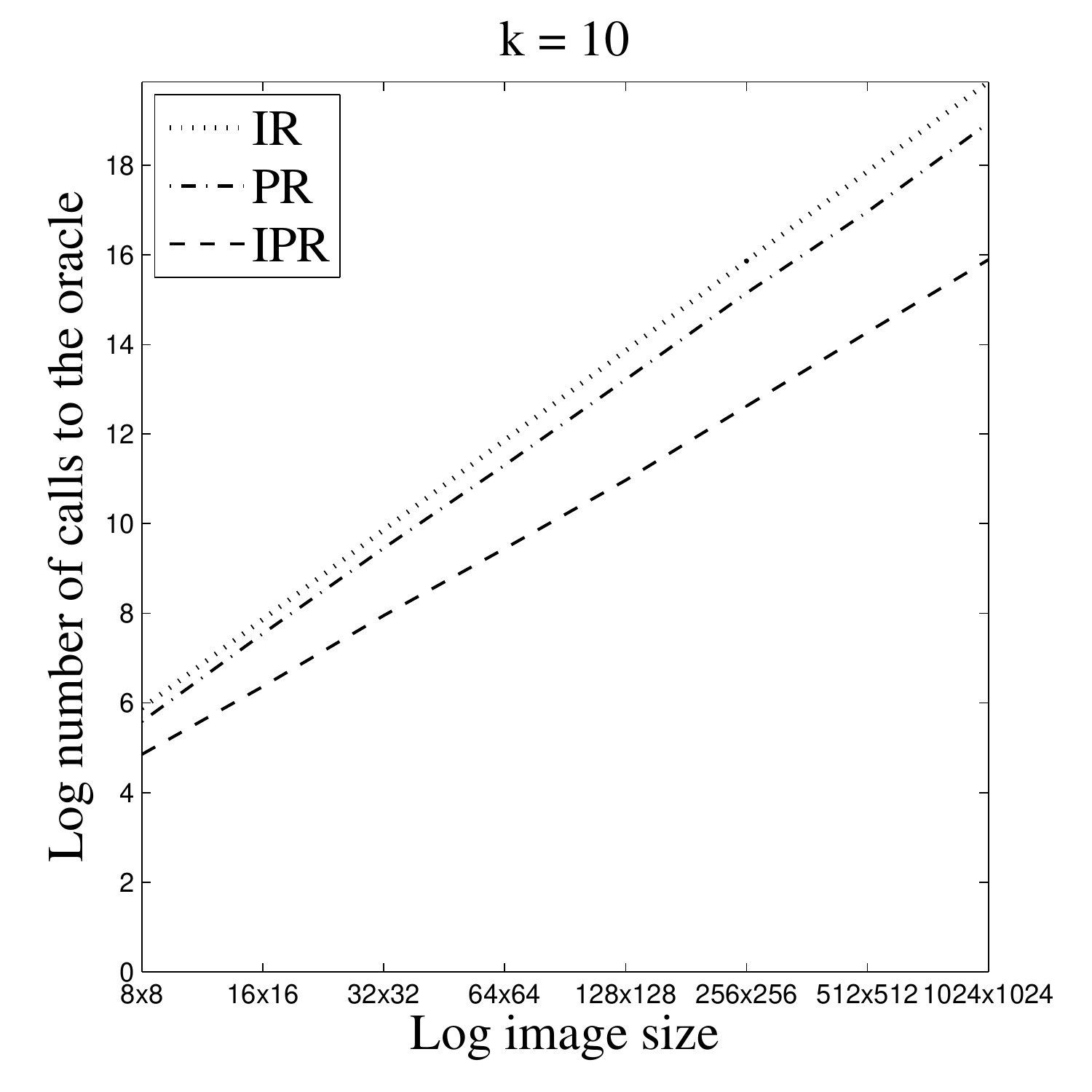} 
\end{array}$
\end{center}
\vspace{-0.7cm}
\caption[caption]{The mean number of calls to the oracle over 100 samples plotted against the image size for $k=2$, $k=3$ and $k=10$ respectively using the algorithms described in section.}
\label{Results_IR_PR_IPR_EP}
\end{figure*}

We use simulations to compare the performances of the three algorithms described above with a baseline algorithm, called Index Rank (IR) .  IR sweeps the image from left to right, top to bottom, until all the objects of the object are found. For the sake of simplicity, the object to be found in our simulation is a dot of size $1$ pixel. We use $100$ random assignments for the locations of the object for each $k$ and each image size in the simulation, and measure the number of calls to the oracle required in each case.


Figure \ref{Results_IR_PR_IPR_EP}, compares the algorithms for $k=2$, $k=3$ and $k=10$ object for image sizes $\{8\times8, 16\times16, \dots, 1024\times1024\}$. Algorithms PR, IPR and EP require a smaller average number of calls to the oracle compared to the baseline IR. An example will show how dramatic this is for large size images. In the case of $1024 \times 1024$ pixel images and $k=2$ objects, IR  requires $2^{20}$ evaluations of the oracle while IPR requires less than $2^8$ on average. IPR is also much more efficient than PR. IPR and EP show similar performances, however, IPR is superior to EP in terms of the computational complexity. Due to the EP algorithm's large computational and memory requirements, we have only plotted EP for $k=2$ and $k=3$, and have only gone up to $512\times512$ image for $k=3$.

\section{Conclusion}
\label{sec:conclusion}

We have considered the problem of twenty questions with noiseless answers, in which we aimed at locating multiple objects simultaneously. There are a variety of applications associated with this problem, such as group testing, computer vision, stochastic simulation and bioinformatics. By adopting the approach of maximizing the rate of reduction in expected entropy of the posterior distribution, we derived an upper bound on the expected entropy and studied two classes of policies, the \emph{dyadic policy} and the \emph{greedy policy}. Although the greedy policy, as we have shown, outperforms the dyadic policy in reducing the expected entropy, the latter employs a series of pre-determined question sets and thus is easy to implement. In addition, the dyadic policy beats traditional policies such as the \emph{sequential bifurcation policy} and is relatively stable in the sense that the average reduction in entropy converges under certain assumptions (Section \ref{sec:dyadic_conv}).

Also, there are several questions calling for future works. First, in real applications, noisy answers provide a more natural and accurate approximation but we only considered noiseless answers in this paper. Second, we assumed the number of the objects is known, but in a more general setting, this assumption should be released. Third, another objective function such as the mean-squared error can replace the expected entropy, which measures the performance of a specific policy differently. We feel that researches in these and other questions will be prosperous and fruitful.

\section*{Acknowledgments}
We would like to thank Li Chen for the fruitful discussions and the preliminary work which eventually led to this manuscript. Bruno Jedynak was partially supported by NSF IIS-0964416 and by the Science of Learning Institute at Johns Hopkins University through the research grant untitled ``Spatial Localization Through Learning: An Information Theoretic Approach". Peter Frazier was supported by NSF CAREER CMMI-1254298, NSF IIS-1247696, AFOSR YIP FA9550-11-1-0083, AFOSR FA9550-12-1-0200, AFOSR FA9550-15-1-0038, and the ACSF AVF.

\section*{Appendix}
\appendix

\section{Lemmas and Proofs}
We first introduce some notation, used here, and throughout the paper.
For any pair of random variables $W,V$, we define $H(W\|V)$ to be the random variable taking the value 
\begin{equation}
-\int_{-\infty}^\infty f(w|V=v)\log f(w|V=v)\,dw
\end{equation} for each $V=v$, assuming the conditional density function $f(w|V=v)$ exists. The ``usual" conditional entropy is related to it by

\begin{equation}
\label{eq:def_condentropy}
H(W|V)=E[H(W\|V)].
\end{equation}

We now provide here, in Lemma~\ref{lem:Eentropy_onestep}, an expression for the expected entropy after additional questions.  This lemma is based on the idea that each additional question reduces the entropy of $\theta_{1:k}$ by an amount that can be expressed in terms of the conditional entropy of the answer to that question.  The total entropy reduction can then be computed as a sum of the contributions from each question, which we use later to study the expected total entropy reduction under specific policies.


\begin{lem}
\label{lem:Eentropy_onestep}
Under any policy $\pi$,
\begin{equation}
\label{eq:Eentropy_onestep}
E[H(p_{n+1})|B_n]=H(p_n)-H(X_{n+1}\|B_n), \text{ for all $n=0,1,\dots,N-1$},
\end{equation}
Moreover,
\begin{equation}
\label{eq:Eentropy_all}
E[H(p_N)]=H_0-\sum_{n=0}^{N-1}H(X_{n+1}|B_n).
\end{equation}
\end{lem}
\begin{proof}
First of all, we prove the recursive relation \eqref{eq:Eentropy_onestep}. $H(p_n)$ is the entropy of the posterior distribution of $\theta$, which is random through its dependence on the past history $B_n$, hence we can rewrite it as $H(p_n)=H(\theta\|B_n)$. Similarly, $H(p_{n+1})=H(\theta\|B_{n+1})=H(\theta\|B_n,A_{n+1},X_{n+1})=H(\theta\|B_n,X_{n+1})$ as $A_{n+1}$ is $B_n$-measurable under any valid policy $\pi$. Since all three terms in \eqref{eq:Eentropy_onestep} are $\sigma(B_n)$-measurable random variables, it suffices to prove \eqref{eq:Eentropy_onestep} holds for any fixed history $B_n=b_n$, i.e.
\begin{equation}
\label{eq:Eentropy_onestep_fixY}
E[H(\theta\|B_n,X_{n+1})|B_n=b_n]=H(\theta|B_n=b_n)-H(X_{n+1}|B_n=b_n).
\end{equation}

Using information theoretic arguments, we have
\begin{subequations}
\begin{align}
E[H(\theta\|B_n,X_{n+1})|B_n=b_n]&=\sum_{x_{n+1}=0}^k H(\theta|B_n=b_n, X_{n+1}=x_{n+1})P(X_{n+1}=x_{n+1}|B_n=b_n)\\
&=H(\theta|X_{n+1},B_n=b_n)\label{eq:cond_entropy}\\
&=H(\theta,X_{n+1}|B_n=b_n)-H(X_{n+1}|B_n=b_n)\label{eq:chainrule1}\\
&=H(X_{n+1}|\theta,B_n=b_n)+H(\theta|B_n=b_n)-H(X_{n+1}|B_n=b_n)\label{eq:chainrule2}\\
&=H(\theta|B_n=b_n)-H(X_{n+1}|B_n=b_n)\label{eq:Y_vanish}
\end{align}
\end{subequations}
where \eqref{eq:cond_entropy} comes from the definition of conditional entropy and \eqref{eq:chainrule1}, \eqref{eq:chainrule2} come from the chain rule for conditional entropy. \eqref{eq:Y_vanish} holds as the first term in \eqref{eq:chainrule2} vanishes because the information of $\theta$ completely determines the answer $X_{n+1}$. This proves \eqref{eq:Eentropy_onestep_fixY}.

Now, in order to prove \eqref{eq:Eentropy_all}, let us first obtain a recursive relation in unconditional expected entropy of posterior distributions. Taking the expectation over $B_n$ on both sides of \eqref{eq:Eentropy_onestep},
\begin{equation}
\label{eq:Eentropy0}
E\left[E[H(p_{n+1})|B_n]\right]=E[H(p_n)]-E\left[H(X_{n+1}\|B_n)\right].
\end{equation}

Note that $E\left[E[H(p_{n+1})|B_n]\right]=E[H(p_{n+1})]$ by the iterated conditioning property of conditional expectation. Moreover, $E\left[H(X_{n+1}\|B_n)\right]=H(X_{n+1}|B_n)$ according to the definition of conditional entropy in \eqref{eq:def_condentropy}. Hence, \eqref{eq:Eentropy0} is equivalent to
\begin{equation}
\label{eq:Eentropy1}
E[H(p_{n+1})]=E[H(p_n)]-H(X_{n+1}|B_n).
\end{equation}

Applying \eqref{eq:Eentropy1} iteratively for $n=N-1,\dots,0$, we obtain \eqref{eq:Eentropy_all}, which concludes the proof.
\end{proof}

Note that the dyadic policy is deterministic, i.e., it does not make use of the random seed $Z$. As a consequence, in the following, we use $X_{1:n}$ to denote the history up to time $n$ without including $Z$ and $A_{1:n}$.

\begin{lem}
\label{lem:dyadic_Hp}
Under the dyadic policy, for all $n=1,2,\dots,N$,
\begin{equation}
\label{eq:Hn}
H(p_n) =  - \sum_{j=1}^n Z_j+I_2(n),
\end{equation}
where $I_2(n)$ is a random variable and $Z_j=k-\log{k \choose X_j}$ with $X_j$ following i.i.d binomial distribution $\Bin(k, \frac{1}{2})$.
\end{lem}

\begin{proof}
Let $X_{1:n}=x_{1:n}$ be fixed. According to Lemma \ref{lem:productk},
\begin{equation}
\label{eq:post}
p_n(u_{1:k}) = \frac{p_0(u_{1:k})}{p_0\left(\bigcup\limits_{\mathcal S \in E_n} C_{\mathcal S}\right)} =\frac{f_0(u_1)\dots f_0(u_k)}{\sum\limits_{\mathcal S\in E_n} f_0(C_{s^{(1)}})\dots f_0(C_{s^{(k)}})},
\end{equation}
where $(u_{1:k}) \in C := \bigcup\limits_{\mathcal S \in E_n} C_{\mathcal S}.$

Under the dyadic policy, the support of $f_0$ is partitioned into $2^n$ subsets with identical probability masses after the final step and each $C_{s^{(i)}}$ is one such subset, for $i = 1,2,\dots, k$. Thus, we have
\begin{equation}
\label{eq:marginal}
f_0(C_{s^{(i)}}) = 2^{-n}, \text{ for $i=1,2,\dots,k$ and $\mathcal S\in E_n$}.
\end{equation}

Let $|E_n|$ be the cardinality of $E_n$. Note that under the dyadic policy, every binary sequence $s$ of length $N$ corresponds to a nonempty set $C_s$. Furthermore, in step $j$, there are ${k \choose x_j}$ ways to choose the $j^{th}$ row in the matrix satisfying the definition in \eqref{eq:collection}, for $j=1, 2, \dots, n$. Thus, by the product rule,
\begin{equation}
\label{eq:En}
|E_n|=\prod_{j=1}^n {k \choose x_j}.
\end{equation}
By \eqref{eq:marginal} and \eqref{eq:En},
\begin{equation}
\label{eq:sum}
p_0(C)=\sum\limits_{\mathcal S\in E_n} f_0(C_{s^{(1)}})\dots f_0(C_{s^{(k)}}) = 2^{-nk}\prod_{j=1}^n {k \choose x_j}.
\end{equation}

Combining the result above and the definition of the differential entropy, we have
\begin{equation}
\label{eq:Hp}
\begin{split}
H(p_n) = & -\int\limits_{C} p_n(u_{1:k}) \log(p_n(u_{1:k}))\, du_{1:k}\\
=& -\int\limits_{C} \frac{p_0(u_{1:k})}{p_0(C)} \log\left(\frac{p_0(u_{1:k})}{p_0(C)}\right)\, du_{1:k}\\
=& \left[\frac{\log \left(p_0(C)\right)}{p_0(C)} \int\limits_{C} p_0(u_{1:k})\, du_{1:k}\right] + \left[- \frac{1}{p_0(C)}\int\limits_{C} p_0(u_{1:k})\log \left(p_0(u_{1:k})\right)\, du_{1:k}\right]\\
=& I_1(n)+I_2(n),
\end{split}
\end{equation}
where $I_1(n)$ and $I_2(n)$ denote the first term and the second term in the last equation above. $I_1(n)$ can be easily computed as
\begin{equation}
\label{eq:I1}
I_1(n) =\frac{\log \left(p_0(C)\right)}{p_0(C)} \int\limits_{C} p_0(u_{1:k})\, du_{1:k}=\log \left(p_0(C)\right)= -\left(nk - \sum\limits_{j=1}^n \log{k \choose x_j}\right).
\end{equation}

Now consider $X_{1:n}$ as random variables. By Theorem \ref{thm:postY}, we see that under the dyadic policy, $X_{1:n}$ is a sequence of i.i.d. random variables $\Bin\left(k,\frac{1}{2}\right)$. Moreover, $I_2(n)$ is random through its dependence on the random support $C$. Therefore, combining \eqref{eq:Hp} and \eqref{eq:I1}, we prove the claim in Lemma \ref{lem:dyadic_Hp} by setting $Z_j=k-\log {k \choose X_j}$.
\end{proof}

Define $I_2(0)=H(p_0)=H_0$ so that \eqref{eq:Hn} is also satisfied for $n=0$. Applying the result above, we can furthermore analyze the term $I_2(n)$ and derive the following lemma.
\begin{lem}
\label{lem:martingale_conv}
Assume there exists $M>0$ such that $f_0(u)\leq M$ for all $u\in \mathbb R$. Then the random variable $I_2(n)$ in \eqref{eq:Hn} converges to a random variable $I_2(\infty)$ almost surely as $n\rightarrow \infty$, where $I_2(\infty)$ is a random variable and $E[|I_2(\infty)|]<\infty$.
\end{lem}

\begin{proof}
We prove almost sure convergence using the martingale convergence theorem (see Theorem 35.5 in \cite{Billingsley}). First, let us calculate the expected value of $Z_j$ as follows.
\begin{equation}
\begin{split}
E(Z_j) &= \sum_{j=0}^k \left(k-\log{k \choose j}\right) {k \choose j} 2^{-k}.
\end{split}
\end{equation}

Therefore, $E(Z_j) = H\left(\Bin\left(k,\frac{1}{2}\right)\right)$ since
\begin{equation}
H\left(\Bin\left(k,\frac{1}{2}\right)\right) = -\sum_{j=0}^k {k \choose j} 2^{-k} \log\left({k \choose j} 2^{-k}\right)=\sum_{j=0}^k \left(k - \log{k \choose j}\right) {k \choose j} 2^{-k}.
\end{equation}

Now, let us verify that $I_2(n)$ is a martingale. According to \eqref{eq:Hn},
\begin{subequations}
\begin{align}
E[I_2(n+1)|X_{1:n}]&=E\left[H(p_{n+1})+\sum_{j=1}^{n+1} Z_j\Bigg| X_{1:n}\right]\\
&=H(p_n)-H(X_{n+1}\|X_{1:n})+\sum_{j=1}^n Z_j+E [Z_{n+1}| X_{1:n}]\label{eq:martingale_1}\\
&=I_2(n)-H(X_{n+1}\|X_{1:n})+E[Z_{n+1}| X_{1:n}]\\
&=I_2(n)-H\left(\Bin\left(k,\frac{1}{2}\right)\right)+E[Z_{n+1}]\label{eq:martingale_2}\\
&=I_2(n)\label{eq:martingale_3},
\end{align}
\end{subequations}
where \eqref{eq:martingale_1} is true by \eqref{eq:Eentropy_onestep} in Lemma \ref{lem:Eentropy_onestep} and the fact that $Z_{1:n}$ is $\sigma(X_{1:n})$-measurable. \eqref{eq:martingale_2} holds because we have proved under the dyadic policy, $X_{n+1}|X_{1:n}\sim \Bin\left(k,\frac{1}{2}\right)$, which is independent of $X_{1:n}$, and $Z_{n+1}$ is also independent of $X_{1:n}$. \eqref{eq:martingale_3} holds because we have proved $E[Z_{n+1}]=H\left(\Bin\left(k,\frac{1}{2}\right)\right)$.

Next, we want to show that $E[|I_2(n)|]<\infty$. Let us fix $X_{1:n}=x_{1:n}$ and expand $I_2(n)$ in as
\begin{equation}
\label{eq:I2}
\begin{split}
I_2(n) &=- \frac{1}{p_0(C)}\sum\limits_{\mathcal S \in E_n} \int\limits_{C_\mathcal S} f_0(u_1)\dots f_0(u_k)\log \left(f_0(u_1)\dots f_0(u_k)\right)\, du_{1:k}\\
&=- \frac{1}{p_0(C)}\sum\limits_{\mathcal S \in E_n} \sum\limits_{i=1}^k \left(\int\limits_{C_{s^{(i)}}} f_0(u_i)\log (f_0(u_i)) \,du_i \prod_{j\neq i}^k  \int\limits_{C_{s^{(j)}}}f_0(u_k)\, du_j\right)\\
&=- \frac{1}{p_0(C)}\sum\limits_{\mathcal S \in E_n} \sum\limits_{i=1}^k 2^{-n(k-1)} \int\limits_{C_{s^{(i)}}} f_0(u_i)\log (f_0(u_i)) \,du_i.
\end{split}
\end{equation}

Now consider the integral $\int_{C_{s^{(i)}}} f_0(u_i)\log (f_0(u_i)) \,du_i$. Since $f_0(u_i)\le M$, we can obtain an upper bound for $\int_{C_{s^{(i)}}} f_0(u_i)\log (f_0(u_i)) \,du_i$ as
\begin{equation}
\label{eq:I2_upper}
\int_{C_{s^{(i)}}} f_0(u_i)\log (f_0(u_i)) \,du_i \le\log M \int_{C_{s^{(i)}}} f_0(u_i) \,du_i=2^{-n} \log M.
\end{equation}

Substituting \eqref{eq:sum} and \eqref{eq:I2_upper} into \eqref{eq:I2}, we have
\begin{equation}
\label{eq:C}
I_2(n)\ge -k \log M.
\end{equation}

Furthermore, define $I_2^+(n)=\max(I_2(n),0), I_2^-(n)=\max(-I_2(n),0)$ and we have
\begin{equation}
E[|I_2(n)|]=E[I_2^+(n)]+E[I_2^-(n)]=E[I_2(n)]+2E[I_2^-(n)]\le H_0+2k\log M,
\end{equation}
where the last equation follows from the fact that $E[I_2(n)]=I_2(0)=H_0$ since $I_2$ is a martingale and $I_2^-(n)\le k\log M$ by \eqref{eq:C}. Therefore, using the martingale convergence theorem, $I_2(n)$ converges to a random variable $I_2(\infty)$ almost surely with $E[|I_2(\infty)|]\le H_0+2k\log M$.

\end{proof}

From the proof above we can see that if $f_0$ is uniform over $(0,1]$, $f_0(u_i)=1$ for all $u_i\in (0,1]$ and thus the term $I_2$ is $0$. Therefore, in this case, $H(p_n)=- \left(nk - \sum_{j=1}^n \log{k \choose X_j}\right)$.

\section{Definition of the Sequential Bifurcation Policy}
\label{sec:bifurcation}

In this appendix, we define the sequential bifurcation policy used as a benchmark in Figure~\ref{fig:numQs}.
This policy is based on the sequential bifurcation policy of
\cite{BettonvilKleijnen1997}, but adapted slightly to the setting considered in this paper.

We define the sequential bifuration (SB) policy as follows.
At each point in time n, SB maintains a disjoint collection of intervals
$\mathcal{D}_n = \{ D_{n,1}, ...., D_{n,m_n} \}$.
At time $0$, $\mathcal{D}_0 = \{ \mathbb R\}$, and for each $n$, SB obtains $\mathcal{D}_{n+1}$ and $A_{n+1}$ recursively as follows.  First, SB chooses the interval $D^*_{n}$ in $\mathcal{D}_n$ with the largest mass under the prior, i.e.,
\begin{equation}
D^*_{n} \in \argmax_{D \in \mathcal{D}_n} \int_D f_0(u)\, du.
\end{equation}
Then, SB obtains $A_{n+1}$ by splitting $D^*_n$ at its conditional median under the posterior, and taking the left-hand portion.  SB then creates $\mathcal{D}_{n+1}$ by adding to $\mathcal{D}_{n}\setminus D^*_n$ those intervals $A_{n+1}$ and $D^*_n \setminus A_{n+1}$ shown by $X_{n+1}$ to have at least one object.

This version of the sequential bifurcation policy differs slightly from the policy presented in \cite{BettonvilKleijnen1997} in that
(1) it is designed for the continuum rather for a discrete domain;
(2) it is designed for the case with known $k$, while running it for unknown $k$ (as does \cite{BettonvilKleijnen1997}) would require an additional query of the number of objects in $\mathbb R$ at the start;
(3) it is generalized for the case of a non-uniform prior distribution.

\bibliography{two-targets,screening,20Qbib}
\bibliographystyle{IEEEtran}

\end{document}